\documentclass[11pt,journal,onecolumn]{IEEEtran}
\usepackage{amsmath}
\usepackage{balance}
\usepackage{mathtools}
\usepackage{amsfonts}
\usepackage{amssymb,microtype}
\usepackage{relsize}
\usepackage{multirow}
\usepackage{booktabs}
\usepackage{floatrow}
\newfloatcommand{capbtabbox}{table}[][\FBwidth]

\newcommand{\comment}[1]{}

\usepackage{amsthm}
\usepackage{graphicx}
\usepackage{psfrag}
\usepackage{enumerate}
\usepackage{tikz}
\usepackage{balance}
\usepackage{color}
\usepackage{picmacros}

\usepackage{cite}

\def\archive{1}
\title{Private Index Coding}
\author{
	\IEEEauthorblockN{Varun~Narayanan\IEEEauthorrefmark{1}, Jithin~Ravi\IEEEauthorrefmark{2}, Vivek~K.~Mishra\IEEEauthorrefmark{3}, Bikash~Kumar~Dey\IEEEauthorrefmark{4}, Nikhil~Karamchandani\IEEEauthorrefmark{4}, Vinod~M.~Prabhakaran\IEEEauthorrefmark{1}}\\
	\IEEEauthorblockA{\IEEEauthorrefmark{1}Tata Institute of Fundamental Research, Mumbai, }
	\IEEEauthorblockA{\IEEEauthorrefmark{2}Universidad Carlos III de Madrid, Legan\'es, Spain,\\}
	\IEEEauthorblockA{\IEEEauthorrefmark{3}Qualcomm, India,}
	\IEEEauthorblockA{\IEEEauthorrefmark{4}Indian Institute of Technology Bombay, Mumbai\\}
	Emails: varun.narayanan@tifr.res.in,  rjithin@tsc.uc3m.es, vivemish@qti.qualcomm.com, \{bikash, nikhilk\}@ee.iitb.ac.in, vinodmp@tifr.res.in }

\begin{document}
	\maketitle
	\begin{abstract}
		We study the fundamental problem of index coding under an additional privacy constraint that requires each receiver to learn nothing more about the collection of messages beyond its demanded messages from the server and what  is available to it as side information. To enable such private communication, we allow the use of a collection of independent secret keys, each of which
		is shared amongst a subset of users and is known to the server. The goal is to study properties of the key access
		structures which make the problem feasible and then design  encoding and decoding schemes efficient in the size of the server transmission as well as the sizes of the secret keys. We call this the \textit{private
			index coding} problem.
		
		We begin by characterizing the key access structures that make private index coding feasible. We also give conditions to check if a given linear scheme is a valid private index code. For up to three users, we characterize the rate region of feasible server transmission and key rates, and show that all feasible rates can be achieved using scalar linear coding and time sharing; we also show that scalar linear codes are sub-optimal for four receivers. The outer bounds used in the case of three users are extended to arbitrary number of users and seen as a generalized version of the well-known polymatroidal bounds for the standard non-private index coding. We also show that the presence of common randomness and private randomness does not change the rate region. Furthermore, we study the case where no keys are shared among the users and provide some necessary and sufficient conditions for feasibility in this setting under a weaker notion of privacy. If the server has the ability to multicast to any subset of users,  we demonstrate how this flexibility can be used to provide privacy and characterize the minimum number of server multicasts required. 
	\end{abstract}

	\section{Introduction}
	
Index coding~\cite{BirkK98, YossefBJK11, MalekiCJ14,AlonLSWH08,ArbabjolfaeiK18} is a fundamental problem in network information
theory which studies a setup consisting of a server with $N$ messages $x_1, x_2, \ldots, x_N$ communicating with $N$ users over a noiseless broadcast channel. Each user $i$ is assumed to have prior knowledge of a subset of the $N$ messages, referred to as its side information, and wants to obtain $x_i$. The goal of the index coding problem is to minimize the server transmission size while satisfying all the user demands. In this work, we introduce an additional privacy constraint wherein each user $i$ should learn no additional information about the messages other than those present in its side information set and its requested message $x_i$. To enable such private communication, we allow the use of a collection of independent secret keys, each of which is shared amongst a subset of users and is known to the server. The goal is to study properties of the key access structures which make the problem feasible and then design encoding and decoding schemes which utilize them efficiently in terms of the sizes of the server transmission and the secret keys. We call this the {\em private index coding} problem.

The index coding problem was introduced by Birk and Kol~\cite{BirkK98}, where it was observed that significant gain in the rate of transmission can be obtained by utilizing the broadcast nature of the network. Since then, the problem has garnered significant attention and  various aspects have been addressed (see~\cite{ArbabjolfaeiK18} for a survey). Characterizing the optimal transmission rate for an index coding problem is known to be hard~\cite{YossefBJK11} in general and there have been several works aimed at obtaining good upper and lower bounds, see, e.g.,~\cite{Blasiak,BlasiakKL13,ShanmugamDL14,SunJ15,ThapaOJ17}. Restricting attention to linear encoding schemes, Bar-Yossef \emph{et al.}~\cite{YossefBJK11} obtained the optimal rate for a general index coding problem. Lubetzky and Stav~\cite{LubetzkyS09} showed that there exist index coding problems where non-linear codes strictly outperform linear codes.

While the broadcast nature of the network helps in reducing the transmission rate in index coding, it does adversely affect user privacy.
Some recent works have studied the security and privacy aspects of index
coding. Roughly, these works can be divided into two groups. The first group, including~\cite{DauSC12,MojahedianGA15,OngVKY16,OngVYKY16,OngKV18,MojahedianAG17,LiuVKS18,Sun20}, considers the 
 security against an eavesdropper who tries to learn some
information about the messages by wiretapping the broadcast link from the
server to the users. The second group, including~\cite{KarmooseSCF19, NarayananRMDKP18, LiuT19,LiuT20, LiuSAS20,SasiR19}, considers privacy aspects amongst the legitimate users.

Security against an eavesdropper who has access to a
subset of messages was first studied by Dau \emph{et al.}~\cite{DauSC12}. 
They obtained the conditions that any linear code should satisfy to
achieve decodability as well as secrecy. Ong \emph{et al.} studied similar problems~\cite{OngVKY16, OngKV18}, where an equivalence between secure index coding and secure
network coding was shown. For this problem, Liu \emph{et al.}~\cite{LiuVKS18} gave  an achievable scheme using a secret key with vanishingly small rate.
Ong \emph{et al.}~\cite{OngVYKY16} further 
considered a weaker notion of security where the eavesdropper cannot gather any information about each individual message outside its side information set. Note that this does not preclude the eavesdropper from gaining some information about the entire collection of such messages.  

 Mojahedian \emph{et al.}~\cite{ MojahedianAG17} studied the security against an eavesdropper
without any side information and provided schemes which achieve secrecy
using a key shared between the server and the users. 
Securely transmitting a single  message
to a group of users against more than one adversary was considered in a recent work~\cite{Sun20}. Here, the legitimate users and the adversaries possess some subsets of the keys that the server has, and the required transmission rate and the necessary key rates have been analysed for some special cases.

Karmoose \emph{et al.}~\cite{KarmooseSCF19} studied the issue of privacy amongst the legitimate users, where each user wants to hide the identities of its side information messages and its demanded messages from the other users. Making the observation that for a linear code,  a user may learn about the identities of other users side information and requested messages from the encoding matrix, they proposed a scheme to preserve privacy by not fully revealing the encoding matrix~\cite{KarmooseSCF19}. In contrast to this, our work considers a different privacy constraint where no user should get any information concerning the set of messages that have neither been requested by it nor are in its side information set. It will turn out that, under all but trivial cases, the above privacy
requirement can be met only  if some subsets of users possess exclusive keys.

We also study a weaker notion of privacy  where 
each user does not learn any information about each individual message not present in its side information or requested by it, but may gain some information about the entire collection of all such messages. This is similar to the notion of $1$-block
security~\cite{DauSC12} or weak security~\cite{OngVYKY16} defined in the context of
eavesdropper security. Liu \emph{et al.}~\cite{LiuSAS20} extended this setup to a setting where the weak privacy constraint is required only against a subset of the messages not available to a user. 
Recently, some works~\cite{SasiR19,LiuT19,LiuT20}
studied the notion of weak privacy for a variant of the index coding problem called the Pliable Index CODing (PICOD($t$)) introduced by Brahma \emph{et al.}~\cite{BrahmaF15}, where each user wants to obtain any $t \ge 1$ messages that it does not have access to. Sasi \emph{et al.}~\cite{SasiR19} studied
PICOD(1) with the privacy condition that each user is able to decode exactly one message outside its side information set for the special case of consecutive side information. A generalization of~\cite{SasiR19} in terms of the form of the side information sets was investigated by Liu \emph{et al.}~\cite{LiuT19}. The authors further extended their study to a decentralized setting~\cite{LiuT20} where communication occurs among users rather than by a central server.

\subsection{Contributions}
Below, we briefly describe the main contributions of this work\footnote{A preliminary version appeared in~\cite{NarayananRMDKP18}.}.  

\begin{enumerate}
\item \emph{Privacy with keys}: We define the private index coding problem, where user privacy is enabled by sharing secret keys of different sizes amongst the server and various subsets of users. Our interest here is to study the impact of the \emph{key access structures} on the feasibility of the private index coding problem and the tradeoff between the server transmission rate and the sizes of the secret keys. 

\begin{enumerate}

\item We characterize the \emph{key access structures} that make private index coding feasible (Theorem~\ref{thm_key_access}). We give conditions under which a given linear coding scheme is a valid private index code (Theorem~\ref{Thm_linear}).

\item We define the rate region of the private index coding problem as the set of all feasible tuples of the server transmission rates and  key rates. We  characterize the rate region when the number of users is at most three and show that all feasible rates can
	be achieved using scalar linear coding and time sharing (Theorem~\ref{thm_3user}). Further, we give an
	example of a private index coding problem with four users where all the feasible points cannot be
	achieved using scalar linear coding (Theorem~\ref{thm:4user}).  This is in contrast to the non-private index coding problem, where it was
	shown that scalar linear coding is optimal up to 4 users~\cite{Ong14}.  The outer bounds used for characterizing the rate region for three users is generalized  to any private index coding problem (Theorem~\ref{polymatroidal}). These are a generalized version of the polymatroidal bounds for the index coding rate proposed in~\cite{Blasiak}.

\item We study the minimum sum key rate over all  feasible key access structures and provide upper and lower bounds (Theorem~\ref{thm:sum_key}). From the characterization of the rate regions up to three users, we observe that the 
	minimum sum key rate and the optimal server transmission rate are achieved simultaneously in all those cases, i.e., there is no trade-off between the optimal server transmission rate and the minimum sum key rate. However, the question of 
	whether this is indeed the case in general is yet unresolved and left open.
	We do show that there exists a trade-off between the transmission rate and the size of the key access structure (Example~\ref{Ex_sumkey_size}).

\item Finally, we consider the most general notion of private index coding that uses private randomness at the server and the receivers, and also employs common randomness that is available to all agents.
We show that the rate region of such schemes coincides with the more limited family of private index codes which use neither private randomness nor common randomness, thus justifying our limiting attention to such schemes in the rest of the paper.
Our result does not apply for \emph{perfect private index coding}, in which we require `the privacy and correctness error' to be zero at all receivers (in our main definition, we only require these to vanish asymptotically). 
Indeed, we leave the role of private randomness in perfect private index codes as an open problem.

\end{enumerate}
 \item \emph{Weak privacy without keys}:  If no keys are shared among the users, then even weak privacy cannot be achieved for all index coding problems. We derive a necessary condition that an index coding problem has to satisfy for weak privacy to be feasible (Theorem~\ref{prop_neces}). Further, we also give a sufficient condition (Theorem~\ref{Thm_Sec_Cliq}). We characterize the condition under which a linear coding scheme achieves weak privacy (Theorem~\ref{Thm_1b_linear}). 
	
	\item \emph{Privacy through multicasts}: We consider a model in which there is no shared key between the server and the users. However, the server can
multicast to any subset of users and this flexibility is used to ensure privacy in the index coding setting. For this model, we characterize  the minimum number of multicasts required as the fractional chromatic number of a certain graph specified by the index coding problem (Theorem~\ref{Thm_multicast}).
\end{enumerate}

The rest of the paper is organized as follows. We describe our private index coding setup
in Section~\ref{sec_model} and give the characterization of a feasible key
access structure in Section~\ref{Sec_feasble}. Our results on linear
codes is presented in Section~\ref{sec_linear}. Characterization of rate regions up to three users and the polymatroidal outer bounds are provided in Section~\ref{sec_rate}. 
We give our results on sum key rate in Section~\ref{sumkeyrate}. The roles of private and common randomness in private index coding are discussed in Section~\ref{determinism}.
Our results on weak privacy
and private index coding through multicast sessions are provided in
Section~\ref{sec_1b} and Section~\ref{sec_multicast}, respectively.

	\section{Problem Formulation and Preliminaries}
	\label{sec_model}

A server possesses $N$ messages, $\x_1, \ldots, \x_N$, and user $i \in [N]:=\{1,\ldots, N\}$ wants the message $\x_i$. 
We assume that  $\x_i$'s are independent and take values uniformly in a field $\mathbb{F}$.
For a subset of indices $\cS \subseteq [N]$, the set of messages $\{\x_i: i \in \cS\}$ is represented by $\x_{\cS}$.
Each user $i$ has a subset of messages  $\x_{\s_i}$ as side information, where $\s_i \subseteq [N] \setminus \{i\}$. Let $\a_i$ denote the set $\s_i \cup \{ i\}$. 
Index coding problem can be represented by a directed graph $G$ with vertex set $V(G)=[N]$ and edge set $E(G)$ such that $(i,j) \in E$ if and only if $ j \in \s_i$.
Complement of graph $G$, denoted by $G^c$, has vertex set $V(G^c)=[N]$ and edge set $ E(G^c) = (E(G))^c$.

The privacy requirement we consider is that user $i$ should not obtain any information about $\x_{[N] \setminus \a_i}$ in an {\em asymptotic} sense as given later.
The server has access to \emph{keys} that are shared among various subsets of users.
A key is a random variable that is independent of the messages and other keys.
For $\cS \subsetneq [N], \cS \neq \emptyset$, let $\b \in \{0, 1\}^N \setminus \{\vec{0}, \vec{1}\}$ represent the characteristic vector of $\cS$, \emph{i.e.,} $i^{\text{th}}$ bit in $\b$, denoted by $\b^{(i)}$, is 1 if and only if $i \in \cS$.
The key that is available exclusively to users in $\cS$ is denoted by $\k_{\b}$.

We allow block coding, i.e., the server observes
$n$ independent copies of each message before transmission. 
For block length $n$, we assume that the key  $\k_{\b}$
takes values in the set $\{1, \cdots, |\mathbb{F}|^{nR_{\b}}\}$ uniformly at random, where $R_{\b}$ denotes the rate of the key\footnotemark.
\footnotetext{Rates and entropies in this paper are expressed in units of $\log |\mathbb{F}|$ bits.}

Let $\ak$ denote the set of all keys, \emph{i.e.,} $\ak = \{0, 1\}^N \setminus \{\vec{0}, \vec{1}\}$.
For $i \in [N]$, let $\b^{(i)}$ denote the $i^{\text{th}}$ bit in $\b$.
The subset of keys available to the user $i$ is denoted by $\ak_i \defeq  \{\b \in \ak: \b^{(i)} = 1\}$.

An $(n,\Rm, (R_{\b}:\b \in \ak))$ scheme consists of an encoding function $\phi$ and $N$ decoding functions $\{\psi_i\}_{i \in [N]}$. The encoding function 
\begin{align}
\label{eq:encoder}
\phi: \prod_{i \in [N]} \mathbb{F}^n \times \prod_{\b \in \ak} [|\mathbb{F}|^{nR_{\b}}] \longrightarrow [|\mathbb{F}|^{n\Rm}],
\end{align}
outputs the random variable $\mv = \phi\left(\xv_{[N]}, \kv_{\ak}\right)$. For $i \in [N]$, the decoding function
\begin{align}
\label{eq:decoder}
\psi_i: [|\mathbb{F}|^{n\Rm}] \times \prod_{j \in \s_i} \mathbb{F}^n \times \prod_{\b \in \ak_i} [|\mathbb{F}|^{nR_{\b}}] \longrightarrow  \mathbb{F}^n.
\end{align}
maps the message received from the transmission and the side information data to an estimate of the file needed at user $i$
\begin{align}
\widehat{\xv_i} := \psi_i\left(\mv, \xv_{\s_i}, \kv_{\ak_i} \right).
\end{align} 

\begin{definition}
A tuple $(\Rm, (R_{\b}:\b \in \ak))$ is said to be {\em achievable}, if for each $\epsilon>0$
there exists an $(n,\Rm, (R_{\b}:\b \in \ak))$ scheme for some large enough $n$ such that
 the following conditions are satisfied:
\begin{align}
\label{Eq_Dec}
\pr{\widehat{\xv_i} = \xv_i, \forall i \in [N]} \ge 1 - \epsilon,
\end{align}
and 
\begin{align}
\label{Eq_Priv}
I \left( \mv ;\xv_{[N] \setminus \a_i} \middle | \xv_{\s_i}, \kv_{\ak_i} \right) \le n\epsilon.
\end{align}
\end{definition}

For a private index coding problem represented by graph $G$, the {\em rate region} is defined as the closure
of all achievable rate tuples, and it is denoted by $\cR(G)$. Two quantities of interest for a private index coding problem are the optimal transmission rate and the sum key rate which are defined next.

\begin{definition}
\label{def:sumkeyrate}
For a private index coding problem represented by graph $G$, the the optimal transmission rate $R^*(G)$ is defined as 
\begin{align*}
R^*(G) & = \min\{ \Rm: (\Rm,(R_{\b}:\b \in \ak)) \in \cR(G) \text{ for some } (R_{\b}:\b \in \ak) \}.
\end{align*}
For a rate tuple $\vec{R}=(R,(R_{\b} : \b \in \ak))$, the sum key rate is defined as $\mathsf{SumKeyRate}(\vec{R}) \defeq \sum_{\b \in \ak} R_{\b}$.
The minimum sum key rate $\mathsf{SumKeyRate}^{*}(G)$  is defined as 
\begin{align*}
\mathsf{SumKeyRate}^{*}(G) \defeq \min_{\vec{R} \in \cR(G)}\mathsf{SumKeyRate}(\vec{R}).
\end{align*}
\end{definition}

\begin{definition}
	\label{Def_Feasible}
	The {\em key access structure} of a private index code, denoted by $\ks$, is the set of indices corresponding to keys with non-zero rates, \emph{i.e.,} $\ks =  \{\b \in \ak : R_{\b} > 0\}$.

	A key access structure $\ks$ is said to be \emph{feasible} for a private index coding problem represented by graph $G$ if there exists a point in $\cR(G)$ with $R_{\b} = 0$ for all $ \b \notin \ks$.
\end{definition}

We define a stronger notion of \emph{perfect private index coding} where the decoding and privacy conditions are not asymptotic.

\begin{definition}
	\label{def:perfect-pic}
	For a private index coding problem represented by graph $G$, an $(n,\Rm, (R_{\b}:\b \in \ak))$ scheme is said to be a \textbf{perfect private index code} if the decoding error and privacy leakage at all users are  simultaneously zero, i.e.,
\paragraph*{Decoding Condition}
\begin{align}
\label{Eq_Dec_perfect}
\pr{\widehat{\xv_i} = \xv_i, \forall i \in [N]}  = 1
\end{align}
\paragraph*{Privacy condition} At each user $i \in [N]$,
\begin{align}
\label{Eq_Priv_perfect}
I \left( \mv ;\xv_{[N] \setminus \a_i} \middle | \xv_{\s_i}, \kv_{\ak_i} \right) = 0.
\end{align}
\end{definition}
All the private index coding schemes described in the paper indeed achieve this stronger notion of perfect privacy and the converses are shown for the weaker notion of asymptotic privacy as previously defined. A scheme which satisfies~\eqref{Eq_Dec_perfect} and~\eqref{Eq_Priv_perfect} with $n=1$ is called a {\em scalar private index code}.

	\section{Feasibility of Private Index Coding}
	\label{Sec_feasble}

In private index coding, achieving \eqref{Eq_Dec} and \eqref{Eq_Priv} relies on the availability of certain keys among users.
Hence, the feasibility of private index coding depends on the key access structure (e.g., see Fig.~\ref{Fig_feasible0}).
The following theorem characterizes the feasible key access structures for a private index coding problem.
\begin{figure}[h]
	\centering
	\includegraphics[scale =0.6]{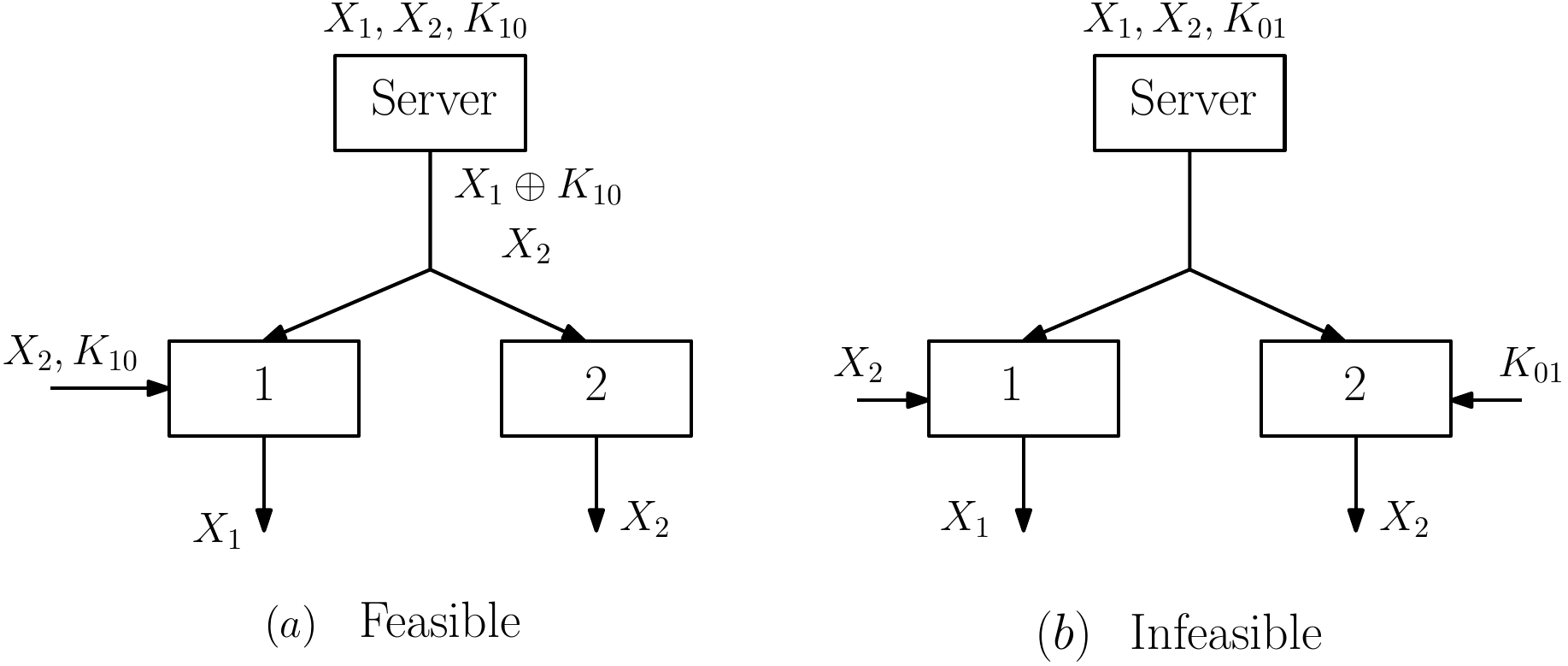}
	\caption{Figure (a) shows a feasible key access structure along with a code that achieves privacy. Figure (b) shows an infeasible key access structure. Here, it can be seen that any encoding that allows users 1 and 2 to decode $X_1$ and $X_2$, respectively, also allows user 2 to decode $X_1$, breaking the privacy condition at user 2.}
	\label{Fig_feasible0}
\end{figure}

\begin{theorem}
	\label{thm_key_access}
	A key access structure $\ks$ is feasible if and only if $\forall \; i,j \in [N] $ such that $i \notin \a_j$, there exists a $\b \in \ks$ such that $b_i =1, b_j=0$.
\end{theorem}

Theorem~\ref{thm_key_access} shows that  if user $j$ does not have $X_i$ as side information, then user $i$ should have a key that is not available with user $j$ to preserve privacy. And also that, if all users satisfy this condition, then we can obtain a private index code.
The full proof of this theorem is provided in the Appendix~\ref{app_thm_key_access}.
Here we give an outline of the proof of this result for \emph{perfect private index coding}.
The general result is shown in a similar way.

\paragraph{If part} Let $\x_i \in \mathbb{F}$, for all $i \in [N]$. 
We consider a scheme which uses $N$ independent copies of each key in the given key access structure, \emph{i.e.,} $\forall \b \in \ks, \kv_{\b} \in \mathbb{F}^N$. 
For $i \in [N]$, similar to the definition of $\ak_i$, we define $\ks_i = \{\b \in \ks : \b^{(i)} = 1\}$, \emph{i.e.,} $\ks_i$ is the subset of keys in the key access structure that is available to user $i$. 
For a representation of $\k_{\b}$ as an $N$-length vector,
let $\kv_{\b}^{(i)}$ denote the $i^{\text{th}}$ co-ordinate of the vector. 
The transmission $\mv \in \mathbb{F}^N$, where  $i^{\text{th}}$ transmission $\mv^{(i)}$, $i\in[N]$, is given by
\begin{align}
\mv^{(i)} = \x_i  + \sum_{\b \in \ks_i} \k_{\b}^{(i)}. \label{Sch_Feas}
\end{align}
From $\mv^{(i)}$, user $i$ can decode $\x_i$ since she has $\k_{\b}$ (specifically, $\k_{\b}^{(i)}$) for all $\b \in \ks_i$.
This scheme also satisfies the privacy condition~\eqref{Eq_Priv_perfect} since if  $i \notin \a_j$, then, by assumption, user $j$ does not have $ \k_{\b}^{(i)}$ for some $\b \in \ks_i$, and it acts as one-time pad for the message $\x_i$ against user $j$.
We note that this scheme is  linear, and a formal proof of the privacy condition follows using Theorem~\ref{Thm_linear}. For details, see Appendix~\ref{app_thm_key_access}.

\paragraph{Only if part} For an $n$ block length perfect private index code, we show that if $i \notin \a_j$, then the rate of the keys that user $i$ has and user $j$ does not have, given by  $ H(\kv_{\ks_i \setminus \ks_j} )$, is lower bounded by $H(\xv_i)$.
From the privacy condition~\eqref{Eq_Priv_perfect} at user $j$, it follows that
\begin{align}
I \left( \mv \; ; \; \xv_{[N] \setminus \a_j} \middle | \xv_{\s_j}, \kv_{\ks_j} \right) & = 0. \nonumber
\intertext{By the decoding condition~\eqref{Eq_Dec_perfect} at user $j$, $\xv_j$ is a function of $(\mv, \xv_{\s_j}, \kv_{\ks_j})$, hence}
I \left( \xv_j, \mv \; ; \; \xv_{[N] \setminus \a_j} \middle | \xv_{\s_j}, \kv_{\ks_j} \right) & = 0. \label{Eq_mutl_info1}
\intertext{Since $i \notin \a_j$, using the chain rule of mutual information,  from \eqref{Eq_mutl_info1} we obtain}
I \left( \mv \; ; \; \xv_{i} \middle | \xv_{[N] \setminus \{i\}}, \kv_{\ks_j} \right) & = 0 \label{Eq_mutl_info}.
\intertext{Using the fact that messages and keys are independent, it can be shown  from~\eqref{Eq_mutl_info} that}
I \left( \mv \; ; \; \xv_{i} \middle | \xv_{[N] \setminus \{i\}}, \kv_{\ks_j \cap \ks_i} \right) & = 0 \label{Eq_feasb2}.
\end{align}
By the decoding condition~\eqref{Eq_Dec_perfect} for user $i$,
\begin{align}
\label{Eq_feasb3}
I \left( \mv ; \xv_i \middle | \xv_{[N] \setminus \{i\}}, \kv_{\ks_i} \right) = H\left(\xv_i\right).
\end{align}
Thus, we have
\begin{align}
I\left( \mv ; \kv_{\ks_i \setminus \ks_j} \middle | \xv_{[N] \setminus \{i\}}, \xv_i, \kv_{\ks_j \cap \ks_i} \right)
&  = I\left( \mv ; \kv_{\ks_i \setminus \ks_j}, \xv_i \middle | \xv_{[N] \setminus \{i\}}, \kv_{\ks_j \cap \ks_i} \right)\label{Eq_feasbl3}\\
& \ge I\left( \mv ; \xv_i \middle | \xv_{[N] \setminus \{i\}}, \kv_{\ks_i} \right)\notag \\ 
& = H\left(\xv_i\right),\label{eq:from_conv_of_feasibility_for_rate_sec}
\end{align}
where \eqref{Eq_feasbl3} follows from~\eqref{Eq_feasb2}, and \eqref{eq:from_conv_of_feasibility_for_rate_sec} follows from~\eqref{Eq_feasb3}. 
Since $I \left( \mv ; \kv_{\ks_i \setminus \ks_j} \middle | \xv_{[N] \setminus \{i\}}, \xv_i, \kv_{\ks_j \cap \ks_i} \right) \leq H(\kv_{\ks_i \setminus \ks_j} )$, and since $H(\xv_i) > 0$, it follows from \eqref{eq:from_conv_of_feasibility_for_rate_sec} that  $R_{\b}$ is non-zero for some $\b \in \ks_i \setminus \ks_j$.

\section{Linear Private Index Codes}
\label{sec_linear}
In this section, we consider linear coding schemes for perfect private index coding.
In Theorem~\ref{Thm_linear}, we characterize linear schemes that satisfy the decoding~\eqref{Eq_Dec_perfect} and privacy~\eqref{Eq_Priv_perfect} conditions.

In the context of linear coding, for block length $n \geq 1$, let $\bX_i$ denote the row-vector corresponding to $\xv_i$ and let $\bK_{\b}$ denote the key uniformly distributed in $\mathbb{F}^{nR_{\b}}$.
The linear encoder is of the form
\begin{align}
\label{Eq_lin_encd}
\bM^T = \sum_{i \in [N]} \bG_i \bX_i^T + \sum_{\b \in \ak} \bH_{\b} \bK_{\b}^T,
\end{align}
where $\bM \in \mathbb{F}^r$, $\bG_i \in \mathbb{F}^{r \times n}$ for $i \in [N]$, and $\bH_{\b} \in \mathbb{F}^{r \times n R_{\b}}$ for $\b \in \ak$.
Transmission rate $R$ is said to be achievable if for some $n\geq 1$ there exists a scheme such that $R=r/n$ and it satisfies \eqref{Eq_Dec_perfect} and \eqref{Eq_Priv_perfect}. If $n = 1$, the scheme is called a \emph{scalar linear code}.

\begin{theorem}
	\label{Thm_linear}
	A linear encoding scheme is a valid perfect private index coding scheme if and only if it satisfies the following conditions for each $i \in [N]$,
	\begin{enumerate}
		\item Let $\bG_i = \left[ \bG_i^{(1)} \ldots \bG_i^{(n)} \right]$, then for each $1 \le k \le n$,
		\begin{align*}
		\bG_i^{(k)} \notin \langle \{\bG_j | {j \notin \a_i}\} \cup \{\bH_{\b} | \b \notin \ak_i\} \rangle,
		\end{align*}
		\item $\langle \{\bG_j | j \notin \a_i\} \rangle \subseteq \langle \{\bH_{\b} | \b \notin \ak_i\} \rangle$.
	\end{enumerate}
\end{theorem}
Here $\langle . \rangle$ denotes the linear span of column vectors.
\begin{proof}
	The first condition is the necessary and sufficient condition for the decodability of $\bX_i$ from the linear code at user $i$ using $\bX_{\s_i}$ and $\bK_{\ak_i}$ as side information.
	
	The second condition is the necessary and sufficient condition for privacy at user $i$.
	Suppose the condition is not satisfied, then $\langle \{\bG_j  | j \notin \a_i \} \rangle$ has a non zero projection onto the space orthogonal to $\langle \{\bH_j | j \notin \ak_i\} \rangle$.
	But this space is a subspace of $\langle \bG_1, \ldots \bG_N, \{\bH_j  | j \in \ak_i\} \rangle$.
	Hence by projecting the transmitted message onto this space, user $i$ may obtain a non-zero linear function of $\bX_{[N] \setminus \a_i}$, hence the privacy condition~\eqref{Eq_Priv_perfect} does not hold at user $i$. The transmitted message can be written as
	
	\begin{align*}
	\bM^T = \sum_{j \in \a_i}\bG_j \bX_j^T & + \sum_{\b \in \ak_i} \bH_{\b} \bK_{\b}^T + \left ( \sum_{j \in [N] \setminus \a_i}\bG_j \bX_j^T + \sum_{\b \notin \ak_i} \bH_{\b} \bK_{\b}^T \right ).
	\end{align*}
	If the second condition is satisfied for $i$, then $\sum_{j \in [N] \setminus \a_i}\bG_j \bX_j^T$ lies in the subspace $\sum_{\b \notin \ak_i} \bH_{\b} \bK_{\b}^T$. But, $\sum_{\b \notin \ak_i} \bH_{\b} \bK_{\b}^T$ is independent of $\bX_{[N]}$ and $\bK_{\ak_i}$ and is uniformly distributed in the space spanned by $\bH_{\b}, \b \notin \ak_i$. Hence, the sum in the brackets is distributed uniformly in the vector space irrespective of the value of $\bX_{[N]}, \bK_{\ak_i}$. This implies that $\bM$ is independent of $\bX_{[N] \setminus \a_i}$ conditioned on $\bX_{\s_i}, \bK_{\ak_i}$, hence the privacy condition is satisfied at user $i$.
\end{proof}

	\section{Rate of Private Index Coding}
	\label{sec_rate}

\subsection{Connection to Index Coding Rate}
We first show a simple connection between private index coding and index coding (without privacy) for the same side information structure.
Given a zero-error index coding scheme (specifically, an optimal scheme) of block-length $n$, we describe a perfect private index code (for a certain key access structure we specify below) with the same transmission rate:

For $i \in [N]$, let $\b_i \in \{0, 1\}^N$ denote the characteristic vector of the set $\a_i$.
We choose $\kv_{\b_i}$ of the same rate as $\xv_i$ (\emph{i.e.,} $\kv_{\b_i}$ is uniform in $\mathbb{F}^n$).
Not that $\kv_{\b_i}$ is available to all users in $\a_i$.
Taking $\{\xv_i + \kv_{\b_i}, i \in [N]\}$ as the messages, the given index coding scheme can be employed to deliver $\xv_i + \kv_{\b_i}$ to user~$i$, $i\in[N]$; note that user~$i$ has access to side-information $\{\xv_j + \kv_{\b_j}: j \in \s_i\}$ as required.
Having access to $\kv_{\b_i}$, user~$i$ can recover $\xv_i$.
Privacy follows from the fact that $\kv_{\b_i}$ is unavailable to any user who should not learn $\xv_i$.
Thus, the optimal transmission rate of the index coding problem is also achievable using \emph{perfect private index codes} for a certain key access structure.

Given a private index coding scheme, it is easy to see, using an averaging argument, that there exists an assignment of values to the keys which gives an index code which guarantees decoding error no more than that in the private index code at each user.
From this observation it is clear that the minimum transmission rate of private index coding cannot be less than that achieved by asymptotic index coding.
Since rate of asymptotic index coding coincides with zero-error index coding~\cite{langberg}, we have the following observation.
\begin{theorem}\label{thm:priv_nopriv}
	For a given side information structure, optimal transmission rates of zero-error index coding and private index coding (optimized over key access structures and key rates) are the same.
\end{theorem}

\subsection{Rate Region of Small Private Index Coding Problems}
The following theorem shows that when $N \le 3$, rate region of private index coding can be characterized.
\begin{theorem}
\label{thm_3user}
	For every private index coding problem $G$ with at most 3 users, the rate region $\mathcal{R}(G)$ is achievable using scalar linear codes and time sharing.
\end{theorem}

\begin{table*}[]
    \begin{tabular}{|l|l|l|l|}
        \hline
        Side info. graph $G$  & Rate Region $\cR(G)$ & Vertices of $\cR(G)$  & A private index code achieving the vertex\\ \cline{1-4}
        \multirow{9}*{
            \begin{tikzpicture}
                \node (1) at (-1, 0) [draw=none]{$1$};
                \node (2) at (1, 0)  [draw=none]{$2$};
                \node (3) at (0, 1)  [draw=none]{$3$};
                \path[->,thick] (1) edge (2);
                \path[->,thick] (2) edge (3);
                \path[->,thick] (3) edge (1);
            \end{tikzpicture}
        } & \multirow{9}*{
            $\begin{array}{lcl}
                R_{100} + R_{101} \ge 1\\
                R_{010} + R_{110} \ge 1\\
                R_{001} + R_{011} \ge 1\\
                R \ge 2\\
                R + R_{011} \ge 3\\
                R + R_{101} \ge 3\\
                R + R_{110} \ge 3\\
            	R_{\b} \ge 0, \b \in \ak
            \end{array}$
        }
        & $(2, 0, 0, 1, 0, 1, 1)$    & $\x_1 + \x_2 + \k_{101} + \k_{110}, \x_2 + \x_3 + \k_{110} + \k_{011}$  \\ \cline{3-4}
        &&    $(3, 0, 0, 1, 1, 1, 0)$    & $\x_1 + \k_{101}, \x_2 + \k_{110}, \x_3 + \k_{001}$\\ \cline{3-4}
        &&    $(3, 0, 1, 0, 0, 1, 1)$    & $\x_1 + \k_{101}, \x_2 + \k_{010}, \x_3 + \k_{011}$\\ \cline{3-4}
        &&    $(3, 0, 1, 0, 1, 1, 0)$    & $\x_1 + \k_{101}, \x_2 + \k_{010}, \x_3 + \k_{001}$\\ \cline{3-4}
        &&    $(3, 1, 0, 1, 0, 0, 1)$    & $\x_1 + \k_{100}, \x_2 + \k_{110}, \x_3 + \k_{011}$\\ \cline{3-4}
        &&    $(3, 1, 0, 1, 1, 0, 0)$    & $\x_1 + \k_{100}, \x_2 + \k_{110}, \x_3 + \k_{001}$\\ \cline{3-4}
        &&    $(3, 1, 1, 0, 0, 0, 1)$    & $\x_1 + \k_{100}, \x_2 + \k_{010}, \x_3 + \k_{011}$\\ \cline{3-4}
        &&    $(3, 1, 1, 0, 1, 0, 0)$    & $\x_1 + \k_{100}, \x_2 + \k_{010}, \x_3 + \k_{001}$\\
        & & &\\ \cline{1-4}
    \end{tabular}
    \caption{For the 3 user private index coding problem $G$, every vertex of the polytope $\cR(G)$ can be achieved using scalar linear PIC. In the table, the vertices are represented as tuples $(R, R_{100}, R_{010}, R_{110}, R_{001}, R_{101}, R_{011})$.}
    \label{Table_Eg1}
\end{table*}

A full proof of this theorem and \emph{the characterization of the rate region} for every private index coding problem with $N \le 3$  is presented in the Appendix~\ref{app_thm_3user}.
Here we illustrate the proof method by characterizing the rate region of \emph{perfect private index codes} for an example.
Consider the side information graph given in Table~\ref{Table_Eg1}.
We first show the necessity of the constraints on the rates. In arguing the ``only if'' part of Theorem~\ref{thm_key_access} we showed that if $i \notin \a_j$, then $\sum_{\b \in \ak_i \setminus \ak_j} R_{\b} \ge H(\xv_i)$.
The first three inequalities in the table follow from this using $H(\xv_i)/n = 1, i = 1, 2, 3$. To see the next inequality, note that the transmission rate of a private index code is lower bounded by the rate of the index coding problem for the same side information graph. Hence, $R$ is lower bounded by number of vertices in the maximum acyclic induced subgraph~\cite{YossefBJK11}, which is 2 in this example. To show the next inequality, we use the bound $H\left(\mv \right) \ge I(\mv ; \xv_1, \xv_2, \xv_3, \kv_{001}, \ldots, \kv_{110})$
and expand the mutual information term as follows,
\begin{multline} \label{Eq_rate_lowr}
    I\left(\mv ; \xv_2, \kv_{100}, \kv_{110}, \kv_{101} \right)+  I\left(\mv ; \xv_1 \middle | \xv_2, \kv_{100}, \kv_{110}, \kv_{101} \right)
	+ I\left(\mv ; \kv_{001}, \kv_{011} \middle | \xv_1, \xv_2, \kv_{100}, \kv_{110}, \kv_{101} \right)\\
    + I\left(\mv ; \xv_3 \middle | \xv_1, \xv_2, \kv_{100}, \kv_{110}, \kv_{101}, \kv_{001}, \kv_{011} \right)
	+I\left(\mv ; \kv_{010} \middle | \xv_{[3]}, \kv_{100}, \kv_{110}, \kv_{101}, \kv_{010}, \kv_{011} \right).
\end{multline}
We lower bound the first and the third terms in the expression by zero.
The decodability condition at user 1 implies that the second term is $H(\xv_1)$.
Using the independence of keys and messages, the fourth term can be written as
\begin{align*}
I \left(\mv ; \xv_3 \middle | \xv_1, \kv_{101}, \kv_{001}, \kv_{011} \right) \stackrel{(a)}{=} H(\xv_3),
\end{align*}
where (a) follows from the decodability condition at user 3.
To bound the fifth term in the expression~\eqref{Eq_rate_lowr}, we note that
\begin{align*}
    & I\left(\mv ; \kv_{010} \middle | \xv_{[3]}, \kv_{100}, \kv_{110}, \kv_{101}, \kv_{010}, \kv_{011} \right) + H\left(\kv_{110}\right) \ge I \left(\mv ; \kv_{010}, \kv_{110} \middle | \xv_{[3]}, \kv_{100}, \kv_{101},  \kv_{010}, \kv_{011} \right)\\
	\stackrel{(a)}= & I \left(\kv_{100}, \kv_{010}, \kv_{011} ; \kv_{010}, \kv_{110} \middle | \xv_{[3]}\right) + I \left(\mv ; \kv_{010}, \kv_{110} \middle | \xv_{[3]}, \kv_{100}, \kv_{101},  \kv_{010}, \kv_{011} \right)\\
    = & I \left(\kv_{100}, \kv_{010}, \kv_{011}, \mv ; \kv_{010}, \kv_{110} \middle | \xv_{[3]}, \kv_{101} \right) \ge I \left(\mv ; \kv_{010}, \kv_{110} \middle | \xv_{[3]}, \kv_{101} \right) \stackrel{(b)}{\ge} H \left(\xv_2 \right).%
\end{align*}
Here, (a) follows from the independence of files and keys and (b) follows from~\eqref{eq:from_conv_of_feasibility_for_rate_sec} with $i = 2$ and $j=3$, since $2 \notin \a_3$.
Hence, the fifth term in~\eqref{Eq_rate_lowr} can be lower bounded by $H(\xv_2) - H(\kv_{110})$.
Putting all these together we have $H(\mv) \ge \sum_{i \in [3]} H(\xv_i) - H(\kv_{110})$ which implies that $R \ge 3 - R_{110}$, similarly we get the next two inequalities.
The table shows that the vertices of the polygon described by these inequalities can be achieved using scalar linear codes.
In the appendix we show this for all graphs with up to 3 vertices, thereby proving the theorem.

In the case of 4 users, we have the following theorem, the proof of which can be found in Appendix~\ref{Sec_4user}.
\begin{theorem}
	\label{thm:4user}
	There is a 4 user private index coding problem where a vector linear code obtains a rate tuple outside the rate region obtained by scalar linear coding and time sharing.
\end{theorem}

\subsection{A Polymatroidal Outer Bound}
The outer bounds we used in establishing the rate regions of all private index coding problems with up to 3 users in Theorem~\ref{thm_3user} can be generalized.
The following theorem provides an outer bound for private index coding by generalizing the polymatroidal bound for the rate of index coding that was proposed in~\cite{Blasiak}.

\begin{theorem}\label{polymatroidal}
    A private index coding rate of $\left(R, R_{\b}: \b \in \ak \right)$ is achievable for $G$ only if there exists a function $f : 2^{[N]} \times 2^{\ak} \rightarrow \mathbb{R}_{\ge 0}$ such that the following conditions are satisfied.
    \begin{align}
    & f\left([N], \ak\right) = 0, \label{pm1}\\
    & f\left(\emptyset, \emptyset\right) \le R, \label{pm2} \\
    & f\left(\cS, \cT\right) \ge f\left(\cS', \cT'\right) \text{ if } \cS \subseteq \cS', \cT \subseteq \cT'\label{pm3}\\
    & \text{For disjoint sets } \cS, \cS', \cS'' \subseteq [N], \cT, \cT', \cT'' \subseteq \ak, \nonumber\\
    & \quad f\left(\cS, \cT\right) - f\left(\cS \cup \cS'', \cT \cup \cT'' \right) \nonumber\\
    & \quad \le f\left(\cS \cup \cS', \cT \cup \cT'\right) - f\left(\cS \cup \cS' \cup \cS'', \cT \cup \cT' \cup \cT'' \right) \label{pm4}\\
    & f\left(\s_i, \ak_i\right) - f\left(\a_i, \ak_i\right) \ge 1, \forall i \in [N], \label{pm5}\\
    & f\left(\a_i, \ak_i\right) - f\left([N], \ak_i\right) = 0, \forall i \in [N], \label{pm6}\\
    & f\left([N], \ak \setminus \cT\right) \le \sum_{\b \in \cT} R_{\b}, \forall \cT \subseteq \ak \label{pm7}.
    \end{align}
\end{theorem}
\begin{proof}
To be more explicit, we will denote the files, keys and transmitted message in an $(n, R, R_{\b} : \b \in \ak)$ scheme by $\xv_i : i \in [N]$, $\bl{\k}{n}_{\b} : \b \in \ak$ and $\bl{\m}{n}$, respectively.
If a rate of $(R, R_{\b} : \b \in \ak)$ is achievable, then there is a sequence of $(n_{\ell}, R, R_{\b} : \b \in \ak)_{\ell \in \mathbb{N}}$ schemes such that for each $\ell \in \mathbb{N}$, when $\epsilon_{\ell} = \frac{1}{\ell}$,
\begin{align}
    \pr{\widehat{\iid{\x}{n_{\ell}}_i} = \iid{\x}{n_{\ell}}_i, \forall i \in [N]} \ge 1 - \epsilon_{\ell} \text{ and }
    I \left( \bl{\m}{n_{\ell}} ;\iid{\x}{n_{\ell}}_{[N] \setminus \a_i} \middle | \bl{\k}{n_{\ell}}_{\ak_i}, \iid{\x}{n_{\ell}}_{\s_i} \right) \le n_{\ell} \epsilon_{\ell}, \forall i \in [N]. \label{eqn:dec_priv_seq}
\end{align}
In the proof, we will use a sub-sequence of the above sequence corresponding to block-lengths $(n_\ell)_{\ell \in \mathbb{S}}$ (where $\mathbb{S}$ is an infinite subset of $\mathbb{N}$) that has the convergence property defined below.
\begin{definition}
A sequence of private index coding schemes of block-lengths $(n_\ell)_{\ell \in \mathbb{S}}$ (where $\mathbb{S}$ is an infinite subset of $\mathbb{N}$) is said to be convergent if the sequence
$
\left(\frac{1}{n_{\ell}} H\left(\bl{\m}{n_{\ell}} \middle | \iid{\x}{n_{\ell}}_{\cS}, \bl{\k}{n_{\ell}}_{\cT}\right)\right)_{\ell \in \mathbb{S}}
$
converges for each $\cT \subseteq \ak$ and $S \subseteq [N]$.
\end{definition}
Note that such a convergent sub-sequence always exists.
This is because, since the rate of all schemes in the above sequence is $(R, R_{\b} : \b \in \ak)$, for every block length $n$ in the sequence and every set $\cT \subseteq \ak$, $\frac{1}{n} H\left(\bl{\m}{n} \middle | \iid{\x}{n}_{[N]}, \bl{\k}{n}_{\cT}\right)$ can be upper bounded as
\begin{align}
\frac{1}{n} H\left(\bl{\m}{n} \middle | \iid{\x}{n}_{[N]}, \bl{\k}{n}_{\cT}\right) \le \frac{1}{n} H\left(\bl{\m}{n}\right) \le R. \label{pm_sub_seq}
\end{align}

Consider the convergent sub-sequence of private index coding schemes corresponding to block-lengths $(n_\ell)_{\ell \in \mathbb{S}}$ described above.
For $\cS \subseteq [N]$ and $\cT \subseteq \ak$, we define $f(\cS, \cT)$ as
\begin{align*}
f(\cS, \cT) \defeq \lim_{\ell \rightarrow \infty} \frac{1}{n_{\ell}}  I\left( \bl{\m}{n_{\ell}} ; \iid{\x}{n_{\ell}}_{[N] \setminus \cS}, \bl{\k}{n_{\ell}}_{\ak \setminus \cT} \mid \iid{\x}{n_{\ell}}_{\cS}, \bl{\k}{n_{\ell}}_{\cT}\right).
\end{align*}
Since the sub-sequence we consider is convergent, the limits in the definition of $f$ is well defined.
Equality~\eqref{pm1} is trivially true.
Inequality~\eqref{pm2} can be shown as follows.
\begin{align*}
f(\emptyset, \emptyset) = \lim_{\ell \rightarrow \infty} \frac{1}{n_{\ell}}  I\left(\bl{\m}{n_{\ell}} ; \iid{\x}{n_{\ell}}_{[N]}, \bl{\k}{n_{\ell}}_{\ak}\right) \le \lim_{\ell \rightarrow \infty} \frac{1}{n_{\ell}} H\left(\bl{\m}{n_{\ell}}\right) = R.
\end{align*}
Inequality~\eqref{pm3} is a direct consequence of the chain rule of mutual information.
By definition of $f$, when $\cS, \cS', \cS''$ and $\cT, \cT', \cT''$ are disjoint,
\begin{align*}
    f(\cS, \cT) - f(\cS \cup \cS'', \cT \cup \cT'')
    = \lim_{\ell \rightarrow \infty} \frac{1}{n_{\ell}} I\left(\bl{\m}{n_{\ell}} ; \iid{\x}{n_{\ell}}_{\cS''}, \bl{\k}{n_{\ell}}_{\cT''} \middle | \iid{\x}{n_{\ell}}_{\cS}, \bl{\k}{n_{\ell}}_{\cT}\right),
\end{align*}
and
\begin{align*}
    f(\cS \cup \cS', \cT \cup \cT') - f(\cS \cup \cS' \cup \cS'', \cT \cup \cT', \cup \cT'')
    = \lim_{\ell \rightarrow \infty} \frac{1}{n_{\ell}} I\left(\bl{\m}{n_{\ell}} ; \iid{\x}{n_{\ell}}_{\cS''}, \bl{\k}{n_{\ell}}_{\cT_{3}} \middle | \iid{\x}{n_{\ell}}_{\cS \cup \cS'}, \bl{\k}{n_{\ell}}_{\cT_{1} \cup \cT'}\right).
\end{align*}
Inequality~\eqref{pm4} now follows from the following lemma when $\cS_1, \cS_2, \cS_3, \cS_4$ are replaced by $\cS'', \emptyset, \cS, \cS'$, respectively and $\cT_1, \cT_2, \cT_3, \cT_4$ are replaced by $\cT'', \emptyset, \cT, \cT'$, respectively.
The proof of the lemma is provided in Appendix~\ref{app_thm_key_access}.
\begin{lemma}\label{lem_ind}
    Let $\cS_1, \cS_2, \cS_3, \cS_4$ be disjoint subsets of  $[N]$ and $\cT_1, \cT_2, \cT_3, \cT_4$ be disjoint subsets of $\{0, 1\}^N \setminus \{\vec{1}, \vec{0}\}$. Then,
    \begin{align*}
    I \left( \bl{\m}{n} ; \xv_{\cS_1 \cup \cS_2}, \bl{\k}{n}_{\cT_1 \cup \cT_2} \middle | \xv_{\cS_3 \cup \cS_4}, \bl{\k}{n}_{\cT_3 \cup \cT_4} \right) \ge I \left( \bl{\m}{n} ; \xv_{\cS_1}, \bl{\k}{n}_{\cT_1} \middle | \xv_{\cS_3}, \bl{\k}{n}_{\cT_3} \right). 
    \end{align*}
\end{lemma}
We now prove inequalities~\eqref{pm5}~and~\eqref{pm6}.
For $i \in [N]$,
\begin{align*}
    f(\s_i, \ak_i) - f(\a_i, \ak_i) = \lim_{\ell \rightarrow \infty} \frac{1}{n_{\ell}} I\left(\bl{\m}{n_{\ell}} ; \iid{\x}{n_{\ell}}_{i} \middle | \iid{\x}{n_{\ell}}_{\s_i}, \bl{\k}{n_{\ell}}_{\ak_i}\right).
\end{align*}
and
\begin{align*}
    f(\a_i, \ak_i) - f([N], \ak_i) = \lim_{\ell \rightarrow \infty} \frac{1}{n_{\ell}} I\left(\bl{\m}{n_{\ell}} ; \iid{\x}{n_{\ell}}_{[N] \setminus \a_i} \middle | \iid{\x}{n_{\ell}}_{\a_i}, \bl{\k}{n_{\ell}}_{\ak_i}\right).
\end{align*}
We state a lemma which immediately implies inequalities~\eqref{pm5}~and~\eqref{pm6}.
The proof of the lemma is provided in Appendix~\ref{app_thm_key_access}.
\begin{lemma}\label{lem_enc_dec}
    For block-length $n$, $\epsilon > 0$ and $i \in [N]$, if
    \begin{align*}
        I \left( \mv; \xv_{[N] \setminus \a_i} \middle |\xv_{\s_i}, \kv_{\ak_i} \right) & \le n\epsilon,\\
        \pr{\psi_i\left(\mv, \xv_{\s_i}, \kv_{\ak_i}\right) \neq \xv_i} & \le \epsilon,
    \end{align*}
    then, when $h(.)$ denotes the Boolean entropy function, \emph{i.e.,} $h(\epsilon) = \epsilon \log_{|\mathbb{F}|}{\frac{1}{\epsilon}} + (1 - \epsilon) \log_{|\mathbb{F}|}{\frac{1}{1 - \epsilon}}$, for sufficiently small values of $\epsilon$,
    \begin{align}
        I \left( \mv; \xv_{[N] \setminus \a_i} \middle |\xv_{\a_i}, \kv_{\ak_i} \right) & \le 3n h(\epsilon)\\
        I \left( \mv ; \xv_i \middle | \xv_{\s_i}, \kv_{\ak_i} \right) &\ge H \left( \xv_i \right) - 3n h(\epsilon).
    \end{align}
\end{lemma}
By the conditions in~\eqref{eqn:dec_priv_seq}, for $\ell \in \mathbb{S}$, $\epsilon_{\ell} \rightarrow 0$ as $\ell \rightarrow \infty$ and, consequently, $h(\epsilon_{\ell}) \rightarrow 0$.
Hence, using Lemma~\ref{lem_enc_dec}, inequalities~\eqref{pm5}~and~\eqref{pm6} can be obtained by taking $\ell \rightarrow \infty$ such that $\ell \in \mathbb{S}$.

Finally, to see~\eqref{pm7}, for every $\cT \subseteq \ak$,
\begin{align*}
    f([N], \ak \setminus \cT) & = \lim_{n_{\ell}} \frac{1}{n_{\ell}}  H\left(\bl{\m}{n_{\ell}} \middle | \iid{\x}{\ell}_{[N]}, \bl{\k}{n_{\ell}}_{\ak \setminus \cT}\right)  - \frac{1}{n_{\ell}}  H\left(\bl{\m}{n_{\ell}} \middle | \iid{\x}{n_\ell}_{[N]}, \bl{\k}{n_{\ell}}_{\ak}\right)\\
    & = \lim_{n_{\ell}} \frac{1}{n_{\ell}}  I\left(\bl{\m}{n_{\ell}} ; \bl{\k}{n_{\ell}}_{\cT} \middle | \iid{\x}{n_{\ell}}_{[N]}, \bl{\k}{n_{\ell}}_{\ak \setminus \cT}\right)\\
    & \le \lim_{n_{\ell}} \frac{1}{n_{\ell}} \cdot H\left(\bl{\k}{n_{\ell}}_{\cT}\right) \\
    & = \sum_{\b \in \cT} R_{\b}.
\end{align*}
\end{proof}
\begin{remark}
We may recover the polymatroidal bound for index coding rate of~\cite[Theorem 5.1]{ArbabjolfaeiK18} (for equal sized files) as a special case of the bound given in the above theorem.
Define $f' : 2^{[N]} \rightarrow \mathbb{R}_{\ge 0}$ as $f'(\cS) = f([N] \setminus \cS, \ak)$.
Theorem~\ref{polymatroidal} imposes the following inequalities of $f'$.
\begin{align}
    f'(\emptyset) &= 0\\
    f'([N]) &\le R\label{pm_ic1}\\ 
    f'(\cJ) &\le f'(\cK) \text{ if } \cJ \subseteq \cK\label{pm_ic2}\\
    f'(\cJ \cap \cK) + f'(\cJ \cup \cK) &\le f'(\cJ) + f'(\cK)\label{pm_ic3}\\ 
    f'([N] \setminus \s_i) - f'([N] \setminus \a_i) &\ge 1, \forall i \in [N].\label{pm_ic4}
\end{align}
Here, $f'([N]) = f(\emptyset, \ak) \le f(\emptyset, \emptyset) \le R$.
Inequality~\eqref{pm_ic3} is obtained from rearranging the inequality~\eqref{pm4} after setting $\cS = [N] \setminus (\cJ \cup \cK), \cS' = \cJ \setminus \cK, \cS'' = \cK \setminus \cJ, \cT = \ak$ and $\cT', \cT'' = \emptyset$.
Inequality~\eqref{pm_ic4} is implied by the inequalities~\eqref{pm5}~and~\eqref{pm3}.
It can be verified that this is indeed the constraint on transmission rate implied by the polymatroidal inner bound when all files are of the same size.
\end{remark}

	\section{Sum Key Rate}
	\label{sumkeyrate}

In this section, we discuss our results on the minimum sum key rate. To this end, we first give definitions of some well-known graph theoretic quantities.
\subsection{Graph-theoretic definitions}
Consider a graph $G$ and assign a color to each vertex in its vertex set $V(G)$, such that no two neighboring vertices share the same color. The minimum number of colors required  is called the chromatic number of graph $G$, which we denote by $\chi(G)$.
 \begin{definition}
	\label{Def_bfold}
	Let  $[L]=\{1, \cdots, L\}$ be a set of $L$ colors.
	Let each $v \in V(G)$ be assigned a subset of colors of size $b$ of the set $[L]$ such that any two adjacent nodes get disjoint sets. Such an assignment is called a $b$-fold coloring and the minimum $L$ for which a $b$-fold coloring exists is called the $b$-fold chromatic number of $G$, denoted by $\chi_b(G)$.
\end{definition}

\begin{definition}
	\label{Def_chrom_num}
	The fractional chromatic number $\chi_f(G)$ of a graph $G$ is defined as 
	\begin{align*}
	\chi_f(G) = \lim_{b\to\infty} \frac{\chi_{b}(G)}{b} = \inf_{b} \frac{\chi_{b}(G)}{b}.
	\end{align*}
\end{definition}
It is easy to verify that $\chi_{b}(G)$ is subadditive. Thus, the limit exists.

\subsection{Bounds on the Minimum Sum Key Rate}

Recall the definition of $\mathsf{SumKeyRate}^{*}(G)$ from Definition~\ref{def:sumkeyrate}. In  Theorem~\ref{thm:sum_key}, we give upper and lower bounds on $\mathsf{SumKeyRate}^{*}(G)$.

\begin{theorem}\label{thm:sum_key}
For a private index coding problem represented by graph $G$, $\mathsf{SumKeyRate}^{*}(G)$ always satisfies the following:
\begin{align}
R^*(G) - 1 \le \mathsf{SumKeyRate}^{*}(G)\leq \chi_f(G^c).  \label{Eq_bounds_keyrate}
\end{align}
\end{theorem}
The upper bound on $\mathsf{SumKeyRate}^{*}(G)$ follows by showing that if $G^c$ has an $n$-fold coloring using $C$ colors, then $G$ has a private index code with sum key rate and transmission rate $C/n$.
To show the lower bound, we first argue that if the index coding problem satisfies that for every $i \in [N]$, there exists a $j \in [N]$,  such that $i \notin \a_j$, then $\mathsf{SumKeyRate}^{*}(G)$ is lower bounded by the optimal transmission rate $R^*(G)$.
We then generalize this result to any arbitrary side information structure to obtain the lower bound in~\eqref{Eq_bounds_keyrate}.
The full proof of Theorem~\ref{thm:sum_key} can be found in Appendix~\ref{app_thm:sum_key}.

Next we show that the bounds on $\mathsf{SumKeyRate}^{*}(G)$ in Theorem~\ref{thm:sum_key} are loose in general using the following example.
\begin{figure}
	\centering
	\begin{tikzpicture}
	\node (1) at (0, 3) [draw=none]{$1$};
	\node (2) at (-2, 1.5)  [draw=none]{$2$};
	\node (3) at (-1, 0)  [draw=none]{$3$};
	\node (4) at (1, 0)  [draw=none]{$4$};
	\node (5) at (2, 1.5)  [draw=none]{$5$};
	\path[->,thick] (1) edge (2);
	\path[->,thick] (1) edge (3);
	\path[->,thick] (2) edge (3);
	\path[->,thick] (2) edge (4);
	\path[->,thick] (3) edge (4);
	\path[->,thick] (3) edge (5);
	\path[->,thick] (4) edge (5);
	\path[->,thick] (4) edge (1);
	\path[->,thick] (5) edge (1);
	\path[->,thick] (5) edge (2);
	\end{tikzpicture}
	\caption{A private index coding problem for which the minimum sum key rate is strictly larger than the minimum transmission rate and strictly smaller than the fractional chromatic number of its complement. This shows that both the bounds in Theorem~\ref{thm:sum_key} can be simultaneously loose.}
	\label{G_Q}
\end{figure}
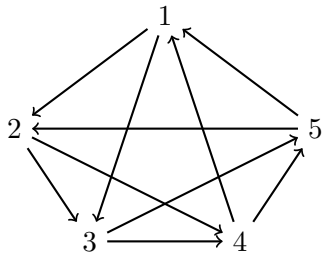 
\begin{example}
	\label{Ex_sumkey_bounds}
	Consider the graph $G$ shown in Fig.~\ref{G_Q}. The complement graph $G^c$ of $G$ is also isomorphic to $G$.
	Since any two nodes in $G^c$ are connected, it follows that $\chi(G^c) =\chi_f(G^c) = 5$. Now consider the following linear private index code.
	\begin{align*}
	\x_1 + \x_2 + \x_3 + \k_{10010} + \k_{11001}, \\
	\x_2 + \x_3 + \x_4 + \k_{11001} + \k_{01110}, \\
	\x_3 + \x_4 + \x_5 + \k_{01110} + \k_{00101}.
	\end{align*}
	It can be verified that this linear code satisfies the conditions in Theorem~\ref{Thm_linear} and is thus a perfect private index code. The sum key rate of this code is 4
	which is strictly less than $\chi_f(G^c) = 5$. 
	This shows that the upper bound in Theorem~\ref{thm:sum_key} is loose in general.
	In Appendix~\ref{app_thm:sum_key}, we show that the minimum sum key rate is strictly larger than 3, while the minimum transmission rate is clearly at most 3 as the scheme constructed shows. This further implies that the lower bound in Theorem~\ref{thm:sum_key} is not tight.
\end{example}

A natural question that arises given these observations is whether there is a trade-off between the transmission rate and the sum key rate.
From the proof of Theorem~\ref{thm:priv_nopriv}, it is clear that the optimal transmission rate may be achieved with a sum key rate of $N$.
Up to $3$ users, the minimum sum key rate and the optimal transmission rate for any private index coding problem can be achieved simultaneously.
Whether this is indeed the case in general remains open. In the next subsection, we discuss a related but different trade-off, the one between the transmission rate and the size of the key access structure.

\subsection{Trade-off between the transmission rate and the size of key access structure}

For a given a key access structure $\ks$, we refer to the cardinality of $\ks$ as the \emph{size of key access structure} $\ks$.
Next, we give an example to show that there indeed exists a trade-off between the transmission rate and the size of the key access structure.
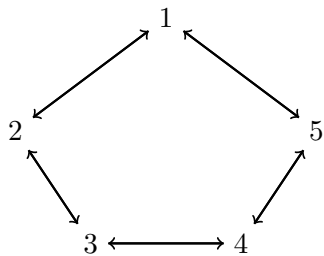
\begin{figure}
	\centering
	\begin{tikzpicture}
	\node (1) at (0, 3) [draw=none]{$1$};
	\node (2) at (-2, 1.5)  [draw=none]{$2$};
	\node (3) at (-1, 0)  [draw=none]{$3$};
	\node (4) at (1, 0)  [draw=none]{$4$};
	\node (5) at (2, 1.5)  [draw=none]{$5$};
	\path[->,thick] (1) edge (2);
	\path[->,thick] (2) edge (1);
	\path[->,thick] (2) edge (3);
	\path[->,thick] (3) edge (2);
	\path[->,thick] (3) edge (4);
	\path[->,thick] (4) edge (3);
	\path[->,thick] (4) edge (5);
	\path[->,thick] (5) edge (4);
	\path[->,thick] (5) edge (1);
	\path[->,thick] (1) edge (5);
	\end{tikzpicture}
	\caption{A private index coding problem for which there is a trade-off between the size of the key access structure and the optimal transmission rate. The transmission rate for every feasible key access structure with 3 keys is at least 3 while the optimal transmission rate is 2.5.}
	\label{G_kas}
\end{figure}

\begin{example}
	\label{Ex_sumkey_size}
	The minimum transmission rate of private index codes for the graph $G$ in Fig.~\ref{G_kas} is at least 3 for any key access structure $\ks$ such that $|\ks| \leq 3$.
	We prove this in Appendix~\ref{app_thm:sum_key}.
	 
	On the other hand, the following scheme with block length 2 achieves a lower transmission rate of $2.5$ for a key access structure $\ks$ such that $|\ks| =5$.
	For $i \in [5]$, let $\x_i^{(1)}$ and $\x_i^{(2)}$ denote the first and second co-ordinates of $\x_i^2$, respectively.
	The server makes five transmissions as described below.
	\begin{align*}
	\x_1^{(1)} + \x_2^{(1)} + K_{11000}, \\
	\x_3^{(1)} + \x_4^{(1)} + K_{00110}, \\
	\x_5^{(1)} + \x_1^{(2)} + K_{10001}, \\
	\x_2^{(2)} + \x_3^{(2)} + K_{01100}, \\
	\x_4^{(2)} + \x_5^{(2)} + K_{00011}.
	\end{align*} 
	Using Theorem~\ref{Thm_linear}, it is easy to verify that this scheme satisfies the decodability and privacy conditions.
	The transmission rate is clearly 2.5 which is strictly less than the transmission rate achievable using a key access structure of size at most 3.
	This demonstrates that there exists a trade-off between the transmission rate and the size of the key access structure. 
\end{example}

	\section{Private and Public Randomness in Private Index Coding}
	\label{determinism}
In this section, we consider \emph{general private index coding schemes} with randomized encoder and decoders with common randomness available at the server and all the users ($\credpic$).
We would like to note that the common randomness can be thought of as a key available to the server and all the users and private randomness at encoder as key available only to the server, which could be treated as $\k_{\vec{1}}$ and $\k_{\vec{0}}$, respectively.
Since we are not interested in the rate of common and private randomness, instead of treating them like keys, we treat the common randomness and private randomness at encoder (denoted by $W$ and $W_{\phi}$, respectively) similar to private randomness of decoders (denoted by $W_{\psi_i}$, for user $i \in [N]$). 
We show that the rate region of such schemes is no larger that the the rate region of private index codes ($\pic$) as defined in Section~\ref{sec_model}.

Formal definitions of general private index coding schemes with randomized encoder and decoders that use common randomness and its different variations follow.
\begin{definition}
A $\credpic$ scheme (Private Index Coding scheme with Common randomness and Randomness at Encoder and Decoder) of rate $\left(R, R_{\b}: \b \in \ak\right)$ and block length $n$ is a private index coding scheme with randomized encoder and decoder that uses common and private randomness described by an encoder
\begin{align}
\label{Eq_encod_gen}
 \phi: \prod_{i \in [N]} \mathbb{F}^n \times \prod_{\b \in \ak} [|\mathbb{F}|^{nR_{\b}}] \times \cW \times \cW_{\phi} \longrightarrow [|\mathbb{F}|^{nR}],
\end{align}
and decoder $\psi_i$ for each user $i \in [N]$ described as,
\begin{align}
\label{Eq_decod_gen}
	\psi_i: [|\mathbb{F}|^{nR}] \times \prod_{j \in \s_i} \mathbb{F}^n \times \prod_{\b \in \ak_i} [|\mathbb{F}|^{nR_{\b}}] \times \cW \times \cW_{\psi_i} \longrightarrow  \mathbb{F}^n.
\end{align}
Here, $\cW, \cW_{\phi}$ and $\cW_{\psi_i} : i \in [N]$ are arbitrarily large finite sets.

The common randomness $W$ available to all users and the server is uniformly distributed in $\cW$, the private randomness $W_{\phi}$ of the server and $W_{\psi_i}$ of user $i$ are uniformly distributed in $\cW_{\phi}$ and $\cW_{\psi_i}$, respectively.
We define the random variables corresponding to the transmitted message and the estimate of $\xv_i$ computed by user $i$, respectively, as
\begin{align*}
	\mv \defeq \phi\left(\xv_{[N]}, \kv_{\ak}, W, W_{\phi}\right) \text{ and }
	\widehat{\xv_i} \defeq \psi_i\left(\mv, \xv_{\s_i}, \kv_{\ak_i}, W, W_{\psi_i}\right).
\end{align*}
For $\epsilon, \delta \ge 0$, the above described encoder $\phi$ and decoders $\psi_i : i \in [N]$ constitute a $(n, \epsilon, \delta)$-$\credpic$ scheme of rate $\left(R, R_{\b}: \b \in \ak\right)$ if the following decoding and privacy conditions are satisfied.
\begin{flalign}
  \text{$\epsilon$-Decoding error:} &&\pr{\widehat{\xv_i} = \xv_i, \forall i \in [N]} & \ge 1 - \epsilon,\label{Eq_Dec_gen}&\\ 
  \text{$\delta$-Privacy error:}  &&I \left( \mv ;\xv_{[N] \setminus \a_i} \middle | \kv_{\ak_i}, \xv_{\s_i}, W \right) & \le n\delta, \text{ for all } i \in [N]. \label{Eq_Priv_gen}
\end{flalign}
Rate $\left(R, R_{\b}: \b \in \ak\right)$ is said to be achievable using $\credpic$ schemes if for every $\epsilon > 0$, there exists a $(n, \epsilon, \epsilon)$-$\credpic$ scheme for a large enough block-length $n$.

Several variations of private index codes of the same rate that use private and common randomness to varying extends are defined below.
These definitions will be crucially used in proving that the rate region of $\credpic$ and $\pic$ coincide.
A summary of the definitions is provided in Table~\ref{table:def-schemes}.
\begin{table}[]
\begin{tabular}{lllll}
                    & \begin{tabular}[c]{@{}l@{}}Private randomness\\ at decoders \end{tabular}& \begin{tabular}[c]{@{}l@{}}Common randomness\\ (encoder and decoders)\end{tabular} & \begin{tabular}[c]{@{}l@{}}Private randomness\\ at Encoder \end{tabular}&\begin{tabular}[c]{@{}l@{}} Zero-error \\ decoding \end{tabular}\\
$\credpic$          &$\checkmark$& $\checkmark$& $\checkmark$& $\times$\\
$\crepic$           &$\times$&$\checkmark$&$\checkmark$&$\times$\\
$\repic$            &$\times$&$\times$&$\checkmark$&$\times$\\
zero-error $\repic$ &$\times$&$\times$&$\checkmark$&$\checkmark$\\
$\pic$              &$\times$&$\times$&$\times$&$\times$                                   
\end{tabular}
\caption{\label{table:def-schemes} A summary of the variations of private index codes that use randomness to varying degrees.}
\end{table}
\begin{itemize}
	\item If $\cW_{\psi_i} = \emptyset$ for all $i \in [N]$, \emph{i.e.,} decoders do not use private randomness, the above scheme is said to be $(n, \epsilon, \delta)$-$\crepic$ (Private Index Coding scheme with Common randomness and Randomness at Encoder).
	\item If, additionally, $\cW = \emptyset$, \emph{i.e.,} the scheme does not use common randomness, the above scheme is said to be $(n, \epsilon, \delta)$-$\repic$(Private Index Coding scheme with Randomness at Encoder).
	\item $(n, \delta)$-zero-error $\repic$ denotes $(n, 0, \delta)$-$\repic$.
	\item If, additionally, $\cW_{\phi} = \emptyset$, \emph{i.e.,} the encoder does not use private randomness, the above scheme is said to be $(n, \epsilon, \delta)$-$\pic$(Private Index Coding scheme with Randomness at Encoder). This coincides with our definition of private index codes given in Section~\ref{sec_model}.
\end{itemize}
Achievability using $\crepic, \repic$ and zero-error $\repic$ schemes are defined similarly as the achievability for $\credpic$.

\begin{figure}
	\centering
	\begin{tikzpicture}
	\node (1) at (0, 0) [draw=none]{$\credpic$};
	\node (01) at (2, 1) [draw=none, text width=2cm, align=center]{Maximum likelihood decoding};
	\node (01a) at (2, -0.5) [draw=none, text width=2cm, align=center]{Lemma~\ref{lem_rand_dec}};
	\node (2) at (4, 0)  [draw=none]{$\crepic$};
	\node (12) at (6, 1) [draw=none, text width=2cm, align=center]{Averaging argument};
	\node (23a) at (6, -0.5) [draw=none, text width=2cm, align=center]{Lemma~\ref{lem_no_cr}};
	\node (3) at (8, 0)  [draw=none]{$\repic$};
	\node (34) at (9.5, 1) [draw=none, text width=2cm, align=center]{Hybrid coding};
	\node (34a) at (9.5, -0.5) [draw=none, text width=2cm, align=center]{Lemma~\ref{lem_0_err}};
	\node (4) at (12, 0)  [draw=none]{zero-error $\repic$};
	\node (56) at (14.5, 1) [draw=none, text width=2cm, align=center]{Channel simulation};
	\node (56a) at (14.5, -0.5) [draw=none, text width=2cm, align=center]{Lemma~\ref{lem_det}};
	\node (5) at (16, 0)  [draw=none]{$\pic$};
	\path[<->,thick] (1) edge (2);
	\path[<->,thick] (2) edge (3);
	\path[<->,thick] (3) edge (4);
	\path[<->,thick] (4) edge (5);
	\end{tikzpicture}
	\caption{Theorem~\ref{thm:private_rnd} is proved in four stages by showing that the rate regions of different variants of private index coding schemes coincide.
	}
	\label{fig:priv-rand-proof}
\end{figure}
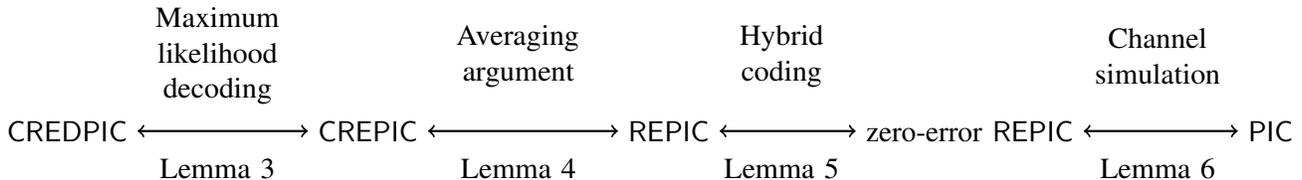
\end{definition}

\begin{theorem}\label{thm:private_rnd}
	For any side information structure $G$, the rate region of $\pic$ is identical to the rate region of $\credpic$.
\end{theorem}

\begin{proof}
The theorem is proved in four stages, summarized in Figure~\ref{fig:priv-rand-proof}.
Lemmas~\ref{lem_rand_dec}~and~\ref{lem_no_cr} show that the use of randomization at decoder and/or using common randomness do not enhance the rate region.
We will use this observation to show in Lemma~\ref{lem_0_err} that the rate region of zero-error $\repic$ coincides with that of $\credpic$.
Finally, Lemma~\ref{lem_det} shows that the rate region of zero-error $\repic$ is identical to that of $\pic$.
Clearly, these lemmas imply the statement of the theorem.
\end{proof}

The following two lemmas show that the rate region of $\credpic, \crepic$ and $\repic$ are identical for a given side information structure.

\begin{lemma}\label{lem_rand_dec}
	Given an $(n, \epsilon, \delta)$-$\credpic$ scheme for a side information structure $G$, one can construct an $(n, \epsilon, \delta)$-$\crepic$ scheme of the same rate for $G$.
\end{lemma}
\begin{proof}
	Note that the privacy condition at a user does not depend on the decoding function.
	Hence, the private randomness at the decoders may be removed by using a maximum likelihood decoder so that the decoding error remains at most $\epsilon$.
	It is worth noting that this argument cannot be used to transform a randomized encoder into a deterministic one, since fixing the randomness at the encoder will directly affect the privacy condition.
\end{proof}

\begin{lemma}\label{lem_no_cr}
	Given an $(n, \epsilon, \delta)$-$\crepic$ scheme for a side information structure $G$, one can construct an $(n, 2\epsilon, 2N\delta)$-$\repic$ scheme of the same rate for $G$.
\end{lemma}
This is proved by showing the existence of a realization of the common randomness for which the decoding error is at most $2\epsilon$ (Condition~\ref{Eq_Dec_gen}) and privacy leakage is $2N\delta$ (Condition~\ref{Eq_Priv_gen}).
The lemma follows from this observation since we may fix the common randomness to be this realization.
Detailed proof is given in Appendix~\ref{app_thm:private_rnd}.
The following corollary follows from the above two lemmas.

\begin{corollary}\label{cor_nocrd}
	For $G$, if a rate $(R, R_{\b}: \b \in \ak)$ is achievable using $\credpic$ schemes, then this rate is also achievable using $\repic$ schemes.
\end{corollary}

\begin{lemma}\label{lem_0_err}
	If a rate of $(R, R_{\b}: \b \in \ak)$ is achievable using $\credpic$ schemes, then for every $\delta > 0$, a rate of $(R + \delta, R_{\b}: \b \in \ak)$ is achievable using zero-error $\repic$ schemes.
\end{lemma}
Given Corollary~\ref{cor_nocrd}, it is sufficient to show that for any $\delta > 0$, a rate of $(R + \delta, R_{\b}: \b \in \ak)$ is achievable using zero-error $\repic$ schemes if a rate of $(R + \delta, R_{\b}: \b \in \ak)$ is achievable using $\repic$ schemes.
Fix $\delta > 0$; for every $\epsilon > 0$, for a large enough block-length $n$, we show that a $(n, \epsilon, \epsilon)$-$\repic$ scheme of rate $(R, R_{\b}: \b \in \ak)$ can be transformed into a $(n, 3N\epsilon)$-zero-error $\repic$ scheme of rate $(R + \delta, R_{\b}: \b \in \ak)$.
Clearly, such a construction will imply the lemma.

The new zero-error $\repic$ scheme is a hybrid code that uses the given $\repic$ scheme and a \emph{non-private} zero-error index code.
An informal description of the construction follows.
For a given realization of files and keys, the sender checks if the transmitted message computed by the encoder of the $(n, \epsilon, \epsilon)$-$\repic$ scheme can be decoded correctly by all users (this can be done since decoders are deterministic).
In this case, the output of the encoder is transmitted after appending enough zeros to ensure that the transmission rate is $R + \delta'$(value of $\delta'$ will be decided later).
Otherwise, the sender uses a zero-error (non-private) index code of rate $R + \delta'$ instead of the original scheme.
Index code with this rate can be constructed for large enough $n$ by Theorem~\ref{thm:priv_nopriv} and the fact that asymptotic index coding rate and zero-error index coding rate are same~\cite{langberg}.
With a single bit as prefix to the transmitted message, the sender can indicate to the receivers whether to use the original scheme or the index coding scheme for decoding.
It is easy to verify that the decoding error is zero for this scheme.
The privacy leakage of the new scheme is comparable to the given scheme. This is because, loosely speaking, the new scheme leaks more information than the original scheme only when it uses the index coding scheme instead of the original scheme for encoding; but the probability of this event is small since the original scheme has low decoding error. 
By choosing $\delta'$ appropriately we can ensure that after adding the prefix bit, the rate is $(R + \delta, R_{\b}: \b \in \ak)$.
A formal proof which includes the construction of a zero-error index code that is required in the proof is given in Appendix~\ref{app_thm:private_rnd}.

It follows from the previous lemma that to prove Theorem~\ref{thm:private_rnd}, we only need to show that the rate region of zero-error $\repic$ schemes is identical to that of $\pic$ schemes.
We now state the lemma which immediately implies this identification.

\begin{lemma}\label{lem_det}
	If a rate of $(R, R_{\b}: \b \in \ak)$ is achievable using zero-error $\repic$ schemes, then for every $\delta > 0$, the rate $(R + \delta, R_{\b}: \b \in \ak)$ is achievable using $\pic$ schemes.
\end{lemma}
We prove the lemma using the following claim which constructs a $\pic$ scheme from a given zero-error $\repic$ scheme.
\begin{claim}\label{clm:in_lem_det}
	For $\epsilon > 0$, If there is a $(n, \epsilon)$-zero-error-$\repic$ scheme of rate $(R, R_{\b} : \b \in \ak)$ with encoder $\phi$ that uses private randomness $W_{\phi}$ and decoders $\psi_i$ for user $i \in [N]$, then for any $\delta > 0$, and a large enough $m$, there exists a $(mn, 2\epsilon, 2N\epsilon')$-$\pic$ scheme of rate $(R' + \delta, R_{\b} : \b \in \ak)$, where
	\begin{align*}
		\epsilon' = 4\epsilon\left(\frac{1}{2} + R + N + \sum_{\b \in \ak}R_{\b} + \frac{1}{mn} \cdot \log{\frac{1}{\epsilon}}\right) \text{ and } nR' = I\left(\phi\left(\xv_{[N]}, \kv_{\ak}, W_{\phi}\right) ; \xv_{[N]}, \kv_{\ak}\right).
	\end{align*}
\end{claim}
The lemma follows directly from the above claim since $\epsilon' \rightarrow 0$ as $\epsilon \rightarrow 0$ and 
	\begin{align*}
		R' = \frac{1}{n} \times I\left(\phi\left(\xv_{[N]}, \kv_{\ak}, W_{\phi}\right) ; \xv_{[N]}, \kv_{\ak}\right) \le \frac{1}{n} \times H\left(\phi\left(\xv_{[N]}, \kv_{\ak}, W_{\phi}\right)\right) = R.
	\end{align*}
The construction in the claim may be informally described as follows.
In the $(n, \epsilon)$-zero-error-$\repic$ scheme, the encoder $\phi$ can be thought of as the channel from $\x_{[N]}, \k_{\ak}$ to the transmitted message $\phi\left(\x_{[N]}, \k_{\ak}\right)$.
The randomness in the channel is exactly the private randomness of the encoder \emph{viz.} $W_{\phi}$.
We use results from channel simulation to \emph{approximately} compress multiple independent uses of this channel using appropriately large common randomness~\cite{cuff}.
The rate of such a compression can be made arbitrarily close to $\frac{1}{n} \times I\left(\phi\left(\xv_{[N]}, \kv_{\ak}, W_{\phi}\right) ; \xv_{[N]}, \kv_{\ak}\right)$.
As long as this approximation is close, the decoding error will be small since $\text{zero-error } \repic$ has zero decoding error.
It can also be shown that the privacy error of the approximation is not much larger than $\epsilon$.
The $\pic$ scheme is obtained by removing the common randomness introduced by the channel simulation scheme.
By Lemma~\ref{lem_no_cr} such a transformation does not increase decoding error or privacy parameter substantially.
The detailed proof of this claim is provided in the Appendix~\ref{app_thm:private_rnd}.

\paragraph*{Role of Private Randomness in Perfect Private index Coding} It is not clear whether private randomness at the encoder helps in enhancing the rate region of \emph{perfect private index codes}.
However, for linear perfect private index coding, this question is answered in the negative, in the following lemma. 
The proof of this lemma is provided in Appendix~\ref{app_thm:private_rnd}.
\begin{lemma}\label{lem:priv_linear}
	The rate region of linear perfect private index coding is not enhanced by the use of private randomness at encoder.
\end{lemma}
We leave the usefulness of private randomness in general perfect private index coding open.

	\section{Weak Privacy}
	\label{sec_1b}
Theorem~\ref{thm_key_access} shows that if the goal is to achieve the privacy required by \eqref{Eq_Dec},
then, in all but trivial cases, we need to distribute keys among the users. In this section, we consider a model 
without the extra resource of keys. In the absence of keys, we aim to achieve \emph{weak privacy}. 
An $(n,R)$ scheme consists of an encoding function $\phi$ and $N$ decoding functions $\{\psi_i\}_{i \in [N]}$. The encoding function 
\begin{align}
	\label{Eq_encod}
	\phi: \prod_{i \in [N]} \mathbb{F}^n  \longrightarrow [|\mathbb{F}|^{n\Rm}],
\end{align}
outputs the random variable $\mv = \phi\left(\xv_{[N]}\right)$. For $i \in [N]$, the decoding function
\begin{align}
	\label{Eq_decod}
	\psi_i: [|\mathbb{F}|^{n\Rm}] \times \prod_{j \in \s_i} \mathbb{F}^n  \longrightarrow  \mathbb{F}^n
\end{align}
maps the message received from the transmission and the side information data to an estimate of the file needed at user $i$
\begin{align}
	\widehat{\xv_i} := \psi_i\left(\mv, \xv_{\s_i} \right).
\end{align} 
 Rate $R$ is said to be \emph{achievable} under weak privacy, if for each $\epsilon>0$ there exists an $(n,R)$ scheme for some large enough $n$ such that the following conditions are satisfied: 
\begin{align}
\label{Eq_Dec_1b}
\pr{\psi_i\left(\mv, \xv_{\s_i}\right) = \xv_i, \forall i \in [N]} \ge 1 - \epsilon,
\end{align}
and
\begin{align}
\label{Eq_Priv_1b}
I \left( \mv;\xv_{j} \middle | \xv_{\s_i} \right) \leq n \epsilon \mbox{ for all } i \in [N], j \in [N] \setminus \a_i.
\end{align}

Observe that under weak privacy defined by~\eqref{Eq_Priv_1b}, if user $i$ does not have $X_j$ and $X_k$ as side information, where $i\neq j\neq k$, then user $i$ must not learn anything about $X_j$ or $X_k$ individually, but the user may gain some information about the pair $(X_j,X_k)$.
\subsection{Feasibility of Weak Privacy}
It is not possible to achieve weak privacy for all index coding problems. For example, Fig.~\ref{Fig_feasible} shows a feasible index coding problem
and an infeasible index coding problem under weak privacy. For the 4 user network in the left, transmitting $X_1 \oplus X_2$  and $X_3 \oplus X_4$ is a scheme under weak privacy.
For the 3 user network in the right, decodability at user 3 implies that for all $\epsilon >0$ and for large enough $n$, we have $H\left(\xv_{3}|\mv,\xv_{1},\xv_{2}\right) \leq n \epsilon$. But from the privacy condition at user 2, we have $ I\left(\mv,\xv_{1},\xv_{2};\xv_{3}\right) \leq n \epsilon$  which further implies that $H \left(\xv_{3}|\mv,\xv_{1}, \xv_{2} \right) \geq n (1-\epsilon)$. So, there is no scheme for this network.
\begin{figure}[h]
\centering
\includegraphics[scale=0.44]{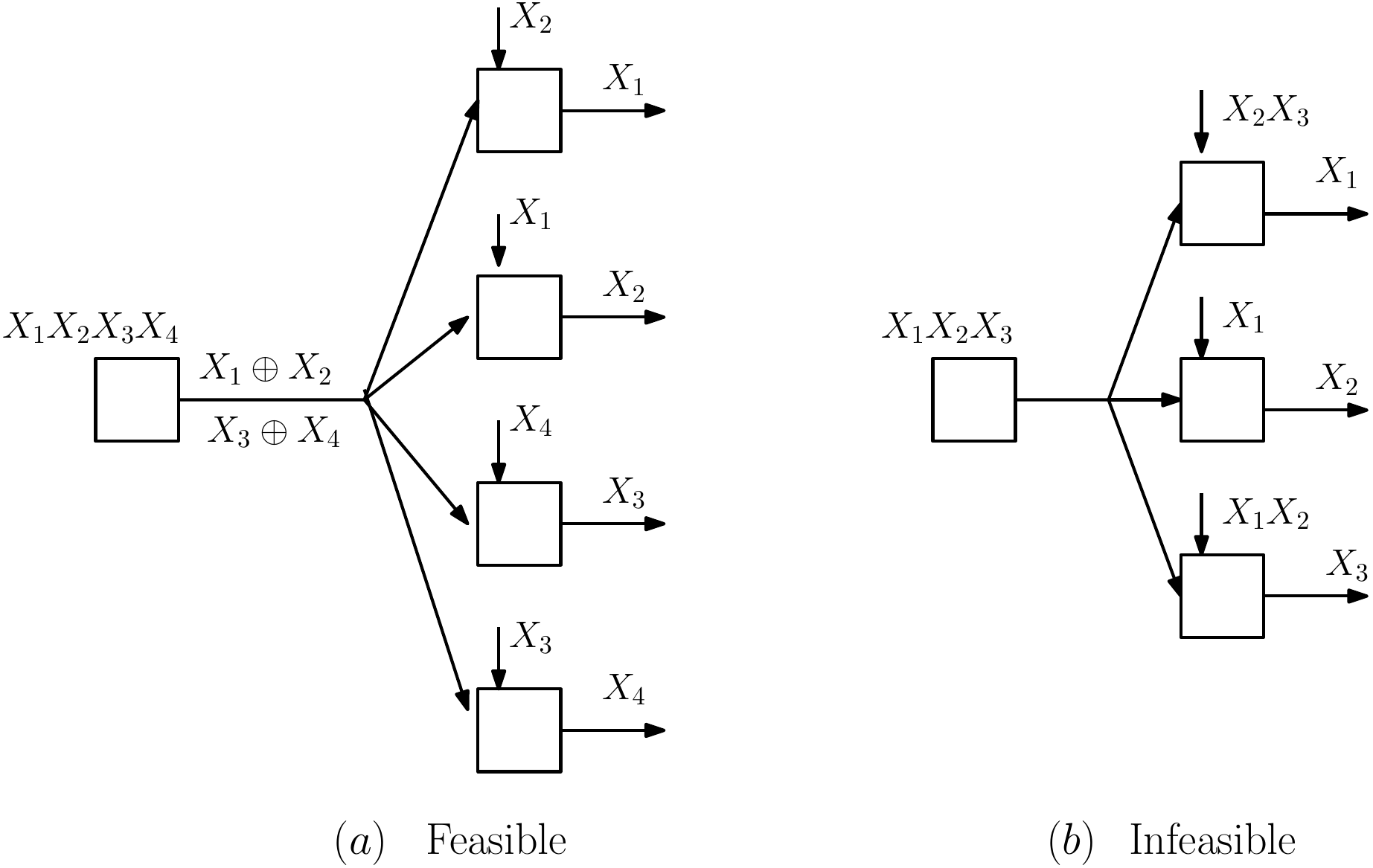}
\caption{Figure $(a)$ shows a feasible index coding instance under weak privacy along with a code that achieves weak privacy. Figure $(b)$ shows an index coding instance where the weak privacy condition at user 2 cannot be met for any code that allows all users to decode their requested file.
}
\label{Fig_feasible}
\end{figure}
Next we study the feasibility of index coding under weak privacy. We first give some necessary conditions that the network
should satisfy in order to be feasible. We start with a simple subset condition that any pair of nodes should satisfy
in order to be feasible. 
\begin{lemma}[Subset Condition]
\label{Lem_Subset}
An index coding problem under weak privacy is not feasible  if the following holds:\\
There exist users $i,j \in[N], i\neq j$ such that  $i \notin \s_j$ and $\s_i \subseteq \a_j$.
\end{lemma}
The proof of Lemma~\ref{Lem_Subset} follows from the fact that if $\s_i \subseteq \a_j$,
then user $j$ can obtain $X_i$ using the broadcast message and the side information since user $i$ obtains $X_i$ and user knows/recovers everything that user $i$ knows. Thus, weak privacy is violated if $i \notin \s_j$. The proof is given in Appendix~\ref{Sec_proof_subset}.
\begin{figure}[h]
	\centering
	\includegraphics[scale=0.5]{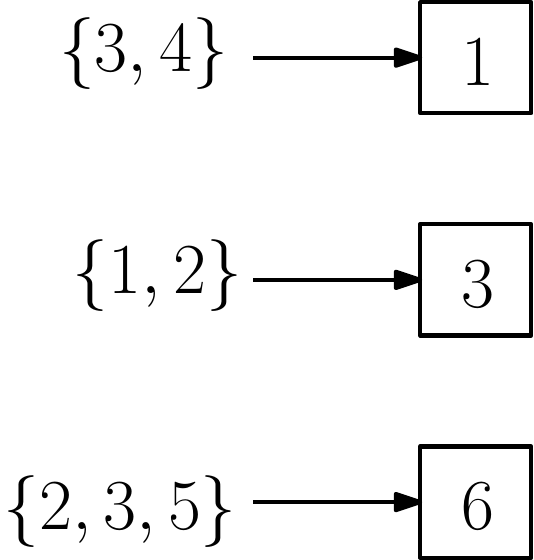}
	\caption{A private index coding instance that is infeasible although  it satisfies the condition in Lemma~\ref{Lem_Subset}.}
	\label{Fig_subset}
\end{figure}
Even if there are no two users $i,j \in[N]$ for which the condition in Lemma~\ref{Lem_Subset} is satisfied, weak privacy may not be feasible. We show this using the example depicted in Fig.~\ref{Fig_subset}. In this example, there are  $6$ users and the side information for users 1,3 and 6 are shown. Further, let $\s_i = [N]\setminus \{i\}$ for $i=2,4,5$. It is easy to verify that for any two users $i,j \in[6]$, the condition in Lemma~\ref{Lem_Subset} does not hold.
However,  in Appendix~\ref{Sec_proof_example}, we show the following.
\begin{claim}
	\label{Clm_subset}
	There is no scheme under weak privacy for the private coding instance in Fig.~\ref{Fig_subset}.
\end{claim}
Next we give an improved necessary condition for feasibility.
\begin{theorem}
	\label{prop_neces}
	If there exists a user $i \in [N]$ such that the following condition holds, then index coding under weak privacy is not feasible:\\
For any $S \subseteq \a_i$ such that $i \in S$, there exists a user $j\in [N], i \neq j$ and a $k \in S$ such that $k \notin \a_j$ and $S\setminus \{k\} \subseteq \a_j$.
\end{theorem}
Theorem~\ref{prop_neces} subsumes Lemma~\ref{Lem_Subset} as explained next. Suppose there exist $i,j \in [N]$ for which the condition in  Lemma~\ref{Lem_Subset} holds. Then for any $S \subseteq \a_i$ such that $i \in S$, it holds that for user $j$ and $k =i$, $k \notin \a_j$ and $S\setminus k \subseteq \a_j$.
The proof of Theorem~\ref{prop_neces} uses the following idea. If the condition in Theorem~\ref{prop_neces} is satisfied, then there exists a $j\in[N]$ such that $ i \notin \a_j$ which follows by taking $S=\{i\}$. Then we cannot transmit $X_i$ without any coding since it violates weak privacy at user $j$. If $X_i$ is conveyed to user $i$ by encoding $X_i$ using some side information of user $i$, then by the given condition, we can show that it violates the weak privacy at some user. The formal proof is given in Appendix~\ref{Sec_proof_prop}.

\begin{figure}[h]
	\centering
	\includegraphics[scale=0.4]{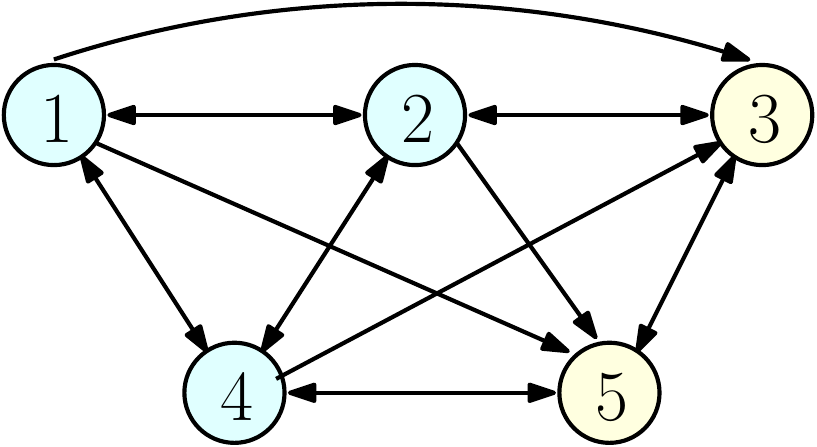}
	\caption{For the graph, $\{\{1,2,4\},\{3,5\}\}$ is a secure clique cover, i.e., $\{1,2,4\}$ and $\{3,5\}$ are cliques which together contains all the vertices such that the vertices $3$ and $5$ have only one outgoing edge to $\{1,2,4\}$, and the vertices $1,2$ and $4$ have two outgoing edges to $\{3,5\}$.}
	\label{Fig_sec_cliq}
\end{figure}

Next we give a sufficient condition to achieve weak privacy. 
Towards that, we first give the definition of a secure clique cover. A \emph{clique} is a subset of vertices where every two vertices within which are adjacent, and 
a {\em clique cover} is a set of cliques which cover all the vertices, i.e., each vertex is in at least one of the cliques.
\begin{definition}[Secure clique cover]
	\label{Def_sec_cliq}
A clique cover $\cC_{G}$ of $G$ is said to be {\em secure} if it satisfies the condition that any $v \in V(G)$ is in exactly one of the cliques in $\cC_G$ and
 for any $c \in \cC_{G}$ with $|c| =k$, there does not exist a  $v \in V(G)\setminus c$ such that $v$ has exactly $k-1$ outgoing edges to the nodes in $c$.
\end{definition}
 Observe that $\{\{1,2,4\},\{3,5\}\}$ is a clique cover for the graph shown in Fig.~\ref{Fig_sec_cliq}. It is also a secure clique cover since nodes $3$ and $5$ do not have two edges to the clique $\{1,2,4\}$, and  the nodes $1,2$ and $4$ have two edges to the clique $\{3,5\}$. Transmission of $X_1\oplus X_2 \oplus X_4$ and $ X_3\oplus X_5$ gives a scheme under weak privacy for this example. We have the following theorem.
 \begin{theorem}
  \label{Thm_Sec_Cliq}
  For the index coding problem represented by $G$,  weak privacy is feasible if $G$ has a secure clique cover.
 \end{theorem}   
 \begin{figure}[h]
	\centering
	\includegraphics[scale=0.4]{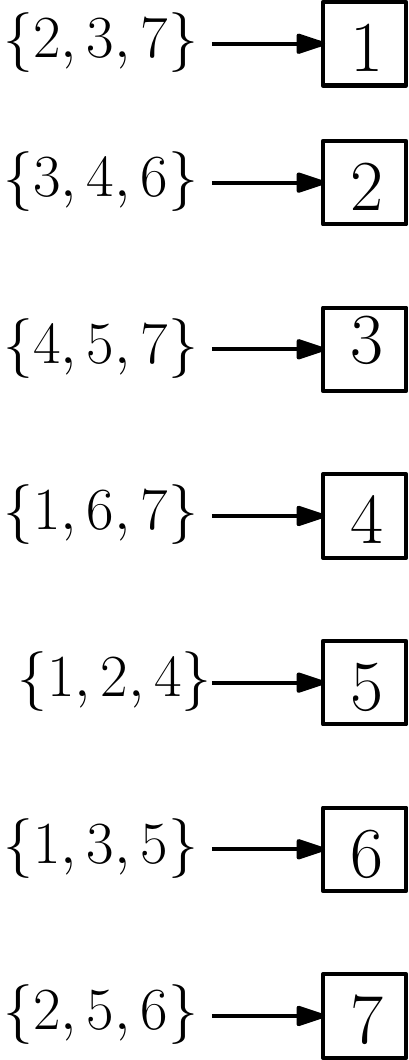}
	\caption{A private index coding instance that is feasible using  a linear scheme but it has no secure clique cover.}
	\label{Fig_linear}
\end{figure}
 The achievability scheme to show Theorem~\ref{Thm_Sec_Cliq}  is linear and the proof is given in Appendix~\ref{Proof_prop_cliq_covr}. In the next subsection, we study linear coding schemes to achieve weak privacy. A characterization of all linear schemes which achieve weak privacy is given in Theorem~\ref{Thm_1b_linear}.  We show using an example (Fig.~\ref{Fig_linear}) that having a secure clique cover  is not a necessary condition to obtain weak privacy. For this example,  we give a linear scheme which satisfies the conditions in Theorem~\ref{Thm_1b_linear}. But it can be verified that the example shown in Fig.~\ref{Fig_linear} has no secure clique cover. 
 Details are given in Appendix~\ref{Sec_linear_schm}.
\subsection{Perfect Linear Index Coding under Weak Privacy}
We consider linear coding schemes under weak privacy which satisfies \eqref{Eq_Dec_1b} and \eqref{Eq_Priv_1b} with $\epsilon =0$. Then, linear encoding for weak privacy is similar to \eqref{Eq_lin_encd} with only the first term involving $\bG_i$s.
In Theorem~\ref{Thm_1b_linear}, we characterize the linear schemes that satisfy the perfect decoding and weak privacy conditions. We use the same notations that we used to describe linear coding in Section~\ref{sec_linear}.
The proof of Theorem~\ref{Thm_1b_linear} is along similar lines as that of Theorem~\ref{Thm_linear}, hence it is omitted.
\begin{theorem}
\label{Thm_1b_linear}
The matrices $(\bG_i)_{i \in [N]}$ is a valid encoding scheme under weak privacy if and only if they satisfy the following conditions for each $i \in [N]$,
\begin{enumerate}
\item $g_i^{k} \notin \langle \{\bG_j\}_{j \notin \a_i} \rangle, \mbox{ for }1 \le k \le n,$
\item For $j \notin \a_i, \; \langle \bG_j \rangle \subseteq \langle \{\bG_k\}_{k \notin \a_i} \rangle$.
\end{enumerate}
\end{theorem}

	\section{Privacy Through Multicasts}
	\label{sec_multicast}

We consider a model in which there is no shared key between the server and the users. 
However, the server can multicast to any subset of users.
A multicast session is defined as transmitting one element from field  $\mathbb{F}$ to a subset of users. A  multicast scheme $\cM_K$  with $K$ sessions consists of $K$ subsets $\cS_k \subseteq  [N], k=1,\ldots,K$, an encoder for each session:
\begin{align}
\phi^{(k)}: \prod_{i \in [N]} \mathbb{F}^n \longrightarrow \mathbb{F}
\end{align}
and a decoder for each user $i$
\begin{align}
\psi_i: \prod_{ i \in \cS_k} \mathbb{F} \times  \prod_{j \in \s_i} \mathbb{F}^n \longrightarrow \mathbb{F}^n.
\end{align}
In the $k^{\text{th}}$ session, the server multicasts the message $M_k= \phi^{(k)}\left(\xv_{[N]}\right)$ to the subset $\cS_k$ of users.
Let $M(i)$ denote the set of messages that user $i$ gets from $K$ multicast sessions, i.e.,
\begin{align*}
M(i) & = \{ M_k: \; i \in \cS_k \}.
\end{align*}

User $i$ decodes $\widehat{\xv_i} := \psi_i\left(\mv(i), \xv_{\s_i} \right)$. A multicast scheme is said to be a perfect private multicast index code if
\begin{align*}
\pr{\psi_i\left( M(i), \xv_{\s_i}\right) = \xv_i, \forall i \in [N]} = 1,
\end{align*}
and for each user $i \in [N]$,
\begin{align*}
I \left( M(i) ;\xv_{[N] \setminus \a_i} \middle | \xv_{\s_i} \right) = 0.
\end{align*}

For a multicast scheme, we are interested in determining the minimum number of multicast sessions required.
Let $\kappa_n(G)$ denote the minimum number of multicast sessions required at block length $n$.
We define $\kappa(G) \triangleq \inf_{n} \frac{\kappa_n(G)}{n}$. 
A characterization of $\kappa(G) $ is given in the next theorem  which shows that  
for the index coding problem represented by $G$, 
$\kappa(G) $ is given by the fractional chromatic number (Definition~\ref{Def_chrom_num}) of $G^c$.

\begin{theorem}
\label{Thm_multicast}
For the index coding problem represented by $G$, $\kappa(G) = \chi_f(G^c)$.	
\end{theorem}

To prove Theorem~\ref{Thm_multicast}, we first show that
 from any multicast scheme given for block length $n$, we can obtain an $n$-fold coloring (Definition~\ref{Def_bfold}) of $G^c$. Further, we also show that
a multicast scheme for block length $n$ can be obtained  from any $n$-fold coloring of $G^c$. Formal proof is  given in Appendix~\ref{Proof_multicast}.
	\comment{
	We would also like to note that, given a multicast scheme, there is a natural perfect private index coding scheme in which, corresponding to $i^{\text{th}}$ multicast session that sends $M_i \in \mathbb{F}$ to a subset of users $\cS_i \subseteq [N]$, there is a secret key $\k_{\b_i}$ that is shared by the server and the users in $\cS_i$ and a server transmission $M_i + \k_{\b_i}$.
}

	\if \archive 1
	\begin{appendices}

\section{Proof of Theorem \ref{thm_key_access}} \label{app_thm_key_access}
The following lemma is used several places in  the following. The proof of the lemma follows from the independence of the messages and keys.
\newtheorem*{lem_ind}{Lemma \ref{lem_ind}}
\begin{lem_ind}
    Let $\cS_1, \cS_2, \cS_3, \cS_4$ be disjoint subsets of  $[N]$ and $\cT_1, \cT_2, \cT_3, \cT_4$ be disjoint subsets of $\{0, 1\}^N \setminus \{\vec{1}, \vec{0}\}$. Then,
    \begin{align*}
    I \left( \mv ; \xv_{\cS_1 \cup \cS_2}, \kv_{\cT_1 \cup \cT_2} \middle | \xv_{\cS_3 \cup \cS_4}, \kv_{\cT_3 \cup \cT_4} \right) \ge I \left( \mv ; \xv_{\cS_1}, \kv_{\cT_1} \middle | \xv_{\cS_3}, \kv_{\cT_3} \right). 
    \end{align*}
\end{lem_ind}
\begin{proof}
	    By assumption, the sets $\cS_1, \cS_2, \cS_3, \cS_4$ are disjoint among themselves and so are $\cT_1, \cT_2, \cT_3, \cT_4$.
	Since messages and keys are independent of other messages and keys, it follows that $\left( \xv_{\cS_1}, \kv_{\cT_1} \right), \left(\xv_{\cS_2 \cup \cS_4}, \kv_{\cT_2 \cup \cT_4}\right)$ and $\left( \xv_{\cS_3}, \kv_{\cT_3} \right)$ are independent. So, we have
	\begin{align}
	I \left( \xv_{\cS_2 \cup \cS_4}, \kv_{\cT_2 \cup \cT_4} ; \xv_{\cS_1}, \kv_{\cT_1} \middle | \xv_{\cS_3}, \kv_{\cT_3} \right) &= 0. \label{Eq_keys_messgae_indpnt} 
	\end{align}
	Then, we obtain
    \begin{align}
    I \left( \mv ; \xv_{\cS_1 \cup \cS_2}, \kv_{\cT_1 \cup \cT_2} \middle | \xv_{\cS_3 \cup \cS_4}, \kv_{\cT_3 \cup \cT_4} \right) & \ge I \left( \mv ; \xv_{\cS_1}, \kv_{\cT_1} \middle | \xv_{\cS_2 \cup \cS_3 \cup \cS_4}, \kv_{\cT_2 \cup \cT_3 \cup \cT_4} \right)\notag\\
     & = I \left( \mv ; \xv_{\cS_1}, \kv_{\cT_1} \middle | \xv_{\cS_2 \cup \cS_3 \cup \cS_4}, \kv_{\cT_2 \cup \cT_3 \cup \cT_4} \right) \nonumber \\
    & \; \; \; \; \; \; \; \; \; \; \; \;+ I \left( \xv_{\cS_2 \cup \cS_4}, \kv_{\cT_2 \cup \cT_4} ; \xv_{\cS_1}, \kv_{\cT_1} \middle | \xv_{\cS_3}, \kv_{\cT_3} \right) \label{Eq_ind2}\\
    & = I \left( \mv, \xv_{\cS_2 \cup \cS_4}, \kv_{\cT_2 \cup \cT_4} ; \xv_{\cS_1}, \kv_{\cT_1} \middle | \xv_{\cS_3}, \kv_{\cT_3} \right) \nonumber \\
    & \ge I \left( \mv; \xv_{\cS_1}, \kv_{\cT_1} \middle | \xv_{\cS_3}, \kv_{\cT_3} \right)\label{Eq_ind3}
    \end{align}
    where in~\eqref{Eq_ind2} we used~\eqref{Eq_keys_messgae_indpnt}.
    Thus, we have the lemma.
\end{proof}
We use the following two lemmas in proving Theorem~\ref{thm_key_access}.
\newtheorem*{lem_enc_dec}{Lemma \ref{lem_enc_dec}}
\begin{lem_enc_dec}
    For block-length $n$, $\epsilon > 0$ and $i \in [N]$, if
    \begin{align}
        I \left( \mv; \xv_{[N] \setminus \a_i} \middle |\xv_{\s_i}, \kv_{\ak_i} \right) & \le n\epsilon,\label{Eq_priv_cond1}\\
        \pr{\psi_i\left(\mv, \xv_{\s_i}, \kv_{\ak_i}\right) \neq \xv_i} & \le \epsilon, \label{Eq_decod_cond}
    \end{align}
    then, when $h(.)$ denotes the Boolean entropy function, \emph{i.e.,} $h(\epsilon) = \epsilon \log_{|\mathbb{F}|}{\frac{1}{\epsilon}} + (1 - \epsilon) \log_{|\mathbb{F}|}{\frac{1}{1 - \epsilon}}$, for sufficiently small values of $\epsilon$,
    \begin{align}
        I \left( \mv; \xv_{[N] \setminus \a_i} \middle |\xv_{\a_i}, \kv_{\ak_i} \right) & \le 3n h(\epsilon)  \label{Eq_priv1_0}\\
        I \left( \mv ; \xv_i \middle | \xv_{\s_i}, \kv_{\ak_i} \right) &\ge H \left( \xv_i \right) - 3n h(\epsilon). \label{Eq_priv6_0}
    \end{align}
\end{lem_enc_dec}
\begin{proof}
    Applying Fano's inequality to~\eqref{Eq_decod_cond}, we get
    \begin{align}
        H \left( \xv_{i} \middle |\xv_{\s_i}, \kv_{\ak_i}, \mv \right) \le h(\epsilon) + \epsilon H\left(\xv_i\right) \le h(\epsilon) + n\epsilon. \label{eq_lem_enc_dec0}
    \end{align}
    Then, 
    \begin{align}
        I \left( \xv_{i} ; \xv_{[N] \setminus \s_i} \middle |\xv_{\s_i}, \kv_{\ak_i}, \mv \right) \le H \left( \xv_{i} \middle |\xv_{\s_i}, \kv_{\ak_i}, \mv \right) \le h(\epsilon) + n\epsilon. \label{eq_lem_enc_dec1}
    \end{align}
    Hence,
    \begin{align*}
        I \left( \mv; \xv_{[N] \setminus \a_j} \middle |\xv_{\a_j}, \kv_{\ak_j} \right) \le I \left( \xv_i, \mv; \xv_{[N] \setminus \a_i} \middle |\xv_{\s_i}, \kv_{\ak_i} \right) \stackrel{(a)}{\le} h(\epsilon) + 2n\epsilon.
    \end{align*}
    Here, (a) follows from the above inequalities~\eqref{eq_lem_enc_dec1}~and~\eqref{Eq_priv_cond1}.
    Furthermore,
    \begin{align}
        I \left( \mv ; \xv_i \middle | \xv_{\s_i}, \kv_{\ak_i} \right) = H \left( \xv_i \middle | \xv_{\s_i}, \kv_{\ak_i} \right) - H \left( \xv_{i} \; | \; \xv_{\s_i}, \kv_{\ak_i}, \mv \right) \ge H \left( \xv_i \right) - h(\epsilon) - n\epsilon.\nonumber
    \end{align}
    Here, the inequality follows from the independence of $\xv_i$ and $\left( \xv_{[N] \setminus \{i\}}, \kv_{\ak_i}\right)$ and the inequality~\eqref{eq_lem_enc_dec0}.
    For sufficiently small values of $\epsilon$, the RHS of both the above inequalities can be bounded as $h(\epsilon) + 2n\epsilon \le 3n \cdot h(\epsilon)$, since $h(\epsilon) \ge \epsilon$. 
    This proves the lemma.
\end{proof}
\begin{lemma}\label{lem_ij}
    Suppose $i, j \in [N]$ such that $i \notin \a_j$. For sufficiently small $\epsilon > 0$, if the privacy condition~\eqref{Eq_Dec} and decoding conditions~\eqref{Eq_Priv} are satisfied for block-length $n$, then
    \begin{align*}
    I \left( \mv ; \k^n_{\ak_i \setminus \ak_j} \middle | \xv_{[N]}, \k^n_{\ak_i \cap \ak_j} \right) \ge H \left( \xv_i \right) - 6nh(\epsilon).
    \end{align*}
\end{lemma}
\begin{proof}
    By the privacy and decoding conditions, we have
    \begin{align*}
        I \left( \mv; \xv_{[N] \setminus \a_i} \middle |\xv_{\s_i}, \kv_{\ak_i} \right) & \le n\epsilon,\\
        \pr{\psi_i\left(\mv, \xv_{\s_i}, \kv_{\cB_i}\right) \neq \xv_i} & \le \epsilon, 
    \end{align*}
    for all $i \in [N]$. By inequality~\eqref{Eq_priv1_0} in Lemma~\ref{lem_enc_dec},
    \begin{multline}
        I \left( \mv; \xv_{[N] \setminus \a_j} \middle |\xv_{\a_j}, \kv_{\ak_j} \right) \le 3n h(\epsilon)
        \stackrel{(a)}{\implies} I \left( \mv; \xv_{i} \middle |\xv_{[N] \setminus \{i\}}, \kv_{\ak_j} \right) \le 3n h(\epsilon)\\
        \stackrel{(b)}{\implies} I \left( \mv; \xv_{i} \middle |\xv_{[N] \setminus \{i\}}, \kv_{\ak_j \cap \ak_i} \right) \le 3n h(\epsilon) \label{Eq_priv1}.
    \end{multline}
    Here (a) follows from chain rule and the assumption that $i \notin \a_j$ and (b) follows from Lemma~\ref{lem_ind}.
    By inequality~\eqref{Eq_priv6_0} in Lemma~\ref{lem_enc_dec},
    \begin{multline}
        I \left( \mv ; \xv_i \middle | \xv_{\s_i}, \kv_{\ak_i \cap \ak_j}, \kv_{\ak_i \setminus \ak_j} \right) \ge H \left( \xv_i \right) - 3n h(\epsilon)\\
        \implies I \left( \mv ; \xv_i \middle | \xv_{[N] \setminus \{i\}}, \kv_{\ak_i \cap \ak_j}, \kv_{\ak_i \setminus \ak_j} \right) \ge H \left( \xv_i \right) - 3n h(\epsilon), \label{Eq_priv6}
    \end{multline}
    where the implication follows by Lemma~\ref{lem_ind}. Then,
    \begin{align}
        I \left( \mv ; \kv_{\ak_i \setminus \ak_j} \middle | \xv_{[N]}, \kv_{\ak_i \cap \ak_j} \right)
        & \stackrel{(a)}{\ge} I \left( \mv; \xv_i \middle | \xv_{[N] \setminus \{i\}}, \kv_{\ak_i \cap \ak_j} \right) - 3n h(\epsilon) \nonumber\\
        & \qquad + I \left( \mv ; \kv_{\ak_i \setminus \ak_j} \middle | \xv_{[N]}, \kv_{\ak_i \cap \ak_j} \right) \nonumber\\
        & = I \left( \mv ; \xv_i, \kv_{\ak_i \setminus \ak_j} \middle | \xv_{[N] \setminus \{i\}}, \kv_{\ak_i \cap \ak_j} \right) - 3n h(\epsilon) \nonumber\\
        & \ge I \left( \mv; \xv_i \middle | \xv_{[N] \setminus \{i\}}, \kv_{\ak_i \cap \ak_j}, \kv_{\ak_i \setminus \ak_j} \right) - 3n h(\epsilon) \nonumber\\
        & \stackrel{(b)}{\ge} H \left( \xv_i \right) - 6n h(\epsilon) \label{eq_priv7},
    \end{align}
    where (a) follows from~\eqref{Eq_priv1} and (b) follows from ~\eqref{Eq_priv6}. 
\end{proof}
We first show the necessity of the condition in Theorem~\ref{thm_key_access}.\\
$(\implies)\\$

For a given key access structure $\ks$, for $i \in [N]$, let $\ks_i$ represent the set of keys in $\ks$ that are available at user $i$, \emph{i.e.,} $\ks_i = \{\b \in \ks : \b^{(i)} = 1\}$.
Since for every block length $n$,
\begin{align*}
    H \left(\kv_{\ks_i \setminus \ks_j} \right) & \ge I \left(\mv ; \kv_{\ks_i \setminus \ks_j} \middle | \xv_{[N]}, \kv_{\ks_i \cap \ks_j}\right)
\end{align*}
and since $i \notin \a_j$, from Lemma~\ref{lem_ij} it follows that for sufficiently small $\epsilon > 0$, there exists a block-length $n$ such that,
\begin{align*}
    H \left(\kv_{\ks_i \setminus \ks_j} \right) \ge H\left(\x^{n}_i\right) - 6 nh(\epsilon),
\end{align*}
or,
\begin{align*}
    \sum_{\b \in \ks_i \setminus \ks_j} H \left(\kv_{\b} \right) \ge H\left(\x^{n}_i\right) - 6nh(\epsilon).
\end{align*}
Since $h(\epsilon) \rightarrow 0$ as $\epsilon \rightarrow 0$, from the above observation, we get
\begin{align*}
    \sum_{\b \in \ks_i \setminus \ks_j} R_{\b} \ge H\left(\x_i\right).
\end{align*}
Since $H(\x_i) > 0$, this implies that there is a key $\kv_{\b}$ with non-zero rate that is available at user $i$ and not available at user $j$.

Next we show that if the key access structure satisfies the condition in Theorem~\ref{thm_key_access}, then we may construct a private index coding scheme.\\
$(\impliedby)$

Let $\ks$ be a key access structure that satisfies the conditions in Theorem~\ref{thm_key_access}.
We assume that the alphabet of $\x_i, i \in [N]$ is a field $\mathbb{F}$ and that of keys $\kv_{\b}, \b \in \ks$ is the vector space $\mathbb{F}^N$. We denote the $i^{th}$ component of the key $\kv_{\b}$ by $\k^i_{\b}$.
Let $e_i, i \in [N]$ be the vector in $\mathbb{F}^N$, with $1$ as the $i^{th}$ component and other values being 0. This is the $i^{th}$ column of an $N \times N$ matrix over $\mathbb{F}$.
When we set $M= \left[M_1, \ldots M_N \right]^T$, the scheme described in~\eqref{Sch_Feas} may represented as follows 
\begin{align}
    M = \sum_{i \in [N]} e_i \left ( \x_i + \sum_{\b \in \ks_i} \kv_{\b}^i \right ). \label{Eq_ach_sch1}
\end{align}
Comparing this with the general linear encoder described in~\eqref{Eq_lin_encd}, we see that $G_i = [e_i]$ for $i \in [N]$ and that $H_{\b} = [\{e_i\}_{i:b_i=1}]$ for any $\b \in \ks$. For $i \in [N]$, if $\b \notin \ks_i$ then $e_i \notin \langle H_{\b} \rangle$. From these observations, it can be seen that for any $i \in [N]$,
\begin{align*}
    \langle e_i \rangle \notin \langle \{G_j\}_{j \notin \a_i},  \{H_{\b}\}_{\b \notin \ks_i} \rangle.
\end{align*}
Hence, our scheme satisfies the first condition in Theorem~\ref{Thm_linear}.

Let $i, j \in [n]$ such that $i \notin \a_j$, then by our assumption $\ks_i \setminus \ks_j$ is non-empty.
Let $\b \in \ks_i \setminus \ks_j$, then $b_i = 1$, hence
\begin{align*}
    e_i \in \langle H_{\b} \rangle = \langle \{e_{\ell}\}_{\ell:b_{\ell}=1} \rangle.
    \shortintertext{Since $G_i = [e_i]$,}
    \langle G_i \rangle \subseteq \langle H_{\b} \rangle \subseteq \langle \{H_{\b}\}_{\b \notin \ks_j} \rangle.
\end{align*}
Thus, the second condition in Theorem~\ref{Thm_linear} is also satisfied. Hence this is a valid linear (perfect) private index code.
This proves Theorem~\ref{thm_key_access}.

\begin{table}[]
    \centering
    \begin{tabular}{|l|l|l|l|}
        \hline
        Side info. graph $G$  & Rate Region $\cR(G)$ & Vertices of $\cR(G)$  & A private index code achieving the vertex\\ \cline{1-4}
        \multirow{4}*{
            \begin{tikzpicture}
                \node (1) at (-1, 0) [draw=none]{$1$};
                \node (2) at (1, 0)  [draw=none]{$2$};
            \end{tikzpicture}
        } & \multirow{4}*{
            $\begin{array}{lcl}
                R_{10}, R_{01} \ge 0\\
                R_{10} \ge 1\\
                R_{01} \ge 1\\
                R \ge 2\\
            \end{array}$
        }
        &    $(2, 1, 1)$    & $\x_1 + \k_{10}, \x_2 + \k_{01}$ \\ 
        & & & \\ 
        & & & \\
        & & & \\ \hline
        \multirow{3}*{
            \begin{tikzpicture}
                \node (1) at (-1, 0) [draw=none]{$1$};
                \node (2) at (1, 0)  [draw=none]{$2$};
                \path[->,thick] (1) edge (2);
            \end{tikzpicture}
        } & \multirow{3}*{
            $\begin{array}{lcl}
                R_{10}, R_{01} \ge 0\\
                R_{10} \ge 1\\
                R \ge 2\\
            \end{array}$
        }
        & $(2, 1, 0)$    & $\x_1 + \k_{10}, \x_2$ \\ 
        & & & \\
        & & & \\ \hline
        \multirow{3}*{
            \begin{tikzpicture}
                \node (1) at (-1, 0) [draw=none]{$1$};
                \node (2) at (1, 0)  [draw=none]{$2$};
                \path[->,thick] (1) edge (2);
                \path[->,thick] (2) edge (1);
            \end{tikzpicture}
        } & \multirow{3}*{
            $\begin{array}{lcl}
                R_{10}, R_{01} \ge 0\\
                R \ge 1\\
            \end{array}$
        }
        & $(1, 0, 0)$    & $\x_1 + \x_2$ \\
        & & & \\ 
        & & & \\ \hline
    \end{tabular}
    \caption{Characterization of rate regions of all private index coding problems with 2 users.
    The rate region is achievable using scalar linear private index codes.
    In the table, the vertices are represented as tuples $(R, R_{10}, R_{01})$.}
    \label{Table_2user}
\end{table}
\begin{table}[]
    \centering
    \begin{tabular}{|l|l|l|p{6cm}|}
        \hline
        Side info. graph $G$  & Rate Region $\cR(G)$ & Vertices of $\cR(G)$  & A private index code achieving the vertex\\ \cline{1-4}
        \multirow{13}*{
            \begin{tikzpicture}
                \node (1) at (-1, 0) [draw=none]{$1$};
                \node (2) at (1, 0)  [draw=none]{$2$};
                \node (3) at (0, 1)  [draw=none]{$3$};
            \end{tikzpicture}
        } & \multirow{13}*{
            $\begin{array}{lcl}
                R_{\b} \ge 0, \b \in \ak \\
                R_{100} + R_{101} \ge 1\\
                R_{100} + R_{110} \ge 1\\
                R_{010} + R_{110} \ge 1\\
                R_{010} + R_{011} \ge 1\\
                R_{001} + R_{011} \ge 1\\
                R_{001} + R_{101} \ge 1\\
                R_{001} + R_{101} + R_{100} \ge 2\\
                R_{001} + R_{011} + R_{010} \ge 2\\
                R_{100} + R_{110} + R_{010} \ge 2\\
                R \ge 3\\
            \end{array}$
        }
        &    $(3, 0, 0, 2, 0, 2, 2)$    & $\x_1 + \k_{101} + \k_{110}, \x_2 + \k'_{110} + \k_{011},$ \newline $\x_3 + \k'_{011} + \k'_{101}$ \\ \cline{3-4}
        &    &    $(3, 0, 0, 2, 1, 1, 1)$    & $\x_1 + \k_{101} + \k_{110}, \x_2 + \k'_{110} + \k_{011},$ \newline $ \x_3 + \k_{001}$\\ \cline{3-4}
        &    &    $(3, 0, 1, 1, 0, 2, 1)$    & $\x_1 + \k_{101} + \k_{110}, \x_2 + \k_{010},$ \newline $ \x_3 + \k_{011} + \k'_{101}$ \\ \cline{3-4}
        &    &    $(3, 0, 1, 1, 1, 1, 0)$    & $\x_1 + \k_{101} + \k_{110}, \x_2 + \k_{010}, \x_3 + \k_{001}$ \\ \cline{3-4}
        &    &    $(3, 1, 0, 1, 0, 1, 2)$    & $\x_1 + \k_{001}, \x_2 + \k_{110} + \k_{011},$ \newline $ \x_3 + \k'_{011} + \k_{101}$ \\ \cline{3-4}
        &    &    $(3, 1, 0, 1, 1, 0, 1)$    & $\x_1 + \k_{100}, \x_2 + \k_{110} + \k_{011}, \x_3 + \k_{001}$ \\ \cline{3-4}
        &    &    $(3, 1, 1, 0, 0, 1, 1)$    & $\x_1 + \k_{100}, \x_2 + \k_{010}, \x_3 + \k_{011} + \k_{101}$ \\ \cline{3-4}
        &    &    $(3, 1, 1, 0, 1, 0, 0)$    & $\x_1 + \k_{100}, \x_2 + \k_{010}, \x_3 + \k_{001}$ \\ 
        & & & \\ \hline
        \multirow{11}*{
            \begin{tikzpicture}
                \node (1) at (-1, 0) [draw=none]{$1$};
                \node (2) at (1, 0)  [draw=none]{$2$};
                \node (3) at (0, 1)  [draw=none]{$3$};
                \path[->,thick] (1) edge (2);
            \end{tikzpicture}
        } & \multirow{11}*{
            $\begin{array}{lcl}
                R_{\b} \ge 0, \b \in \ak \\
                R_{100} + R_{101} \ge 1\\
                R_{100} + R_{110} \ge 1\\
                R_{010} + R_{110} \ge 1\\
                R_{001} + R_{011} \ge 1\\
                R_{001} + R_{101} \ge 1\\
                R_{001} + R_{101} + R_{100} \ge 2\\
                R_{100} + R_{110} + R_{010} \ge 2\\
                R \ge 3\\
            \end{array}$
        }
        & $(3, 0, 0, 2, 0, 2, 1)$    & $\x_1 + \k_{101} + \k_{110}, \x_2 + \k'_{110},$ \newline $\x_3 + \k_{011} + \k'_{101}$ \\ \cline{3-4}
        &    &    $(3, 0, 0, 2, 1, 1, 0)$    & $\x_1 + \k_{101} + \k_{110}, \x_2 + \k'_{110}, \x_3 + \k_{001}$\\ \cline{3-4}
        &    &    $(3, 0, 1, 1, 0, 2, 1)$    & $\x_1 + \k_{101} + \k_{110}, \x_2 + \k_{010},$ \newline $\x_3 + \k_{011} + \k'_{101}$ \\ \cline{3-4}
        &    &    $(3, 0, 1, 1, 1, 1, 0)$    & $\x_1 + \k_{101} + \k_{110}, \x_2 + \k_{010}, \x_3 + \k_{001}$ \\ \cline{3-4}
        &    &    $(3, 1, 0, 1, 0, 1, 1)$    & $\x_1 + \k_{001}, \x_2 + \k_{110}, \x_3 + \k_{011} + \k_{101}$ \\ \cline{3-4}
        &    &    $(3, 1, 0, 1, 1, 0, 0)$    & $\x_1 + \k_{100}, \x_2 + \k_{110}, \x_3 + \k_{001}$ \\ \cline{3-4}
        &    &    $(3, 1, 1, 0, 0, 1, 1)$    & $\x_1 + \k_{100}, \x_2 + \k_{010}, \x_3 + \k_{011} + \k_{101}$ \\ \cline{3-4}
        &    &    $(3, 1, 1, 0, 1, 0, 0)$    & $\x_1 + \k_{100}, \x_2 + \k_{010}, \x_3 + \k_{001}$ \\ 
        & & & \\ 
        \hline
        \multirow{9}*{
            \begin{tikzpicture}
                \node (1) at (-1, 0) [draw=none]{$1$};
                \node (2) at (1, 0)  [draw=none]{$2$};
                \node (3) at (0, 1)  [draw=none]{$3$};
                \path[->,thick] (1) edge (2);
                \path[->,thick] (1) edge (3);
            \end{tikzpicture}
        } & \multirow{9}*{
            $\begin{array}{lcl}
                R_{\b} \ge 0, \b \in \ak \\
                R_{100} + R_{101} \ge 1\\
                R_{100} + R_{110} \ge 1\\
                R_{010} + R_{110} \ge 1\\
                R_{001} + R_{101} \ge 1\\
                R_{001} + R_{101} + R_{100} \ge 2\\
                R_{100} + R_{110} + R_{010} \ge 2\\
                R \ge 3\\
            \end{array}$
        }
        &    $(3, 0, 0, 2, 0, 2, 0)$    & $\x_1 + \k_{101} + \k_{110}, \x_2 + \k'_{110}, \x_3 + \k'_{101}$ \\ \cline{3-4}
        &    &    $(3, 0, 0, 2, 1, 1, 0)$    & $\x_1 + \k_{101} + \k_{110}, \x_2 + \k'_{110}, \x_3 + \k_{001}$\\ \cline{3-4}
        &    &    $(3, 0, 1, 1, 0, 2, 0)$    & $\x_1 + \k_{101} + \k_{110}, \x_2 + \k_{010}, \x_3 + \k'_{101}$ \\ \cline{3-4}
        &    &    $(3, 0, 1, 1, 1, 1, 0)$    & $\x_1 + \k_{101} + \k_{110}, \x_2 + \k_{010}, \x_3 + \k_{001}$ \\ \cline{3-4}
        &    &    $(3, 1, 0, 1, 0, 1, 0)$    & $\x_1 + \k_{001}, \x_2 + \k_{110}, \x_3 + \k_{101}$ \\ \cline{3-4}
        &    &    $(3, 1, 0, 1, 1, 0, 0)$    & $\x_1 + \k_{100}, \x_2 + \k_{110}, \x_3 + \k_{001}$ \\ \cline{3-4}
        &    &    $(3, 1, 1, 0, 0, 1, 0)$    & $\x_1 + \k_{100}, \x_2 + \k_{010}, \x_3 + \k_{101}$ \\ \cline{3-4}
        &    &    $(3, 1, 1, 0, 1, 0, 0)$    & $\x_1 + \k_{100}, \x_2 + \k_{010}, \x_3 + \k_{001}$ \\ 
        & & & \\ 
        \hline
        \multirow{8}*{
            \begin{tikzpicture}
                \node (1) at (-1, 0) [draw=none]{$1$};
                \node (2) at (1, 0)  [draw=none]{$2$};
                \node (3) at (0, 1)  [draw=none]{$3$};
                \path[->,thick] (1) edge (2);
                \path[->,thick] (2) edge (1);
            \end{tikzpicture}
        } & \multirow{8}*{
            $\begin{array}{lcl}
                R_{\b} \ge 0, \b \in \ak \\
                R_{100} + R_{110} \ge 1\\
                R_{010} + R_{110} \ge 1\\
                R_{001} + R_{011} \ge 1\\
                R_{001} + R_{101} \ge 1\\
                R \ge 2\\
                R + R_{110} \ge 3\\
            \end{array}$
        }
        &    $(2, 0, 0, 1, 0, 1, 1)$    & $\x_1 + \x_2 + \k_{110}, \x_3 + \k_{011} + \k_{101}$ \\ \cline{3-4}
        &    &    $(2, 0, 0, 1, 1, 0, 0)$    & $\x_1 + \x_2 + \k_{110}, \x_3 + \k_{001}$\\ \cline{3-4}
        &    &    $(3, 1, 1, 0, 0, 1, 1)$    & $\x_1 + \k_{100}, \x_2 + \k_{010}, \x_3 + \k_{011} + \k_{101}$ \\ \cline{3-4}
        &    &    $(3, 1, 1, 0, 1, 0, 0)$    & $\x_1 + \k_{100}, \x_2 + \k_{010}, \x_3 + \k_{001}$ \\ 
        & & & \\ 
        & & & \\ 
        & & & \\
        & & & \\
        \hline
        \multirow{9}*{
            \begin{tikzpicture}
                \node (1) at (-1, 0) [draw=none]{$1$};
                \node (2) at (1, 0)  [draw=none]{$2$};
                \node (3) at (0, 1)  [draw=none]{$3$};
                \path[->,thick] (1) edge (2);
                \path[->,thick] (2) edge (3);
            \end{tikzpicture}
        } & \multirow{9}*{
            $\begin{array}{lcl}
                R_{\b} \ge 0, \b \in \ak \\
                R_{100} + R_{101} \ge 1\\
                R_{100} + R_{110} \ge 1\\
                R_{010} + R_{110} \ge 1\\
                R_{001} + R_{011} \ge 1\\
                R_{001} + R_{101} + R_{100} \ge 2\\
                R \ge 3\\
            \end{array}$
        }
        &    $(3, 0, 0, 2, 0, 1, 1)$    & $\x_1 + \k_{101} + \k_{110}, \x_2 + \k'_{110}, \x_3 + \k_{011}$ \\ \cline{3-4}
        &    &    $(3, 0, 0, 2, 1, 1, 0)$    & $\x_1 + \k_{101} + \k_{110}, \x_2 + \k'_{110}, \x_3 + \k_{001}$\\ \cline{3-4}
        &    &    $(3, 0, 1, 1, 0, 1, 1)$    & $\x_1 + \k_{101} + \k_{110}, \x_2 + \k_{010}, \x_3 + \k_{011}$ \\ \cline{3-4}
        &    &    $(3, 0, 1, 1, 1, 1, 0)$    & $\x_1 + \k_{101} + \k_{110}, \x_2 + \k_{010}, \x_3 + \k_{001}$ \\ \cline{3-4}
        &    &    $(3, 1, 0, 1, 0, 0, 1)$    & $\x_1 + \k_{001}, \x_2 + \k_{110}, \x_3 + \k_{011}$ \\ \cline{3-4}
        &    &    $(3, 1, 0, 1, 1, 0, 0)$    & $\x_1 + \k_{100}, \x_2 + \k_{110}, \x_3 + \k_{001}$ \\ \cline{3-4}
        &    &    $(3, 1, 1, 0, 0, 0, 1)$    & $\x_1 + \k_{100}, \x_2 + \k_{010}, \x_3 + \k_{011}$ \\ \cline{3-4}
        &    &    $(3, 1, 1, 0, 1, 0, 0)$    & $\x_1 + \k_{100}, \x_2 + \k_{010}, \x_3 + \k_{001}$ \\ 
        & & & \\
        \hline
        \multirow{6}*{
            \begin{tikzpicture}
                \node (1) at (-1, 0) [draw=none]{$1$};
                \node (2) at (1, 0)  [draw=none]{$2$};
                \node (3) at (0, 1)  [draw=none]{$3$};
                \path[->,thick] (1) edge (2);
                \path[->,thick] (2) edge (1);
                \path[->,thick] (2) edge (3);
            \end{tikzpicture}
        } & \multirow{6}*{
            $\begin{array}{lcl}
                R_{\b} \ge 0, \b \in \ak \\
                R_{100} + R_{110} \ge 1\\
                R_{010} + R_{110} \ge 1\\
                R_{001} + R_{011} \ge 1\\
                R \ge 2\\
                R + R_{110} \ge 3\\
            \end{array}$
        }
        &    $(2, 0, 0, 1, 0, 0, 1)$    & $\x_1 + \x_2 + \k_{110}, \x_3 + \k_{011}$ \\ \cline{3-4}
        &    &    $(2, 0, 0, 1, 1, 0, 0)$    & $\x_1 + \x_2 + \k_{110}, \x_3 + \k_{001}$\\ \cline{3-4}
        &    &    $(3, 1, 1, 0, 0, 0, 1)$    & $\x_1 + \k_{100}, \x_2 + \k_{010}, \x_3 + \k_{011}$ \\ \cline{3-4}
        &    &    $(3, 1, 1, 0, 1, 0, 0)$    & $\x_1 + \k_{100}, \x_2 + \k_{010}, \x_3 + \k_{001}$ \\ 
        & & & \\ 
        & & & \\
        \hline    
    \end{tabular}
\end{table}
\begin{table}
    \begin{tabular}{|l|l|l|p{6cm}|}
    	    	\hline
                Side info. graph $G$  & Rate Region $\cR(G)$ & Vertices of $\cR(G)$  & A private index code achieving the vertex\\ \cline{1-4}
    	    	     \multirow{8}*{
    	    		\begin{tikzpicture}
    	    		\node (1) at (-1, 0) [draw=none]{$1$};
    	    		\node (2) at (1, 0)  [draw=none]{$2$};
    	    		\node (3) at (0, 1)  [draw=none]{$3$};
    	    		\path[->,thick] (1) edge (2);
    	    		\path[->,thick] (3) edge (2);
    	    		\end{tikzpicture}
    	    	} & \multirow{8}*{
    	    		$\begin{array}{lcl}
    	    		R_{\b} \ge 0, \b \in \ak \\
    	    		R_{100} + R_{101} \ge 1\\
    	    		R_{100} + R_{110} \ge 1\\
    	    		R_{001} + R_{011} \ge 1\\
    	    		R_{001} + R_{101} \ge 1\\
    	    		R_{001} + R_{101} + R_{100} \ge 2\\
    	    		R \ge 3\\
    	    		\end{array}$
    	    	}
    	    	&    $(3, 0, 0, 1, 0, 2, 1)$    & $\x_1 + \k_{101} + \k_{110}, \x_2, \x_3 + \k_{011} + \k'_{101}$ \\ \cline{3-4}
    	    	&    &    $(3, 0, 0, 1, 1, 1, 0)$    & $\x_1 + \k_{101} + \k_{110}, \x_2, \x_3 + \k_{001}$\\ \cline{3-4}
    	    	&    &    $(3, 1, 0, 0, 0, 1, 1)$    & $\x_1 + \k_{001}, \x_2, \x_3 + \k_{101} + \k_{101}$ \\ \cline{3-4}
    	    	&    &    $(3, 1, 0, 0, 1, 0, 0)$    & $\x_1 + \k_{100}, \x_2, \x_3 + \k_{001}$ \\ 
    	    	& & & \\ 
    	    	& & & \\
    	    	& & & \\
    	    	& & & \\
    	    	\hline
    	 \multirow{6}*{
    		\begin{tikzpicture}
    		\node (1) at (-1, 0) [draw=none]{$1$};
    		\node (2) at (1, 0)  [draw=none]{$2$};
    		\node (3) at (0, 1)  [draw=none]{$3$};
    		\path[->,thick] (1) edge (2);
    		\path[->,thick] (3) edge (2);
    		\path[->,thick] (1) edge (3);
    		\end{tikzpicture}
    	} & \multirow{6}*{
    		$\begin{array}{lcl}
    		R_{\b} \ge 0, \b \in \ak \\
    		R_{100} + R_{101} \ge 1\\
    		R_{100} + R_{110} \ge 1\\
    		R_{001} + R_{101} \ge 1\\
    		R_{001} + R_{101} + R_{100} \ge 2\\
    		R \ge 3\\
    		\end{array}$
    	}
    	&    $(3, 0, 0, 1, 0, 2, 0)$    & $\x_1 + \k_{101} + \k_{110}, \x_2, \x_3 + \k'_{101}$ \\ \cline{3-4}
    	&    &    $(3, 0, 0, 1, 1, 1, 0)$    & $\x_1 + \k_{101} + \k_{110}, \x_2, \x_3 + \k_{001}$\\ \cline{3-4}
    	&    &    $(3, 1, 0, 0, 0, 1, 0)$    & $\x_1 + \k_{001}, \x_2, \x_3 + \k_{101}$ \\ \cline{3-4}
    	&    &    $(3, 1, 0, 0, 1, 0, 0)$    & $\x_1 + \k_{100}, \x_2, \x_3 + \k_{001}$ \\ 
    	& & & \\ 
    	& & & \\ 
        \hline
        \multirow{5}*{
            \begin{tikzpicture}
                \node (1) at (-1, 0) [draw=none]{$1$};
                \node (3) at (1, 0)  [draw=none]{$2$};
                \node (2) at (0, 1)  [draw=none]{$3$};
                \path[->,thick] (1) edge (2);
                \path[->,thick] (3) edge (2);
                \path[->,thick] (3) edge (1);
                \path[->,thick] (1) edge (3);
            \end{tikzpicture}
        } & \multirow{5}*{
            $\begin{array}{lcl}
                R_{\b} \ge 0, \b \in \ak \\
                R_{100} + R_{110} \ge 1\\
                R_{010} + R_{110} \ge 1\\
                R \ge 2\\
                R + R_{110} \ge 3
            \end{array}$
        }
        &    $(2, 0, 0, 1, 0, 0, 0)$    & $\x_1 + \x_2 + \k_{110}, \x_3$ \\ \cline{3-4}
        &    &    $(3, 1, 1, 0, 0, 0, 0)$    & $\x_1 + \k_{100}, \x_2 + \k_{010}, \x_3$\\ 
        & & & \\ 
        & & & \\ 
        & & & \\ 
        \hline
        \multirow{6}*{
            \begin{tikzpicture}
                \node (1) at (-1, 0) [draw=none]{$1$};
                \node (2) at (1, 0)  [draw=none]{$2$};
                \node (3) at (0, 1)  [draw=none]{$3$};
                \path[->,thick] (1) edge (2);
                \path[->,thick] (2) edge (1);
                \path[->,thick] (3) edge (1);
            \end{tikzpicture}
        } & \multirow{6}*{
            $\begin{array}{lcl}
                R_{\b} \ge 0, \b \in \ak \\
                R_{010} + R_{110} \ge 1\\
                R_{001} + R_{011} \ge 1\\
                R_{001} + R_{101} \ge 1\\
                R \ge 2\\
                R + R_{110} \ge 3
            \end{array}$
        }
        &    $(2, 0, 0, 1, 0, 1, 1)$    & $\x_1 + \x_2 + \k_{110}, \x_3 + \k_{101} + \k_{011}$ \\ \cline{3-4}
        &    &    $(2, 0, 0, 1, 1, 0, 0)$    & $\x_1 + \x_2 + \k_{110}, \x_3 + \k_{001}$\\ \cline{3-4}
        &    &    $(3, 0, 1, 0, 0, 1, 1)$    & $\x_1, \x_2 + \k_{010}, \x_3 + \k_{011}$\\ \cline{3-4}
        &    &    $(3, 0, 1, 0, 1, 0, 0)$    & $\x_1, \x_2 + \k_{010}, \x_3 + \k_{001}$\\ 
        & & & \\ 
        & & & \\ 
        \hline
        \multirow{6}*{
            \begin{tikzpicture}
                \node (1) at (-1, 0) [draw=none]{$1$};
                \node (2) at (1, 0)  [draw=none]{$2$};
                \node (3) at (0, 1)  [draw=none]{$3$};
                \path[->,thick] (2) edge (3);
                \path[->,thick] (3) edge (2);
                \path[->,thick] (3) edge (1);
                \path[->,thick] (1) edge (3);
            \end{tikzpicture}
        } & \multirow{6}*{
            $\begin{array}{lcl}
                R_{\b} \ge 0, \b \in \ak \\
                R_{100} + R_{101} \ge 1\\
                R_{010} + R_{011} \ge 1\\
                R \ge 2\\
                R + R_{101} + R_{011} \ge 3
            \end{array}$
        }
        &    $(2, 0, 0, 0, 0, 1, 1)$    & $\x_1 + \k_{101}, \x_2 + \x_3 + \k_{011}$ \\ \cline{3-4}
        &    &    $(2, 1, 0, 0, 0, 0, 1)$    & $\x_1 + \k_{100}, \x_2 + \x_3 + \k_{011}$\\ \cline{3-4}
        &    &    $(2, 0, 1, 0, 0, 1, 0)$    & $\x_2 + \k_{010}, \x_1 + \x_3 + \k_{101}$\\ \cline{3-4}
        &    &    $(3, 1, 1, 0, 0, 0, 0)$    & $\x_1 + \k_{100}, \x_2 + \k_{010}, \x_3$\\ 
        & & & \\ 
        & & & \\ 
        \hline
        \multirow{9}*{
            \begin{tikzpicture}
                \node (1) at (-1, 0) [draw=none]{$1$};
                \node (2) at (1, 0)  [draw=none]{$2$};
                \node (3) at (0, 1)  [draw=none]{$3$};
                \path[->,thick] (1) edge (2);
                \path[->,thick] (2) edge (3);
                \path[->,thick] (3) edge (1);
            \end{tikzpicture}
        } & \multirow{9}*{
            $\begin{array}{lcl}
                R_{\b} \ge 0, \b \in \ak \\
                R_{100} + R_{101} \ge 1\\
                R_{010} + R_{110} \ge 1\\
                R_{001} + R_{011} \ge 1\\
                R \ge 2\\
                R + R_{110} \ge 3\\
                R + R_{101} \ge 3\\
                R + R_{011} \ge 3
            \end{array}$
        }
             &    $(2, 0, 0, 1, 0, 1, 1)$    & $\x_1 + \x_2 + \k_{110} + \k_{101},$ \newline $\x_2 + \x_3 + \k_{110} + \k_{011}$ \\ \cline{3-4}
        &    &    $(3, 0, 0, 1, 1, 1, 0)$    & $\x_1 + \k_{101}, \x_2 + \k_{110}, \x_3 + \k_{001}$\\ \cline{3-4}
        &    &    $(3, 0, 1, 0, 0, 1, 1)$    & $\x_1 + \k_{101}, \x_2 + \k_{010}, \x_3 + \k_{101}$\\ \cline{3-4}
        &    &    $(3, 0, 1, 0, 1, 1, 0)$    & $\x_1 + \k_{101}, \x_2 + \k_{010}, \x_3 + \k_{001}$\\ \cline{3-4}
        &    &    $(3, 1, 0, 1, 0, 0, 1)$    & $\x_1 + \k_{100}, \x_2 + \k_{110}, \x_3 + \k_{011}$\\ \cline{3-4}
        &    &    $(3, 1, 0, 1, 1, 0, 0)$    & $\x_1 + \k_{100}, \x_2 + \k_{110}, \x_3 + \k_{001}$\\ \cline{3-4}
        &    &    $(3, 1, 1, 0, 0, 0, 1)$    & $\x_1 + \k_{100}, \x_2 + \k_{010}, \x_3 + \k_{011}$\\ \cline{3-4}
        &    &    $(3, 1, 1, 0, 1, 0, 0)$    & $\x_1 + \k_{100}, \x_2 + \k_{010}, \x_3 + \k_{001}$\\ 
        & & & \\ 
        \hline
        \multirow{5}*{
            \begin{tikzpicture}
                \node (1) at (-1, 0) [draw=none]{$1$};
                \node (2) at (1, 0)  [draw=none]{$2$};
                \node (3) at (0, 1)  [draw=none]{$3$};
                \path[->,thick] (1) edge (2);
                \path[->,thick] (2) edge (1);
                \path[->,thick] (2) edge (3);
                \path[->,thick] (3) edge (1);
            \end{tikzpicture}
        } & \multirow{5}*{
            $\begin{array}{lcl}
                R_{\b} \ge 0, \b \in \ak \\
                R_{010} + R_{110} \ge 1\\
                R_{001} + R_{011} \ge 1\\
                R \ge 2\\
                R + R_{110} \ge 3\\
            \end{array}$
        }
             &    $(2, 0, 0, 1, 0, 0, 1)$    & $\x_1 + \x_2 + \k_{110}, \x_3 + \k_{011}$ \\ \cline{3-4}
        &    &    $(2, 0, 0, 1, 1, 0, 0)$    & $\x_1 + \x_2 + \k_{110}, \x_3 + \k_{001}$ \\ \cline{3-4}
        &    &    $(3, 0, 1, 0, 0, 0, 1)$    & $\x_1, \x_2 + \k_{010}, \x_3 + \k_{011}$\\ \cline{3-4}
        &    &    $(3, 0, 1, 0, 1, 0, 0)$    & $\x_1, \x_2 + \k_{010}, \x_3 + \k_{001}$\\ 
        & & & \\ 
        \hline
        \multirow{5}*{
            \begin{tikzpicture}
                \node (1) at (-1, 0) [draw=none]{$1$};
                \node (2) at (1, 0)  [draw=none]{$2$};
                \node (3) at (0, 1)  [draw=none]{$3$};
                \path[->,thick] (1) edge (2);
                \path[->,thick] (2) edge (1);
                \path[->,thick] (3) edge (2);
                \path[->,thick] (3) edge (1);
            \end{tikzpicture}
        } & \multirow{5}*{
            $\begin{array}{lcl}
                R_{\b} \ge 0, \b \in \ak \\
                R_{001} + R_{011} \ge 1\\
                R_{001} + R_{101} \ge 1\\
                R \ge 2\\
            \end{array}$
        }
             &    $(2, 0, 0, 0, 0, 1, 1)$    & $\x_1 + \x_2, \x_3 + \k_{101} + \k_{011}$ \\ \cline{3-4}
        &    &    $(2, 0, 0, 0, 1, 0, 0)$    & $\x_1 + \x_2, \x_3 + \k_{001}$ \\ 
        & & & \\ 
        & & & \\ 
        & & & \\ 
        \hline
        
    \end{tabular}

\end{table}

\begin{table}[]
	\centering
	\begin{tabular}{|l|l|l|p{6cm}|}
		\hline
        Side info. graph $G$  & Rate Region $\cR(G)$ & Vertices of $\cR(G)$  & A private index code achieving the vertex\\ \cline{1-4}
		\multirow{5}*{
			\begin{tikzpicture}
			\node (1) at (-1, 0) [draw=none]{$1$};
			\node (2) at (1, 0)  [draw=none]{$2$};
			\node (3) at (0, 1)  [draw=none]{$3$};
			\path[->,thick] (1) edge (2);
			\path[->,thick] (2) edge (1);
			\path[->,thick] (2) edge (3);
			\path[->,thick] (3) edge (1);
			\path[->,thick] (1) edge (3);
			\end{tikzpicture}
		} & \multirow{5}*{
			$\begin{array}{lcl}
			R_{\b} \ge 0, \b \in \ak \\
			R_{010} + R_{110} \ge 1\\
			R \ge 2\\
			\end{array}$
		}
		&    $(2, 0, 0, 1, 0, 0, 0)$    & $\x_1 + \x_2 + \k_{110}, \x_3$ \\ \cline{3-4}
		&    &    $(2, 0, 1, 0, 0, 0, 0)$    & $\x_1 + \x_3, \x_2 + \k_{010}$ \\ 
		& & & \\ 
		& & & \\ 
		& & & \\ 
		\hline
		\multirow{5}*{
			\begin{tikzpicture}
			\node (1) at (-1, 0) [draw=none]{$1$};
			\node (2) at (1, 0)  [draw=none]{$2$};
			\node (3) at (0, 1)  [draw=none]{$3$};
			\path[->,thick] (1) edge (2);
			\path[->,thick] (2) edge (1);
			\path[->,thick] (2) edge (3);
			\path[->,thick] (3) edge (2);
			\path[->,thick] (3) edge (1);
			\path[->,thick] (1) edge (3);
			\end{tikzpicture}
		} & \multirow{5}*{
			$\begin{array}{lcl}
			R_{\b} \ge 0, \b \in \ak \\
			R \ge 1\\
			\end{array}$
		}
		&    $(1, 0, 0, 0, 0, 0, 0)$    & $\x_1 + \x_2 + \x_3$ \\ 
		& & & \\ 
		& & & \\ 
		& & & \\ 
		& & & \\ 
		\hline
		  \end{tabular}
    \caption{Rate regions of all private index coding problems with 3 users.
    The rate region is achievable using scalar linear private index codes.
    In the table, the vertices are represented as tuples $(R, R_{100}, R_{010}, R_{110}, R_{001}, R_{101}, R_{011})$.}
\label{Table_3user}    
\end{table}

    \section{Proof of Theorem \ref{thm_3user}} \label{app_thm_3user}
    \label{Sec_3user_proof}
    \begin{proof}[Proof of the Theorem~\ref{thm_3user}]
        Tables~\ref{Table_2user} and~\ref{Table_3user} list all the graphs on 2 and 3 vertices respectively (up to isomorphism) and characterize their rate region.
        For each graph $G$, the second column lists the inequalities that describe the rate region, $\cR(G)$, which turns out be a polygon. 
        The vertices of the polygon thus described are listed in the third column of the table.
        Note that the achievable rate region for each of these graphs is necessarily contained in these polygons.
        One can verify that these are indeed the vertices by checking that these points satisfy all the listed inequalities and that at least 7 inequalities are satisfied with equalities.
        The fourth column presents scalar linear private index coding schemes that achieve the rate tuples represented by each of these vertices.
        It can be easily verified that these schemes satisfy the conditions in Theorem~\ref{Thm_linear} and hence are valid linear (perfect) private index codes.

        In the rest of this section we show the necessity of the inequalities given in the second column.
        From the decoding condition~\eqref{Eq_Dec} and privacy condition~\eqref{Eq_Priv}, we obtain two kinds of inequalities: ones that bound the key rates and others that bound the rate of transmission.
        First we describe the kind of inequalities that bound the key rates.
        \begin{paragraph}{Bounds on key rates}
            To describe the lower bounds on the rate region, we use two kinds of bounds on key rates.
            These bounds are general and can be used for any $N$ and are described in the following claims.
        \end{paragraph}
        \begin{claim}\label{clm:keyrate-bound1}
            If $i, j \in [N]$ are users such that $i \notin \a_j$, \emph{i.e.,} in the directed graph $G(V, E)$ corresponding to the private index coding problem, $(j , i) \notin E$, then,
            \begin{align*}
                \sum_{\b \in \ak_i \setminus \ak_j} R_{\b} \ge 1.
            \end{align*}
        \end{claim}
        \begin{proof}
            If $i \notin \a_j$, we have shown in the proof of Theorem~\ref{thm_key_access}, that
            \begin{align*}
                \sum_{\b \in \ak_i \setminus \ak_j} R_{\b} \ge H \left(\x_i \right).
            \end{align*}
            The claim follows by substituting $H(\x_i) = 1$ for all $i \in [N]$.
        \end{proof}
        \begin{claim}
            If $i, j, \ell \in [N]$ are users such that $i \notin \a_j$ and $i, j \notin \a_{\ell}$, \emph{i.e.,} in the directed graph $G(V, E)$ corresponding to the private index coding problem, $(j , i), (\ell , i), (\ell , j) \notin E$, then
            \begin{align*}
            \sum_{\b \in (\ak_i \cup \ak_j) \setminus \ak_{\ell}} R_{\b} \ge 2.
            \end{align*}
        \end{claim}
        \begin{proof}
            The proof for this claim is similar to the previous one.
            For a given $\epsilon > 0$, there exists large enough block-length $n$ such that the decoding condition and privacy conditions are satisfied, \emph{i.e.,}
            \begin{align*}
                \pr{\psi_i \left(\mv, \xv_{\s_i}, \kv_{\ak_i} \right) = \xv_i, i \in [N]} & \ge 1 - \epsilon,\\
                I \left( \mv; \xv_{[N] \setminus \a_i} \middle | \xv_{\s_i}, \kv_{\ak_i} \right) & \le n\epsilon, \text{ for all } i \in [N].
            \end{align*}
            Applying Fano's inequality to the decodability condition of user ${\ell}$, we have 
            \begin{align*}
            H \left( \xv_{{\ell}} \middle | \xv_{\s({\ell})}, \kv_{\ak_{\ell}}, \mv \right) \le h(\epsilon) + n \epsilon \le 2nh(\epsilon).
            \end{align*}
            The last inequality follows from $h(\epsilon)$ being at least $\epsilon$ for sufficiently small $\epsilon > 0$. Hence, from the above condition and the privacy condition for user $\ell$,
            \begin{multline*}
                I \left( \xv_{\ell}, \mv; \xv_{[N] \setminus \a_{\ell}} \middle | \xv_{\s_{\ell}}, \kv_{\ak_{\ell}} \right) \le I \left(\mv; \xv_{[N] \setminus \a_{\ell}} \middle | \xv_{\s_{\ell}}, \kv_{\ak_{\ell}} \right) + H \left( \xv_{{\ell}} \middle | \xv_{\s({\ell})}, \kv_{\ak_{\ell}}, \mv \right) \\ \le h(\epsilon) + 2nh(\epsilon) \le 3nh(\epsilon)
            \end{multline*}
            Hence,
            \begin{multline}
                I \left( \mv; \xv_{[N] \setminus \a_{\ell}} \middle | \xv_{\a_{\ell}}, \kv_{\ak_{\ell}} \right) \le 3n h(\epsilon)
                \stackrel{(a)}{\implies} I \left( \mv; \xv_{i}, \xv_{j} \middle | \xv_{[N] \setminus \{i, j\}}, \kv_{\ak_{\ell}} \right) \le 3n h(\epsilon) \\
                \stackrel{(b)}{\implies} I \left( \mv; \xv_{i}, \xv_j \middle | \xv_{[N] \setminus \{i, j\}}, \kv_{\left( \ak_i \cup \ak_j \right ) \cap \ak_{\ell} } \right) \le 3n h(\epsilon) \label{Eq_KeyRate_2}.
            \end{multline}
            Here (a) follows from chain rule and our assumption that $i, j \notin \a_{\ell}$ and (b) follows from Lemma~\ref{lem_ind}.
            Applying Fano's inequality on the decodability condition of user $j$, we have 
            \begin{align*}
            H \left( \xv_{j} \middle | \xv_{\s_j}, \kv_{\ak_j}, \mv \right) \le h(\epsilon) + n \epsilon \le 2n h(\epsilon).
            \end{align*}
            But,
            \begin{align*}
                I \left( \mv; \xv_{i}, \xv_j \middle | \xv_{[N]\setminus\{i, j\}}, \kv_{ \ak_{i} \cup \ak_{j} } \right) =
                I \left( \mv; \xv_j  \middle | \xv_{[N]\setminus\{i, j\}}, \kv_{\ak_{i} \cup \ak_{j}} \right)
                + I \left( \mv; \xv_i \middle | \xv_{[N]\setminus\{i\}}, \kv_{\ak_{i} \cup \ak_{j}} \right)
            \end{align*}
            Since $i \notin \s_j$, $\s_j \subseteq [N] \setminus \{i, j\}$, 
            \begin{multline*}
                I \left( \mv; \xv_j \middle | \xv_{[N]\setminus\{i, j\}}, \kv_{\ak_{i} \cup \ak_{j}} \right) = H \left( \xv_j \middle |  \xv_{[N]\setminus\{i, j\}}, \kv_{\ak_{i} \cup \ak_{j}} \right) - H \left( \xv_j | \mv, \xv_{[N]\setminus\{i, j\}}, \kv_{\ak_{i} \cup \ak_{j}} \right)\\
                \ge H \left( \xv_j \right) - 2n h(\epsilon).
            \end{multline*}
            Where the last inequality follows from the fact that the first term is equal to $H(\xv_j)$ as $\xv_j$ is independent of $\left( \xv_{[N]\setminus\{i, j\}}, \kv_{\ak_{i} \cup \ak_{j}}\right)$ and the second term is at most $2n h(\epsilon)$ by the decodability condition at user $j$.
            Similarly, since $\s_i \subseteq [N] \setminus \{i\}$, applying Fano's inequality to decodability condition at user $i$, we have
            \begin{multline*}
                I \left(\mv; \xv_i \middle | \xv_{[N]\setminus\{i\}}, \kv_{\ak_{i} \cup \ak_{j}} \right) = H \left(\xv_i \middle |  \xv_{[N]\setminus\{i\}}, \kv_{\ak_{i} \cup \ak_{j}} \right) - H \left(\xv_i \middle | \mv, \xv_{[N]\setminus\{i\}}, \kv_{\ak_{i} \cup \ak_{j}} \right)\\
                \ge H \left(\xv_i \right) - 2n h(\epsilon).
            \end{multline*}
            Hence,
            \begin{align}
                \label{Eq_KeyRate_3}
                I \left( \mv; \xv_{i}, \xv_j \middle | \xv_{[N]\setminus\{i, j\}}, \kv_{\ak_{i} \cup \ak_{j}} \right) \ge H \left( \xv_i \right) + H \left( \xv_j \right) - 2n h(\epsilon).
            \end{align}
            Then,
            \begin{align*}
                & H \left( \kv_{ \left( \ak_i \cup \ak_j \right) \setminus \ak_{\ell}} \right) \stackrel{(a)}{=} H \left( \kv_{ \left( \ak_i \cup \ak_j \right) \setminus \ak_{\ell}} \middle | \xv_{[N]}, \kv_{ \left( \ak_{i} \cup \ak_{j} \right) \cap \ak_{\ell}} \right)
                \ge I \left( \mv ; \kv_{ \left( \ak_i \cup \ak_j \right) \setminus \ak_{\ell}} \middle | \xv_{[N]}, \kv_{ \left( \ak_{i} \cup \ak_{j} \right) \cap \ak_{\ell}} \right)\\
                \stackrel{(b)}{=} & I \left( \mv ; \xv_i, \xv_j \middle | \xv_{[N]\setminus\{i, j\}}, \kv_{ \left( \ak_{i} \cup \ak_{j} \right) \cap \ak_{\ell}} \right) + I \left( \mv ; \kv_{ \left( \ak_i \cup \ak_j \right) \setminus \ak_{\ell}} \middle | \xv_{[N]}, \kv_{ \left( \ak_{i} \cup \ak_{j} \right) \cap \ak_{\ell}} \right) - 3n h(\epsilon)\\
                = & I \left( \mv ; \kv_{ \left( \ak_i \cup \ak_j \right) \setminus \ak_{\ell}}, \xv_i, \xv_j \middle | \xv_{[N] \setminus \{i, j\}}, \kv_{ \left( \ak_{i} \cup \ak_{j} \right) \cap \ak_{\ell}} \right) - 3n h(\epsilon)\\
                = & I \left( \mv ; \kv_{ \left( \ak_i \cup \ak_j \right) \setminus \ak_{\ell}} \middle | \xv_{[N] \setminus \{i, j\}}, \kv_{ \left( \ak_{i} \cup \ak_{j} \right) \cap \ak_{\ell}} \right)
                + I \left( \mv ; \xv_i, \xv_j \middle | \xv_{[N]\setminus\{i, j\}}, \kv_{\ak_{i} \cup \ak_{j}} \right) - 3n h(\epsilon)\\
                \stackrel{(c)}{\ge} & H \left( \xv_i \right) + H \left( \xv_j \right) - 5nh(\epsilon).
            \end{align*}
            Where (a) follows from the independence of $\kv_{ \left( \ak_i \cup \ak_j \right) \setminus \ak_{\ell}}$ and $\left( \xv_{[N]}, \kv_{ \left( \ak_{i} \cup \ak_{j} \right) \cap \ak_{\ell}}\right)$, (b) follows from~\eqref{Eq_KeyRate_2} and (c) follows from~\eqref{Eq_KeyRate_3}. The result follows from this inequality by substituting $H(\xv_i) = n$ for $i \in [N]$ and taking $\epsilon$ to zero since $h(\epsilon) \rightarrow 0$ when $\epsilon \rightarrow 0$.
        \end{proof}
        \begin{paragraph}{Bounds on transmission rate}
            We describe three bounds of this kind. Although the first bound in Claim~\ref{Claim_MAIS} is applicable to any general private index coding problem, the latter three bounds apply specifically to $N = 3$.
        \end{paragraph}
        \begin{claim}\label{Claim_MAIS}
            The transmission rate of a private index coding scheme for a graph $G$ is lower bounded by the number of vertices in the \emph{maximum acyclic induced subgraph} of $G$, denoted by $\mathsf{MAIS}(G)$.
        \end{claim}
        \begin{proof}
            The index coding rate for a graph $G$ is bounded from below by $\mathsf{MAIS}(G)$~\cite{YossefBJK11}.
            The claim follows from the observation that a lower bound on the transmission rate for (non-private) index coding is also a lower bound for the transmission rate for the private index coding.
        \end{proof}
        \begin{claim}\label{clm:g_1}
            If the graph of the private index coding problem is a subgraph of $G_1$ given in Figure~\ref{G_Tran_1}, then the transmission rate is lower bounded as
            \begin{align*}
                R \ge 3 - R_{110}.
            \end{align*}
        \end{claim}
        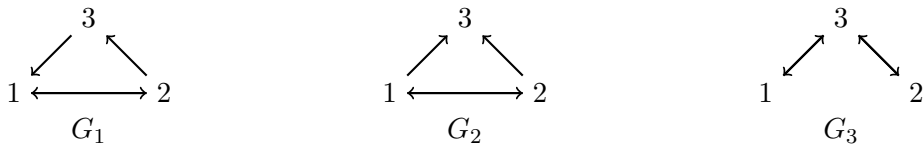
\begin{figure}
            \centering
            \begin{tikzpicture}
                \node (3) at (-1, 0) [draw=none]{$1$};
                \node (2) at (1, 0)  [draw=none]{$2$};
                \node (1) at (0, 1)  [draw=none]{$3$};
                \node (a) at (0, -0.5)  [draw=none]{$G_1$};
                \path[->,thick] (2) edge (1);
                \path[->,thick] (2) edge (3);
                \path[->,thick] (3) edge (2);
                \path[->,thick] (1) edge (3);

                \node (13) at (-1 + 5, 0) [draw=none]{$1$};
                \node (12) at (1 + 5, 0)  [draw=none]{$2$};
                \node (11) at (0 + 5, 1)  [draw=none]{$3$};
                \node (b) at (0 + 5, -0.5)  [draw=none]{$G_2$};
                \path[->,thick] (12) edge (11);
                \path[->,thick] (12) edge (13);
                \path[->,thick] (13) edge (12);
                \path[->,thick] (13) edge (11);

                \node (23) at (-1 + 10, 0) [draw=none]{$1$};
                \node (22) at (1 + 10, 0)  [draw=none]{$2$};
                \node (21) at (0 + 10, 1)  [draw=none]{$3$};
                \node (c) at (0 + 10, -0.5)  [draw=none]{$G_3$};
                \path[->,thick] (21) edge (22);
                \path[->,thick] (22) edge (21);
                \path[->,thick] (23) edge (21);
                \path[->,thick] (21) edge (23);
            \end{tikzpicture}
            \caption{Graphs $G_1, G_2$ and $G_3$ used in the statements of Claims~\ref{clm:g_1},~\ref{clm:g_2}~and~\ref{clm:g_3}, respectively.}
            \label{G_Tran_1}
        \end{figure}
        \begin{proof}
            For block length $n$, we have $nR \ge H(\mv) \ge I(\mv ; \xv_1, \xv_2, \xv_3, \kv_{001}, \kv_{010}, \kv_{100}, \kv_{011}, \kv_{101}, \kv_{110})$, but
            \begin{align}
                I \left(\mv ; \xv_1, \xv_2, \xv_3, \kv_{001}, \ldots , \kv_{110} \right) \; =\;  & I \left(\mv ; \xv_1, \xv_2, \kv_{100}, \kv_{101}, \kv_{110} \right) \nonumber\\
                & + I \left(\mv ; \xv_3, \kv_{001}, \kv_{011} \middle | \xv_1, \xv_2, \kv_{100}, \kv_{101}, \kv_{110} \right) \nonumber\\
                & + I \left(\mv ; \kv_{010} \middle | \xv_1, \xv_2, \xv_3, \kv_{100}, \kv_{101}, \kv_{110}, \kv_{001}, \kv_{011} \right).
                \label{eq_clm_4_main}
            \end{align}
            The first term may be bounded as,
            \begin{align*}
                I \left(\mv ; \xv_1, \xv_2, \kv_{100}, \kv_{101}, \kv_{110} \right) & \ge I \left(M ; \xv_1 \middle | \xv_2, \kv_{100}, \kv_{101}, \kv_{110} \right) = I \left(M ; \xv_1 \middle | \xv_{\s_{1}}, \kv_{\ak_1} \right).
            \end{align*}
            Consider the second term,
            \begin{align*}
                I \left(\mv ; \xv_3, \kv_{001}, \kv_{011} \middle | \xv_1, \xv_2, \kv_{100}, \kv_{101}, \kv_{110} \right) & \ge I \left(\mv ; \xv_3 \middle | \xv_1, \xv_2, \kv_{100}, \kv_{101}, \kv_{110}, \kv_{001}, \kv_{011} \right)\\
                & \stackrel{(a)}{\ge} I \left(\mv ; \xv_3 \middle | \xv_1, \kv_{101}, \kv_{001}, \kv_{011} \right)\\
                & = I \left(\mv ; \xv_3 \middle | \xv_{\s_{3}}, \kv_{\ak_3} \right).
            \end{align*}
            Here, (a) follows from Lemma~\ref{lem_ind}.
            The third term may be bounded as follows,
            \begin{align}
                I \left(\mv ; \kv_{010} \middle | \xv_{[3]}, \kv_{100}, \kv_{101}, \kv_{110}, \kv_{001}, \kv_{011} \right) & = I \left(\mv ; \kv_{010}, \kv_{110} \middle | \xv_{[3]}, \kv_{100}, \kv_{101}, \kv_{001}, \kv_{011} \right)  \nonumber\\
                & \qquad - I \left(\mv ; \kv_{110} \middle | \xv_1, \xv_2, \xv_3, \kv_{100}, \kv_{101}, \kv_{001}, \kv_{011} \right)\nonumber\\
                & \stackrel{(a)}{\ge} I \left(\mv ; \kv_{010}, \kv_{110} \middle | \xv_{[3]}, \kv_{011} \right)  \nonumber\\
                & \qquad - I \left(\mv ; \kv_{110} \middle | \xv_1, \xv_2, \xv_3, \kv_{100}, \kv_{101}, \kv_{001}, \kv_{011} \right)\nonumber\\
                & \stackrel{(b)}{\ge} I \left(\mv ; \kv_{\ak_2 \setminus \ak_3} \middle | \xv_{[3]}, \kv_{\ak_2 \cap \ak_3} \right) - H\left(\kv_{110}\right).
                \label{Eq_tran_2}
            \end{align}
            Where, (a) follows from Lemma~\ref{lem_ind} and (b) follows from the subtracted mutual information term being upper bounded by $H(\kv_{110})$.
            From~\eqref{eq_clm_4_main} and the bounds on each of its terms, we have
            \begin{align}
                nR \ge I \left(\mv ; \xv_1 \middle | \xv_{\s_{1}}, \kv_{\ak_1} \right) + I \left(\mv ; \xv_3 \middle | \xv_{\s_{3}}, \kv_{\ak_3} \right) + I \left(\mv ; \kv_{\ak_2 \setminus \ak_3} \middle | \xv_{[3]}, \kv_{\ak_2 \cap \ak_3} \right) - H\left(\kv_{110}\right) \label{eq:clm-4-rate}
            \end{align}
            We have already established that for sufficiently small $\epsilon > 0$, there exists a large enough block length $n$ such that
            \begin{align*}
                H \left(\xv_1 \middle | \xv_{\s_{1}}, \kv_{\ak_1}, \mv \right) & \stackrel{(a)}{\le} 2n h(\epsilon) \implies I \left(\mv ; \xv_1 \middle | \xv_{\s_{1}}, \kv_{\ak_1} \right) \ge H\left(\xv_1\right) - 2nh(\epsilon),\\
                H \left(\xv_3 \middle | \xv_{\s_{3}}, \kv_{\ak_3}, \mv \right) & \stackrel{(b)}{\le} 2n h(\epsilon)\implies I \left(\mv ; \xv_3 \middle | \xv_{\s_{3}}, \kv_{\ak_3} \right) \ge H\left(\xv_3\right) - 2nh(\epsilon),
            \end{align*}
            and
            \begin{align*}
                I \left(\mv ; \kv_{\ak_2 \setminus \ak_3} \middle | \xv_{[3]}, \kv_{\ak_2 \cap \ak_3} \right) \stackrel{(c)}{\ge} H(\xv_2) - 6nh(\epsilon).
            \end{align*}
            Where (a) (respectively (b)) can be shown using Fano's inequality and decoding condition at user 1 (respectively user 3) and (c) follows from Lemma~\ref{lem_ij} since $2 \notin \a_3$. 
            From the above observations and the inequality~\eqref{eq:clm-4-rate}, we have
            \begin{align*}
                nR \ge H\left(\xv_1\right) + H\left(\xv_3\right) + H\left(\xv_2\right) - h\left(\kv_{110}\right) - 10nh(\epsilon).
            \end{align*}
            Since for any block length $n$, $H\left(\xv_i\right)/n = 1$ for all $i \in [N]$, the claim now follows from taking $\epsilon$ to 0 .

            Consider any subgraph $H$ of $G_1$. A private index coding scheme for side information structure $H$ is a valid private index coding scheme for $G_1$ too. Hence this bound on the transmission rate applies the index coding schemes for the subgraph too. This proves the claim.
        \end{proof}
        \begin{claim}\label{clm:g_2}
            If the graph of the private index coding problem is a subgraph of $G_2$ given in Figure~\ref{G_Tran_1}, then the transmission rate is lower bounded as
            \begin{align*}
                R \ge 3 - R_{110}.
            \end{align*}
        \end{claim}
        \begin{proof}
            Similar to the previous claim, for a block length $n$, we consider the quantity $I(\mv ; \xv_{[3]}, \kv_{001}, \ldots, \kv_{110})$ which lower bounds $H(M)$.
            \begin{align}
                I \left(\mv ; \xv_1, \xv_2, \xv_3, \kv_{001}, \ldots , \kv_{110} \right) \; =\;  & I \left(\mv ; \xv_3, \kv_{001}, \kv_{011}, \kv_{101} \right)\nonumber\\
                & + I \left(\mv ; \xv_1, \xv_2, \kv_{100}, \kv_{110} \middle | \xv_3, \kv_{001}, \kv_{011}, \kv_{101} \right)\nonumber\\
                & + I \left(\mv ; \kv_{010} \middle | \xv_1, \xv_2, \xv_3, \kv_{001}, \kv_{011}, \kv_{101}, \kv_{100}, \kv_{110} \right).
                \label{eq_clm_5_main}
            \end{align}
            The first term may be bounded as follows,
            \begin{align*}
                I \left(\mv ; \xv_3, \kv_{001}, \kv_{011}, \kv_{101} \right) \ge I \left(\mv ; \xv_3 \middle | \kv_{001}, \kv_{011}, \kv_{101} \right) = I \left(\mv ; \xv_3 \middle | \xv_{\s_3}, \kv_{\ak_3} \right).
            \end{align*}
            Note that here, $\s_3 = \emptyset$.
            Consider the second term,
            \begin{align*}
                I \left(\mv ; \xv_1, \xv_2, \kv_{100}, \kv_{110} \middle | \xv_3, \kv_{001}, \kv_{011}, \kv_{101} \right) & \ge I \left(\mv ; \xv_1 \middle | \xv_2, \xv_3, \kv_{001}, \kv_{011}, \kv_{101}, \kv_{100}, \kv_{110} \right)\\
                & \stackrel{(a)}{\ge} I \left(\mv ; \xv_1 \middle | \xv_2, \xv_3, \kv_{101}, \kv_{100}, \kv_{110} \right)\\
                & = I \left(\mv ; \xv_1 \middle | \xv_{\s_1}, \kv_{\ak_1} \right).
            \end{align*}
            Here, (a) follows from Lemma~\ref{lem_ind}.
            We bound the last term identically as we did in the previous claim.
            \begin{align*}
                I \left(\mv ; \kv_{010} \middle | \xv_{[3]}, \kv_{100}, \kv_{101}, \kv_{110}, \kv_{001}, \kv_{011} \right) & = I \left(\mv ; \kv_{010}, \kv_{110} \middle | \xv_{[3]}, \kv_{100}, \kv_{101}, \kv_{001}, \kv_{011} \right)\\
                & \qquad - I \left(\mv ; \kv_{110} \middle | \xv_1, \xv_2, \xv_3, \kv_{100}, \kv_{101}, \kv_{001}, \kv_{011} \right)\nonumber\\
                & \ge I \left(\mv ; \kv_{\ak_2 \setminus \ak_3} \middle | \xv_{[3]}, \kv_{\ak_2 \cap \ak_3} \right) - H\left(\kv_{110}\right).
            \end{align*}
            From~\eqref{eq_clm_5_main} and the bounds on each of its terms, we have
            \begin{align*}
                H \left(\mv \right) \ge I \left(\mv ; \xv_3 \middle | \xv_{\s_3}, \kv_{\ak_3} \right) + I \left(\mv ; \xv_1 \middle | \xv_{\s_1}, \kv_{\ak_1} \right) + I \left(\mv ; \kv_{\ak_2 \setminus \ak_3} \middle | \xv_{[3]}, \kv_{\ak_2 \cap \ak_3} \right) - H\left(\kv_{110}\right).
            \end{align*}
            This is identical to the inequality~\ref{Eq_tran_2} in the proof of Claim~\ref{clm:g_1} since here also $2 \notin \a_3$.
            Hence, we obtain the bound
            \begin{align*}
                & R \ge 3 - R_{110}
            \end{align*}
            the same way we obtained it in the previous claim.
            Since this bound on the transmission rate applies the index coding schemes for the subgraphs too, the claim follows.

        \end{proof}
        \begin{claim}\label{clm:g_3}
            If the graph of the private index coding problem is a subgraph of $G_3$ given in Figure~\ref{G_Tran_1}, then the transmission rate 
            \begin{align*}
                R_{M} \ge 3 - R_{101} - R_{011} .
            \end{align*}
        \end{claim}
        \begin{proof}
            We have,
            \begin{align}
                H \left( \mv \right) \ge I \left( \mv ; \xv_{[3]}, \kv_{001}, \ldots , \kv_{110} \right) = & I \left( \mv ; \xv_{[3]}, \kv_{001}, \kv_{011}, \kv_{101} \right)\nonumber \\
                & + I \left( \mv ; \kv_{100}, \kv_{010}, \kv_{110} \middle | \xv_{[3]}, \kv_{001}, \kv_{011}, \kv_{101} \right).\label{eq_clm_6_main}
            \end{align}
            We may bound the first term as,
            \begin{align}
                I \left( \mv ; \xv_{[3]}, \kv_{001}, \kv_{011}, \kv_{101} \right) & \ge I \left( \mv ; \xv_3 \middle | \xv_1, \xv_2, \kv_{001}, \kv_{011}, \kv_{101} \right) = I \left( \mv ; \xv_3 \middle | \xv_{\s_3}, \kv_{\ak_3} \right).\label{eq_clm_6_dec}
            \end{align}
            The second term
            \begin{align}
                I \left( \mv ; \kv_{100}, \kv_{010}, \kv_{110} \middle | \xv_{[3]}, \kv_{001}, \kv_{011}, \kv_{101} \right) = I & \left( \mv ; \kv_{110} \middle | \xv_{[3]}, \kv_{001}, \kv_{011}, \kv_{101} \right)\nonumber \\
                & + I \left( \mv ; \kv_{010} \middle | \xv_{[3]}, \kv_{001}, \kv_{011}, \kv_{101}, \kv_{110} \right)\nonumber \\
                & + I \left( \mv ; \kv_{100} \middle | \xv_{[3]}, \kv_{001}, \kv_{011}, \kv_{101}, \kv_{110}, \kv_{010} \right).\label{eq_clm_6_key_main}
            \end{align}
            We lower bound the first term in~\eqref{eq_clm_6_key_main} by 0.
            Consider the second term.
            \begin{align}
                I \left( \mv ; \kv_{010} \middle | \xv_{[3]}, \kv_{001}, \kv_{011}, \kv_{101}, \kv_{110} \right) & \stackrel{(a)}{\ge} I \left( \mv ; \kv_{010} \middle | \xv_{[3]}, \kv_{011}, \kv_{110} \right) \nonumber \\
                & = I \left( \mv ; \kv_{010}, \kv_{011} \middle | \xv_{[3]}, \kv_{110} \right) - I \left( \mv ; \kv_{011} \middle | \xv_{[3]}, \kv_{110} \right) \nonumber \\
                & \stackrel{(b)}{\ge} I \left( \mv ; \kv_{\ak_2 \setminus \ak_1} \middle | \xv_{[3]}, \kv_{\ak_1 \cap \ak_2} \right) - H\left(\kv_{011}\right).\label{eq_clm_6_key_1}
            \end{align}
            Here, (a) follows from Lemma~\ref{lem_ind} and (b) follows from upper bounding the subtracted mutual information term by $H(\kv_{011})$.
            Similarly, third term of~\eqref{eq_clm_6_key_main} may be bounded as
            \begin{align}
                I \left( \mv ; \kv_{100} \middle | \xv_{[3]}, \kv_{001}, \kv_{011}, \kv_{101}, \kv_{110}, \kv_{010} \right) & \ge I \left( \mv ; \kv_{100} \middle | \xv_{[3]}, \kv_{101}, \kv_{110} \right) \nonumber \\
                & = I \left( \mv ; \kv_{100}, \kv_{101} \middle | \xv_{[3]}, \kv_{110} \right) - I \left( \mv ; \kv_{101} \middle | \xv_{[3]}, \kv_{110} \right) \nonumber \\
                & \stackrel{(b)}{\ge} I \left( \mv ; \kv_{\ak_1 \setminus \ak_2} \middle | \xv_{[3]}, \kv_{\ak_1 \cap \ak_2} \right) - H\left(\kv_{101}\right).\label{eq_clm_6_key_2}
            \end{align}
            Here, (a) follows from Lemma~\ref{lem_ind} and (b) follows from upper bounding the subtracted mutual information term by $H(\kv_{101})$.
            From the bounds~\eqref{eq_clm_6_dec},~\eqref{eq_clm_6_key_1},~\eqref{eq_clm_6_key_2} on the terms in the equality~\eqref{eq_clm_6_main}, we have
            \begin{align*}
                H \left( \mv \right) & \ge I \left( \mv ; \xv_3 \middle | \xv_{\s_3}, \kv_{\ak_3} \right) \\
                & \qquad + I \left( \mv ; \kv_{\ak_2 \setminus \ak_1} \middle | \xv_{[3]}, \kv_{\ak_1 \cap \ak_2} \right) - H\left(\kv_{011}\right)\\
                & \qquad + I \left( \mv ; \kv_{\ak_1 \setminus \ak_2} \middle | \xv_{[3]}, \kv_{\ak_1 \cap \ak_2} \right) - H\left(\kv_{101}\right).
            \end{align*}
            For any $\epsilon > 0$, we have already seen that there exists a large enough block-length $n$ such that,
            \begin{align*}
                I \left( \mv ; \xv_3 \middle | \xv_{\s_3}, \kv_{\ak_3} \right) \stackrel{(a)}{\ge} H\left(\xv_3\right) - 2n h(\epsilon),\\
                I \left( \mv ; \kv_{\ak_2 \setminus \ak_1} \middle | \xv_{[3]}, \kv_{\ak_1 \cap \ak_2} \right) \stackrel{(b)}{\ge} H\left(\xv_2\right) - 6n h(\epsilon),\\
                I \left( \mv ; \kv_{\ak_1 \setminus \ak_2} \middle | \xv_{[3]}, \kv_{\ak_1 \cap \ak_2} \right) \stackrel{(c)}{\ge} H\left(\xv_1\right) - 6n h(\epsilon).
            \end{align*}
            Where (a) follows from the decodability condition at user 3, (b) and (c) follows from Lemma~\ref{lem_ij} since $2 \notin \a_1$ and $1 \notin \a_2$ respectively.
            The claim can now be obtained by taking $\epsilon$ to zero.
    \end{proof}

    \section{Proof of Theorem~\ref{thm:4user}}
    \label{Sec_4user}

    \begin{figure}
      \begin{tikzpicture}
          \node (1) at (-1, 2) [draw=none]{$1$};
          \node (2) at (1, 2)  [draw=none]{$2$};
          \node (3) at (1, 0)  [draw=none]{$3$};
          \node (4) at (-1, 0)  [draw=none]{$4$};
          \path[->,thick] (1) edge (2);
          \path[->,thick] (3) edge (2);
          \path[->,thick] (4) edge (3);
          \path[->,thick] (4) edge (1);
          \path[->,thick] (2) edge (4);
          \path[->,thick] (1) edge (3);
          \path[->,thick] (3) edge (1);
        \end{tikzpicture}
        \caption{A 4 user private index coding problem for which the rate region cannot be achieved using time-sharing of scalar linear private index codes.}
        \label{G_4}
    \end{figure}
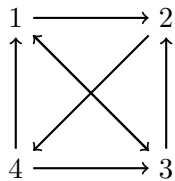
    When $N=4$, there exist private index coding problems where rate region cannot be achieved using scalar linear codes and time sharing.
    Consider the private index coding for the graph $G$ in Fig.~\ref{G_4}.
    The following vector linear scheme achieves the rate $(R, R_{1001}, R_{1010}, R_{1100}, R_{0101}, R_{0101}, R_{0011}) = (2.5, 0.5, 0.5, 0.5, 1, 0.5, 0.5)$ when every other key rate is 0
    \begin{align*}
      & \x_1^1 + \x_2^1 + \kv_{1100} + \k^1_{1001}, 
      \x_2^1 + \x_4^1 + \kv_{1100} + \k^1_{0101},\\ 
      & \x_2^2 + \x_4^2 + \kv_{0110} + \k^2_{0101},
      \x_4^2 + \x_3^1 + \k^2_{0101} + \kv_{0011},\\
      & \x_1^2 + \x_3^2 + \kv_{1010}.
    \end{align*}
    Here, $\k^i_{0101}$ is the $i$th independent bit of $\k^{(2)}_{0101}$ and $\x^i_j$ is the $i$th bit of $\x^{(2)}_j$. It can be verified that this vector linear code satisfies the conditions in Theorem~\ref{Thm_linear}.

    Next we show that this rate tuple cannot be achieved using scalar linear codes and time sharing.
        For this, it is sufficient to show that, for any scalar linear code, when
        \begin{align*}
            (R_{1001}, R_{1010}, R_{1100}, R_{0101}, R_{0110}, R_{0011}) \le (1, 1, 1, 2, 1, 1)
        \end{align*}
        and the rate of other keys are zero, the transmission rate of any scalar linear code is at least 3. To show a contradiction let $H \in \mathbb{F}_q^{2 \times 10}$ be a matrix that achieves transmission rate 2.
        \begin{align*}
            H = [v_1 \; v_2 \; v_3 \; v_4 \; v_{1001} \; v_{1010} \; v_{1100} \; v^1_{0101} \; v^2_{0101} \; v_{0110} \; v_{0011}]
        \end{align*}
        and the transmission $M$ is,
        \begin{eqnarray*}
            M = \left[ v_1 \; v_2 \; v_3 \; v_4 \; v_{1001} \; v_{1010} \; v_{1100} \; v^1_{0101} \; v^2_{0101} \; v_{0110} \; v_{0011} \right] \left[ {\scriptstyle \x_1, \x_2, \x_3, \x_4, \kv_{1001}, \kv_{1010}, \kv_{1100}, \k^1_{0101}, \k^2_{0101}. \kv_{0110}, \kv_{0011}} \right]^T
        \end{eqnarray*}
        According to Theorem~\ref{Thm_linear}, for this scheme to be a private linear index code it should satisfy the following conditions,
        \begin{align*}
            &v_1 \notin \langle v_4 \; v^1_{0101} \; v^2_{0101} \; v_{0110} \; v_{0011} \rangle,
            v_2 \notin \langle v_1 \; v_3 \; v_{1001} \; v_{1010} \; v_{0011} \rangle,\\
            &v_3 \notin \langle v_4 \; v_{1001} \; v_{1100} \; v^1_{0101} \; v^2_{0101} \rangle,
            v_4 \notin \langle v_2 \; v_{1010} \; v_{1100} \; v_{0110} \; \rangle,\\
            &\langle v_4 \rangle \in \langle v^1_{0101} \; v^2_{0101} \; v_{0110} \; v_{0011} \rangle,
            \langle v_1 \; v_3 \rangle \in \langle v_{1001} \; v_{1010} \; v_{0011} \rangle,\\
            &\langle v_4 \rangle \in \langle v_{1001} \; v_{1100} \; v^1_{0101} \; v^2_{0101} \rangle,
            \langle v_2 \rangle \in \langle v_{1010} \; v_{1100} \; v_{0110} \rangle.
        \end{align*}
        Since $v_1, \ldots, v_4$ are non-zero two 2-dimensional vectors, $v_2 \notin \langle v_1 \; v_3\rangle$ implies that $v_1 = v_3$.
        $v_2 \notin \langle v_4 \rangle$ implies that $v_2 \neq v_4$ and finally the decoding condition at user 1 implies that $v_1 \neq v_4$.
        Hence we conclude that $v_1, v_2, v_4$ are distinct and non-zero and that $v_1 = v_3$.

        Statements $v_1 \in \langle v_{1001} \; v_{1010} \; v_{0011}\rangle$ and $v_4 \notin \langle v_2 \; v_{1010}\rangle$ imply that $v_{1010}$ is either $\vec{0}$ or $v_2$.
        But, statement $v_2 \notin \langle v_1 \; v_3 \; v_{1010}\rangle$ implies that $v_{1010}$ is either $\vec{0}$ or $v_1$.
        But since $v_1$ and $v_2$ are distinct, $v_{1010} = \vec{0}$.
        $v_2 \notin \langle v_1, v_3, v_{1001}, v_{1010}, v_{0011}\rangle$, implies that $v_{1001}$ is either $\vec{0}$ or the same as $v_1 = v_3$ and that $v_{0011}$ is either $\vec{0}$ or the same as $v_1 = v_3$.
        But since $v_1 \notin \langle v_{0011} \rangle$, $v_{0011} \neq v_1$  and since $v_3 \notin \langle v_{1001} \rangle$, $v_{1001} \neq v_3$.
        Hence we have that $v_{1001} = v_{1010} = v_{0011} = \vec{0}$. This is a contradiction since $v_1$ is non-zero and $v_1 \in \langle v_{1001} \; v_{1010} \; v_{0011}\rangle$.
        \end{proof}

\section{Proof of Theorem~\ref{thm:sum_key}}\label{app_thm:sum_key}
 
 To show the upper bound in Theorem~\ref{thm:sum_key}, we use the next claim.
	\begin{claim}\label{nfold_achieve}
		If $G^c$ has an $n$-fold coloring using $C$ colors, then $G$ has a private index code with sum key rate and transmission rate $C/n$.
	\end{claim}
    \begin{proof}
    	An $n$-fold coloring of $G^c$ can be given by assigning a number from $[C]$ to each of the $n$ components of $x_i^n, i=1,\ldots,N$ such that no two adjacent vertices have the same colors. Let $l\left(x_{i}^{(k)}\right))$ denote the color given to the $k^{\text{th}}$ co-ordinate of $x_i^n$.
     For a given color  $c=1, \ldots, C$ of an $n$-fold coloring, let  $\cS_c$ denote the set of vertices where color $c$ is present, i.e.,
	\begin{align*}
	\cS_c & = \left\{i:  l\left(x_{i}^{(k)}\right) = c \text{ for some } k = 1, \ldots, n\right\}.
	\end{align*}

	Consider the key access structure such that key $K_{\b_c} \in \mathbb{F}$ is shared between all the users in $\cS_c$ for $c=1, \ldots, C$.
	 The server makes the following $C$ transmissions
	\begin{align}
	\left\{ K_{\b_c} +  X_{i}^{(k)} :  l\left(x_{i}^{(k)}\right) =c \right\},  \quad \text{ for } c=1, \ldots, C.
	\end{align}
	Using Theorem~\ref{Thm_linear}, it can be verified that this scheme satisfies the decoding and privacy constraints. 
	The sum key rate  and the transmission rate of this scheme are given by $C/n$.
    \end{proof}
	From Claim~\ref{nfold_achieve}, it follows that $\frac{\chi_n(G^c)}{n}$ is an achievable sum key rate at block length $n$.
	The upper bound on $\mathsf{SumKeyRate}^{*}(G)$ follows since $\lim\limits_{n \to \infty}  \frac{\chi_n(G^c)}{n} = \chi_f(G^c)$.

To show the lower bound in Theorem~\ref{thm:sum_key}, we first consider a specific class of side information structures in the following claim and then extend it to the general case.
\begin{claim}\label{clm:spcl_inf_struct}
	If $G$ is such that for every $i \in [N]$, there exists $j \in [N]$ such that $i \notin \a_j$, then $\vec{R} = \left(R, R_{\b}: \b \in \ak\right)$ is in the rate region $\cR(G)$ only if $\mathsf{SumKeyRate}(\vec{R}) \ge R$.
\end{claim}
\begin{proof}
	For $\epsilon > 0$, for large enough block length $n$ we know there exists a private index coding scheme such that the decoding condition~\eqref{Eq_Dec} and privacy conditions~\eqref{Eq_Priv} are satisfied.
    \begin{align*}
        nR & = H(\mv) \\
        &\stackrel{(a)}{=} I\left(\mv; \xv_{[N]}, \kv_{\ak}\right) \\
        &= I\left(\mv; \xv_{[N]}\right) + I\left(\mv; \kv_{\ak} \middle | \xv_{[N]}\right)\\
             & \le I\left(\mv; \xv_{[N]}\right) + H\left(\kv_{\ak} \middle | \xv_{[N]}\right)\\
             & \stackrel{(b)}{=} I\left(\mv; \xv_{[N]}\right) + H\left(\kv_{\ak} \right)\\
             & \le I\left(\mv; \xv_{[N]}\right) + n\mathsf{SumKeyRate}(R),
    \end{align*}
    where $(a)$ follows from the transmitted message being a deterministic function of the files and the keys, and $(b)$ follows from the independence between files and keys.
    We now show that $I\left(\mv; \xv_{[N]}\right)$ vanishes with $\epsilon$ for the side information structure considered in the claim.
    This is because if the transmitted message is not independent of the files, then it reveals some information about the files.
    By our assumption, every message has to be kept private from at least one of the users, hence this is not allowed.
    Formally,
    \begin{align}
        I\left(\mv; \xv_{[N]}\right) & = \sum_{i = 1}^N I\left(\mv; \xv_{i} \middle | \xv_1, \dots, \xv_{i-1} \right) \notag \\
        & \le \sum_{i = 1}^N I\left(\mv; \xv_{i} \middle | \xv_{[N] \setminus \{i\}} \right), \label{Eq_sum_key_low}
    \end{align}
   where the last inequality follows from the independence between files. Next we show that all $N$ terms in the summation of \eqref{Eq_sum_key_low} vanish as $n \to \infty$. To prove this, we use our assumption that for each $i \in [N]$, there exists $j$ such that $i \notin \a_j$.
   Consider the $i^{\text{th}}$ term in the summation of \eqref{Eq_sum_key_low}. By the privacy and decodability conditions at $j$,
    \begin{align*}
        I\left(\mv; \xv_{i} \middle | \xv_{[N] \setminus \{i\}} \right) & \le I\left( \mv; \xv_{i}, \xv_{[N] \setminus \a_j \cup \{i\}}\middle | \xv_{\a_j} \right)\\
        & = I \left( \mv;\xv_{[N] \setminus \a_j} \middle | \xv_{\a_j} \right)\\
        & \stackrel{(a)}{\le} I \left( \mv;\xv_{[N] \setminus \a_j} \middle | \xv_{\a_j}, \kv_{\ak_j} \right)\\
        & \le I \left( \mv, \xv_j ;\xv_{[N] \setminus \a_j} \middle | \xv_{\s_j}, \kv_{\ak_j} \right)\\
        & = I \left( \mv ; \xv_{[N] \setminus \a_j} \middle | \xv_{\s_j}, \kv_{\ak_j} \right) + I \left( \xv_j ; \xv_{[N] \setminus \a_j} \middle | \mv, \xv_{\s_j}, \kv_{\ak_j} \right)\\
        & \stackrel{(b)}{\le} n\epsilon + I \left( \xv_j ; \xv_{[N] \setminus \a_j} \middle | \mv, \xv_{\s_j}, \kv_{\ak_j} \right) \\
        & \le n\epsilon + H \left( \xv_j \middle | \mv, \xv_{\s_j}, \kv_{\ak_j} \right) \\
        &\stackrel{(c)}{\le} n\epsilon + 1 + n\epsilon,
    \end{align*}
    where $(a)$ follows from Lemma~\ref{lem_ind}, $(b)$ follows from the privacy condition~\eqref{Eq_Priv} at user $j$ and $(c)$ is obtained from the decoding condition~\eqref{Eq_Dec} by applying Fano's inequality.
\end{proof}
Consider the private index coding problem with general side information structure $G$.
Let $S $ be the subset of vertices of $G$ such that for all $i \in S$ and $j \neq i$, $i \in \s_j$, i.e., if $i \in S$, then $X_i$ is  available as side information to all users except user $i$.
Let $G'$ be the induced subgraph on the vertices $[N] \setminus S$.
The rate region for $G$ is contained in that of $G'$.
This can be argued as follows.
Any private index code for $G$ of block length $n$ can be modified into a private index code for $G'$ by setting $\xv_i = \vec{0}$ for all $i \in S$ and $\kv_{\b}^n = \vec{0}$ for all $\b$ such that $b_i = 1$ if and only if $i \in S$ (\emph{i.e.,} keys available exclusively to users in $S$).
This implies that $\mathsf{SumKeyRate}^{*}(G')$, is upper bounded by $\mathsf{SumKeyRate}^{*}(G)$.

Consider any private index coding scheme for $G'$ of block length $n$.
It can be easily verified that appending $\sum_{i \in S}\xv_i$ to the transmitted message creates a private index code for $G$.
Hence, if $R^{*}(G')$ is the optimal transmission rate for $G'$, then the optimal transmission rate for $G$, say $R^{*}(G)$ is upper bounded by $R^{*}(G') + 1$.
Hence, 
\begin{align}
R^{*}(G) - 1 \le R^{*}(G'). \label{Eq_sumkey_lowr}
\end{align}
Since $R^{*}(G') \le \mathsf{SumKeyRate}^{*}(G')$ by Claim~\ref{clm:spcl_inf_struct}, and since $ \mathsf{SumKeyRate}^{*}(G') \le \mathsf{SumKeyRate}^{*}(G)$, the lower bound on $\mathsf{SumKeyRate}^{*}(G)$ follows from~\eqref{Eq_sumkey_lowr}.
This proves Theorem~\ref{thm:sum_key}.

\subsection{Details of Example~\ref{Ex_sumkey_bounds}}

	The linear code in Example~\ref{Ex_sumkey_bounds} can also 
	be given as follows:
	\begin{align*}
	\begin{bmatrix}
	1 & 1 & 1& 0& 0& 1& 1& 0& 0\\
	0 & 1 & 1& 1& 0& 0& 1& 1& 0\\
	0 & 0 & 1& 1& 1& 0& 0& 1& 1\\
	\end{bmatrix}
	\begin{bmatrix}
	\x_1\\
	\x_2\\
	\x_3\\
	\x_4\\
	\x_5\\
	\kv_{10010}\\
	\kv_{11001}\\
	\kv_{01110}\\
	\kv_{00101}
	\end{bmatrix}
	=
	\begin{bmatrix}
	m_1\\
	m_2\\
	m_3
	\end{bmatrix}
	\end{align*}
	Here users 1, 2 and 3 decode the required data from, $\m_1, \m_2$ and $\m_3$, respectively. The decodability 
	Users 4 and 5 obtain their data since $\x_4 = \m_1 + \m_2 - \x_1$ and $\x_5 = \m_2 + \m_3 - \x_2$.
	It can be easily verified that this scheme is a perfect private linear index code as it satisfies the conditions described in Theorem~\ref{Thm_linear}.
	
	This scheme also shows that the optimal transmission rate for $G$ is at most 3.
	Next, we show that the sum key rate is strictly larger than 3 (at least $3.33$), thereby establishing that both the bounds on sum key rate could be simultaneously loose.
	
	Consider the bound on key rates in Claim~\ref{clm:keyrate-bound1} used in the proof of Theorem~\ref{thm_3user}.
	Since $4 \notin \a_1$, we get the following inequality
	\begin{align*}
	H(\kv_{\ak_4 \setminus \ak_1}) & \ge H(\x_4)
	\end{align*}
	which implies that
	\begin{align*}
	 \sum_{\b \in \ak_4 \setminus \ak_1} R_{\b} & \ge 1.
	\end{align*}
	Consider all the inequalities that can be obtained in this manner and minimize the sum key rate under these constraints, \emph{i.e.,} we solve the following LP.
	\begin{align*}
	\begin{array}{ll@{}ll}
	\text{minimize}  & \displaystyle\sum_{\b=1}^{2^5 - 2} &R_{\b} &\\
	\text{subject to}& \displaystyle\sum_{\b \in \ak_i \setminus \ak_j}   &R_{\b} \geq 1, & (i, j) \in \{(4, 1), (5, 1), (5, 2), (1, 2), (1, 3), (2, 3), (2, 4), (3, 4), (3, 5), (4, 5)\}\\
	& 													  &R_{\b} \geq 0, & \b \in \ak.
	\end{array}
	\end{align*}
	The solution of this LP, as computed by a program, comes to 3.33.
	This is strictly greater than the the minimum transmission rate for the private index code (which is at most 3).

\subsection{Details of Example~\ref{Ex_sumkey_size}}
Consider the private index coding problem with side information structure $G$, given in Figure~\ref{G_kas}.
In the sequel, we will often represent a key access structure as a family of subsets of vertices in $G$ that have access to exclusive keys. 
For $G$, there are 30 feasible key access structures of size $3$.
These are described in Table~\ref{table_kas}.
Note that, in the table, all the addition and subtraction are modulo 5.
One may verify using simple counting arguments that these are indeed the only feasible key access structures of size 3.
These key access structures are divided into 4 classes in the table.
We will show that the transmission rate allowed by the key access structures in each class is lower bound by at least 3.

\begin{table*}[t]
	\begin{tabular}{|l|l|}
		\hline
		Class 1 & $\left\{\{i\}, \{i + 1, i + 2\}, \{i - 1, i - 2\}\right\} : i \in [5]$, \\
		& $\left\{\{i, i-1\}, \{i, i + 1\}, \{i + 2, i + 3\}\right\} : i \in [5]$, \\
		& $\left\{\{i-1, i, i+1\}, \{i + 1, i + 2\}, \{i - 1, i - 2\}\right\} : i \in [5]$, \\
		\hline
		Class 2 & $\left\{\{i, i+1\}, \{i + 1, i + 2, i+3\}, \{i, i - 1, i-2\}\right\} : i \in [5]$, \\
		\hline
		Class 3 & $\left\{\{i-1, i, i+1\}, \{i + 1, i + 2, i+3\}, \{i - 1, i - 2, i-3\}\right\} : i \in [5]$, \\
		\hline
		Class 4 & $\left\{\{i, i+1, i+2\}, \{i, i-1, i-2\}, \{i-1, i-2, i+1, i+2\}\right\} : i \in [5]$, \\
		\hline
	\end{tabular}
	\caption{Feasible key access structures for the private index coding problem $G$ in Figure~\ref{G_kas} are classified into 4 classes. In the table, all addition and subtraction are modulo 5.}
	\label{table_kas}
\end{table*}

\paragraph{Class 1} It can be verified that for each key access structure $\ks$ in class 1, there exists a set $S \in \ks$ and $i, j, k \in [5]$ such that $i, j \notin S$, $j \not\equiv i \pm 1 \mod 5$, and $\ks_k = \{S\}$.
For example, consider the key access structure $\{\{i-1, i, i+1\}, \{i+1, i+2\}, \{i-1, i-2\}\}$, for some $i \in [n]$.
The key $\kv_{\{i+1, i+2\}}$ is not available to users $i, i-2$ who are not adjacent to each other ($i \not\equiv i - 2 \pm 1 \mod 5$), and it is the only key available at user $i + 2$.
We will use this property to show that the transmission rate is at least 3.

Fix $i = 5$ and consider the specific key access structure $\ks = \{\{1, 4, 5\}, \{1, 2\}, \{3, 4\}\}$.
We will show the bound on transmission rate for $\ks$.
For all the key access structures in this class, the bound can be obtained similarly.
For $\epsilon > 0$, consider a private index code of block length $n$ such that the decoding condition~\eqref{Eq_Dec} and privacy condition~\eqref{Eq_Priv} are satisfied. That is,
\begin{align}
&\pr{\widehat{\xv_i} = \xv_i, \forall i \in [5]} \ge 1 - \epsilon, \label{eq_dec_sumki}\\
&I \left( \mv;\xv_{[5] \setminus \a_i} \middle | \kv_{\ks_i}, \xv_{\s_i} \right) \le n\epsilon \mbox{ for all } i \in [5]\label{eq_priv_sumki}.
\end{align}
We have,
\begin{align*}
nR = H(\mv) \ge I \left( \mv;\xv_{[5]}, \kv_{11000}, \kv_{10011}, \kv_{00110} \right).
\end{align*}
Keeping in mind the observation that $\kv_{11000}$ is not available at users 3 and 5, we expand the mutual information as follows.
\begin{align*}
I \left( \mv;\xv_{[5]}, \kv_{11000}, \kv_{10011}, \kv_{00110} \right) =& I \left( \mv;\xv_{2}, \xv_3, \xv_4, \kv_{00110} \right) + I \left( \mv;\xv_{1}, \xv_5, \kv_{10011} \middle | \xv_2, \xv_3, \xv_4, \kv_{00110} \right)\\
 & \quad + I \left( \mv; \kv_{11000} \middle | \xv_{[5]}, \kv_{00110}, \kv_{10011} \right).
\end{align*}
This expansion helps us bound the first term using the decoding conditions at user 3 and the second term using that for user 5.
Finally, we will bound the last term using the decoding condition at user 2 who does not possess any key but $K_{11000}$.
Formally, the first term can be bound as follows.
\begin{align*}
I \left( \mv;\xv_{2}, \xv_3, \xv_4, \kv_{00110} \right) & \ge I \left( \mv;\xv_{3} \middle | \xv_2, \xv_4, \kv_{00110} \right)\\
& = H \left( \xv_{3} \middle | \xv_2, \xv_4, \kv_{00110} \right) - H \left( \xv_{3} \middle | \mv, \xv_2, \xv_4, \kv_{00110} \right)\\
& = H \left( \xv_{3} \right) - H \left( \xv_{3} \middle | \mv, \xv_2, \xv_4, \kv_{00110} \right)\\
& \stackrel{(a)}{\ge} n - 1 - n\epsilon = n(1 - \epsilon) - 1,
\end{align*}
where $(a)$ follows from the decoding condition at user 3 by applying Fano's inequality. Similarly, the second term may be bounded as follows,
\begin{align*}
I \left( \mv;\xv_{1}, \xv_5, \kv_{10011} \middle | \xv_2, \xv_3, \xv_4, \kv_{00110} \right) & \ge I \left( \mv;\xv_5 \middle | \xv_{1}, \xv_2, \xv_3, \xv_4, \kv_{00110}, \kv_{10011} \right) \\
&= H \left( \xv_{5} \right) - H \left( \xv_5 \middle | \mv, \xv_{1}, \xv_2, \xv_3, \xv_4, \kv_{00110}, \kv_{10011} \right) \\
&\stackrel{a}{\ge} n - 1 - n\epsilon = n(1 - \epsilon) - 1,
\end{align*}
where, $(a)$ is obtained directly from the decoding condition at user 5 by applying Fano's inequality.
Note that $K_{11000}$ is the only key that is available at user 2, \emph{i.e.,} $\ks_2 = \{11000\}$.
Hence the third term may be bounded as follows.
\begin{align*}
I \left( \mv; \kv_{11000} \middle | \xv_{[5]}, \kv_{00110}, \kv_{10011} \right) & = I \left( \mv; \kv_{\ks_2} \middle | \xv_{[5]}, \kv_{\ks \setminus \ks_2} \right) \\
& \stackrel{(a)}{\ge} I \left( \mv; \kv_{\ks_2} \middle | \xv_{[5]} \right)\\
& \stackrel{(b)}{\ge} n - nO(\epsilon),
\end{align*}
where $(a)$ follows from Lemma~\ref{lem_ind} and $(b)$ follows from Lemma~\ref{lem_ij}.
Taking $\epsilon$ to zero, we get the desired lower bound of 3 on the transmission rate. Note that this bound can be proved in an identical way for all instances in Class 1.

\paragraph{Class  2} Since all the key access structures in this class are isomorphic, we will show the bound for the specific key access structure $\ks = \{\{1, 2, 5\}, \{4, 5\}, \{2, 3, 4\}\} = \{11001, 00011, 01110\}$.
For $\epsilon > 0$, consider a private index code of block length $n$ such that conditions~\eqref{eq_dec_sumki} and~\eqref{eq_priv_sumki} are met.
The transmission rate can be bounded as follows.
\begin{align}
 nR &\ge H(\mv) \notag \\
& = I \left( \mv ; \xv_{[5]}, \kv_{11001}, \kv_{00011}, \kv_{01110} \right) \notag \\
& \ge I \left( \mv ; \xv_{1} \middle | \xv_2, \xv_5, \kv_{11001} \right) + I \left( \mv ;\xv_{3} \middle | \xv_1, \xv_4, \xv_2, \xv_5, \kv_{11001}, \kv_{01110} \right) \notag \\
& \quad +I ( \mv ; \kv_{00011} | \xv_{[5]}, \kv_{11001}, \kv_{01110}).\label{Eq_class1_lowr}
\end{align}
Similar to the proof for Class 1, the first and second term in the RHS of \eqref{Eq_class1_lowr} may be lower bounded by $n(1 - \epsilon) - 1$ using the decoding conditions at users 1 and 3, respectively.
Since $5 \notin \s_2$ and $\ks_5 \setminus \ks_2 = \{00011\}, \ks_5 \cap \ks_2 = \{11001\}$, the third term may be bounded as
\begin{align*}
I \left( \mv ; \kv_{00011} \middle | \xv_{[5]}, \kv_{11001}, \kv_{01110} \right) & = I \left( \mv; \kv_{\ks_5 \setminus \ks_2} \middle | \xv_{[5]}, \kv_{\ks_5 \cap \ks_2}, \kv_{01110} \right) \\
& \stackrel{(a)}{\ge} I \left( \mv; \kv_{\ks_5 \setminus \ks_2} \middle | \xv_{[5]}, \kv_{\ks_5 \cap \ks_2} \right)\\
& \stackrel{(b)}{\ge} n - nO(\epsilon),
\end{align*}
where $(a)$ follows from Lemma~\ref{lem_ind} and $(b)$ follows from Lemma~\ref{lem_ij}.
Then taking $\epsilon$ to zero, we get the desired lower bound of 3 on the transmission rate.

\paragraph{Class  3} We will prove the bound for the key access structure $\ks = \{11001, 01110, 00111\}$ since all the key access structures in this class are isomorphic to it.
For $\epsilon > 0$, we consider a private index code of block length $n$ such that conditions~\eqref{eq_dec_sumki} and~\eqref{eq_priv_sumki} are met.
The transmission rate is bounded as follows,
\begin{align*}
nR_\mv & \ge H(\mv) \\
& \ge I \left( \mv;\xv_{[5]}, \kv_{11001}, \kv_{01110}, \kv_{00111} \right)\\
& \ge I \left( \mv ; \xv_{1} \middle | \xv_2, \xv_5, \kv_{11001} \right) + I \left( \mv ; \kv_{00111} \middle | \xv_{[5]}, \kv_{11001} \right) + I \left( \mv ; \kv_{01110} \middle | \xv_{[5]}, \kv_{11001}, \kv_{00111} \right).
\end{align*}
The first term may be lower bounded by $n(1 - \epsilon) - 1$ using the decoding condition at user 1.
Since $5 \notin \s_2$ and $\ks_5 \setminus \ks_2 = \{00111\}$ and $\ks_5 \cap \ks_2 = \{11001\}$, the second term may be bounded as
\begin{align*}
I \left( \mv ; \kv_{00111} \middle | \xv_{[5]}, \kv_{11001} \right) & = I \left( \mv; \kv_{\ks_5 \setminus \ks_2} \middle | \xv_{[5]}, \kv_{\ks_5 \setminus \ks_2} \right) \\
& \stackrel{(a)}{\ge} n - nO(\epsilon),
\end{align*}
where $(a)$ follows from Lemma~\ref{lem_ij}.
Since $3 \notin \s_5$ and $\ks_3 \setminus \ks_5 = \{01110\}$ and $\ks_3 \cap \ks_5 = \{00111\}$, the third term may be bounded as
\begin{align*}
I \left( \mv ; \kv_{01110} \middle | \xv_{[5]}, \kv_{11001}, \kv_{00111} \right) & = I \left( \mv; \kv_{\ks_3 \setminus \ks_5} \middle | \xv_{[5]}, \kv_{\ks_3 \cap \ks_5}, \kv_{11001} \right) \\
& \stackrel{(a)}{\ge} I \left( \mv; \kv_{\ks_3 \setminus \ks_5} \middle | \xv_{[5]}, \kv_{\ks_3 \cap \ks_5} \right)\\ &\stackrel{(b)}{\ge} n - nO(\epsilon),
\end{align*}
where $(a)$ follows from Lemma~\ref{lem_ind} and $(b)$ follows from Lemma~\ref{lem_ij}.
Now, taking $\epsilon$ to zero, we get the desired lower bound of 3 on the transmission rate.

\paragraph{Class 4} We prove the bound for the key access structure $\ks = \{10011, 11100, 01111\}$ since all the key access structures in this class are isomorphic to it.
For $\epsilon > 0$, consider a private index code of block length $n$ such that conditions~\eqref{eq_dec_sumki} and~\eqref{eq_priv_sumki} are met.
We can bound the transmission rate as follows,
\begin{align*}
nR \ge H(\mv) & \ge I \left( \mv;\xv_{[5]}, \kv_{10011}, \kv_{11100}, \kv_{01111} \right)\\
& \ge I \left( \mv ; \xv_{5} \middle | \xv_1, \xv_4, \kv_{10011}, \kv_{01111} \right) + I \left( \mv ; \kv_{11100} \middle | \xv_1, \xv_2, \xv_4, \xv_5, \kv_{10011}, \kv_{01111} \right) \\
& \quad + I \left( \mv ; \xv_3 \middle | \xv_1, \xv_2, \xv_4, \xv_5, \kv_{\ks} \right).
\end{align*}
The first and third terms may be lower bounded by $n(1 - \epsilon) - 1$ using the decoding conditions at receivers 5 and 3, respectively.
The second term cannot be directly bounded using Lemma~\ref{lem_ij} like we bounded the third term in the previous cases.
But, using Lemma~\ref{lem_ind}, we  lower bound this term by
\begin{align*}
& I \left( \mv ; \kv_{11100} \middle | \xv_1, \xv_2, \xv_4, \xv_5, \kv_{10011}, \kv_{01111} \right) \ge I \left( \mv ; \kv_{11100} \middle | \xv_1, \xv_2, \xv_5, \kv_{10011} \right).
\end{align*}
Consider the following mutual information.
\begin{align*}
I \left( \mv ; \xv_1, \kv_{11100} \middle | \xv_2, \xv_5, \kv_{10011} \right) & = I \left( \mv ; \xv_1 \middle | \xv_2, \xv_5, \kv_{10011} \right) + I \left( \mv ; \kv_{11100} \middle | \xv_1, \xv_2, \xv_5, \kv_{10011} \right)\\
& \stackrel{(a)}{\le} I \left( \mv ; \kv_{11100} \middle | \xv_1, \xv_2, \xv_5, \kv_{10011} \right) + I \left( \mv ; \xv_1 \middle | \xv_{[5] \setminus \{1\}}, \kv_{10011}, \kv_{01111} \right)\\
& = I \left( \mv ; \kv_{11100} \middle | \xv_1, \xv_2, \xv_5, \kv_{10011} \right) + I \left( \mv ; \xv_1 \middle | \xv_{[5] \setminus \{1\}}, \kv_{\ks_4} \right)\\
& \stackrel{(b)}{\le} I \left( \mv ; \kv_{11100} \middle | \xv_1, \xv_2, \xv_5, \kv_{10011} \right) + n\epsilon,
\end{align*}
where $(a)$ follows from Lemma~\ref{lem_ind} and $(b)$ follows from privacy condition at receiver 4 since $1 \notin \s_4$. The same mutual information term can also be expanded as
\begin{align*}
I \left( \mv ; \xv_1, \kv_{11100} \middle | \xv_2, \xv_5, \kv_{10011} \right) & \ge I \left( \mv ; \xv_1 \middle | \xv_2, \xv_5, \kv_{10011}, \kv_{11100} \right)\\
& \stackrel{(a)}{=} H\left(\xv_1\right) - H \left( \xv_1 \middle | \mv, \xv_2, \xv_5, \kv_{10011}, \kv_{11100} \right)\\
& \stackrel{(b)}{\ge} n - (n\epsilon + 1)
\end{align*}
where, $(a)$ follows from the independence of files and keys and $(b)$ is obtained from decoding condition at user 1 and Fano's inequality.
From the above two observations, we have,
\begin{align*}
I \left( \mv ; \kv_{11100} \middle | \xv_1, \xv_2, \xv_4, \xv_5, \kv_{10011}, \kv_{01111} \right) & \ge I \left( \mv ; \kv_{11100} \middle | \xv_1, \xv_2, \xv_5, \kv_{10011} \right)\\
& \ge n - (2n\epsilon + 1).
\end{align*}
Taking $\epsilon$ to zero, we get the desired lower bound of 3 on the transmission rate.

\section{Proof of Theorem~\ref{thm:private_rnd}}\label{app_thm:private_rnd}

\subsection{Proof of Lemma~\ref{lem_no_cr}}
	Consider a $(n, \epsilon, \delta)$-$\crepic$ scheme with encoder $\phi$ and decoder $\psi_i$ at user $i : i \in [N]$.
	We show the existence of $w^* \in \cW$ such that conditioned $W = w^*$, the decoding error is at most $2\epsilon$ and privacy leakage is at most $2N\delta$.
	Define the events $\mathsf{Bad}^{dec}$ and $\mathsf{Bad}^{priv}_i: i \in [N]$ as follows.
	\begin{align*}
		\mathsf{Bad}^{dec} & \defeq \left\{w \in \cW : \exists i \in [N] \text{ such that } \pr{\widehat{\xv}_i \neq \xv_i | W = w} > 2\epsilon\right\},\\
		\mathsf{Bad}^{priv}_i & \defeq \left\{w \in \cW : I \left( \mv ; \xv_{[N] \setminus \a_i} \middle | \kv_{\ak_i}, \xv_{\s_i}, W = w \right) > 2nN\delta\right\}.
	\end{align*}
	Since the given scheme is $(n, \epsilon, \delta)$-$\crepic$, applying Markov's inequality to conditions~\eqref{Eq_Dec_gen} and~\eqref{Eq_Priv_gen}, respectively, we get,
	\begin{align*}
		\pr{\mathsf{Bad}^{dec}} < 1/2 \text{ and } \pr{\mathsf{Bad}^{priv}_i} < 1/2N \text{ for all } i \in [N].
	\end{align*}
	Taking a union bound of these events, it is seen that there exists $w^* \in \cW$ such that $w^* \notin \cup_{i \in [N]}\mathsf{Bad}^{priv}_i \cup \mathsf{Bad}^{dec}$.
	Define encoder $\phi'$ and decoders $\{\psi'_i\}_{i \in [N]}$ as follows
	\begin{align*}
		\phi'\left(\xv_{[N]}, \kv_{\ak}\right) & = \phi\left(\xv_{[N]}, \kv_{\ak}, W = w\right),\\ 
		\psi'_i\left(\mv, \xv_{\s_i}, \kv_{\ak_i}\right) & = \psi_i\left(\mv, \xv_{\s_i}, \kv_{\ak_i}, W = w\right).
	\end{align*}
	Clearly, the above scheme is $(n, 2\epsilon, 2N\delta)$-$\repic$.
	This proves the lemma.
\hfill{\rule{2.1mm}{2.1mm}}

\subsection{Proof of Lemma~\ref{lem_0_err}}
	By Corollary~\ref{cor_nocrd}, it is sufficient to show that if the rate $(R, R_{\b} : \b \in \ak)$ is achievable using $\repic$ schemes then for all $\delta > 0$, rate $(R + \delta, R_{\b} : \b \in \ak)$ can be achieved using zero-error $\repic$ schemes.
	Fix $\delta > 0$, $\epsilon > 0$.
	For an appropriately chosen block-length $n$, using an $(n, \epsilon, \epsilon)$-$\repic$ scheme of rate $(R, R_{\b} : \b \in \ak)$, we construct a $(n, \epsilon')$-zero-error $\repic$ scheme of rate $(R + \delta, R_{\b} : \b \in \ak)$, where $\epsilon' = (3 + N)\epsilon$.
	Clearly, such a construction proves the lemma.
	A formal description such a construction and its analysis follows.

	For $\epsilon > 0$, choose $n$ large enough such that the following properties are satisfied.
	\begin{enumerate}
		\item There exits a $(n, \epsilon, \epsilon)$-$\repic$ scheme with encoder $\phi$ and decoders $\{\psi_i\}_{i \in [N]}$.
		\item Since transmission rate $(R, R_{\b} : \b \in \ak)$ is achievable using $\repic$ schemes, by Theorem~\ref{thm:priv_nopriv}, there exists a zero-error index code of rate $R$ for large enough block-length $k$.
		\item $n \ge \max \left(\frac{3}{\delta} kR, \frac{3}{\delta}, \frac{1}{\epsilon}\right)$.
	\end{enumerate}
	Recall that when $x^n \in \cX^n$, the co-ordinate $i$ of $x^n$ is denoted by $x^{(i)}$.
	For $x^{n} \in \cX^{n}$ and $1 \le j \le \lceil \frac{n}{k}\rceil$, define
	\begin{align*}
		\mathrm{block}_j(x^{n}) \defeq \begin{cases}
		(x^{(k(j - 1) + 1)}, \ldots, x^{(kj)}) \text{ for } 1 \le j < \lceil \frac{n}{k} \rceil,\\
		(x^{(k(j - 1) + 1)}, \ldots, x^{(n)}, 0, \ldots, 0) \text{ if } j = \lceil \frac{n}{k} \rceil.
		\end{cases}
	\end{align*}
	Consider a \emph{non-private} zero-error index coding scheme of block-length $k$ and rate $R$ with encoder $\alpha$ and decoders $\beta_{i}$ for user $i \in [N]$.
	From this scheme, we construct a zero-error index coding scheme of block-length $n$ with encoder $\alpha'$ and decoders $\beta'_i$ for user $i \in [N]$ as follows.
	For $x_i^{n} \in \cX^{n} : i \in [N]$,
	\begin{align*}
		\alpha'\left(x_{[N]}^{n}\right) = (m_1, m_2, \ldots, m_{\lceil \frac{n}{k} \rceil}) \text{ where } m_j = \alpha\left(\left(\mathrm{block}_j(x_{\ell}^{n})\right)_{\ell \in [N]}\right) \text{ for } 1 \le j \le {\lceil \frac{n}{k} \rceil},\\
		\beta'_i\left(m_1, \ldots, m_{\lceil \frac{n}{k} \rceil}, x^{n}_{\s_i}\right) = (\widehat{\mathrm{block}_1(x_i^{n})}, \ldots, \widehat{\mathrm{block}_{\lceil \frac{n}{k} \rceil}(x_i^{n})}) \text{ where } \widehat{\mathrm{block}_j(x_i^n)} = \beta_i\left(m_j, \left(\mathrm{block}_j(x_{\ell}^n)\right)_{\ell \in \s_i}\right).
	\end{align*}
	The index code for block-length $n$ is the concatenation of $\lceil \frac{n}{k} \rceil$ copies of zero-error index code for block-length $k$, hence it is also a zero-error index code.
	The rate of this scheme is $\frac{k \cdot \lceil \frac{n}{k} \rceil \cdot R}{n} \le \frac{R}{n} \cdot (n + k) \le R + \delta/3$.

	The $(n, \epsilon')$ zero-error $\repic$ with encoder $\phi'$ and decoders $\psi'_i$ for user $i \in [N]$, is a hybrid of the above mentioned $(n, \epsilon, \epsilon)$-$\repic$ scheme and the $n$ block-length index coding scheme.
	Intuitively, the new scheme falls back to the zero-error index coding whenever the $\repic$ scheme is found to make an error at any of the decoders.
	A formal description follows.
	For files $x_j^n \in \cX^n$ for $j \in [N]$ and keys $k_{\b} \in \mathbb{F}^{nR_{\b}}$ for $\b \in \ak$ and private randomness $w_{\phi} \in \cW_{\phi}$ of encoder $\phi'$ (which is identical to the private randomness of encoder $\phi$), when $m = \phi(x_{[N]}^n, k_{\ak}, w_{\phi})$
	\begin{align}
		& \phi'\left(x_{[N]}^n, k_{\ak}, w_{\phi}\right) = (\theta, \hat{m}, \text{zeros}_{\theta}) \text{ where }
		\begin{cases}
		\theta = 0 \text{ and } \hat{m} \defeq m \text{ if } \psi_i(m, x^n_{\s_i}, k_{\ak_i}) = x_i^n, \forall i \in [N],\\
		\theta = 1 \text{ and } \hat{m} \defeq \alpha'(x_{[N]}^n) \text{ otherwise.}
		\end{cases}\label{eq_0repic_enc}
	\end{align}
	$\text{zeros}_{\theta}$ is the appropriate length of zero padding needed to keep the transmission at $R + \delta$.
	The encoder uses the $(n, \epsilon)$-$\repic$ scheme whenever it can be correctly decoded at all users, if not, it uses the zero-error index coding scheme whenever the $\repic$ scheme commits an error.
	The value of $\theta$ records the used scheme.
	That the transmission can be made $R + \delta$ (by zero padding) follows from the fact that $(\theta, \hat{m})$ is at most $1 + n(R + \frac{\delta}{3})$ long since rate of the index code being being at most $R + \delta/3$; recollect that $\frac{1}{n} \le \frac{\delta}{3}$.
	For $i \in [N]$,
	\begin{align}
		& \psi'_i : \left(\theta, \hat{m}, \text{zeros}_{\theta}, x_{\s_i}^n, k_{\ak_i}\right) =
		\begin{cases}
		\psi_i(\hat{m}, x_{\s_i}^n, k_{\ak_i}) \text{ if } \theta = 0,\\
		\beta'_i(\hat{m}, x_{\s_i}^n) \text{ otherwise}.
		\end{cases}\label{eq_0repic_dec}
	\end{align}
	We will establish that this is a $(n, \epsilon')$-zero-error $\repic$ scheme of rate $(R + \delta, R_{\b} : \b \in \ak)$.
	By construction, it has block-length $n$ and rate $(R + \delta, R_{\b} : \b \in \ak)$.
	It is easy to observe that it also guarantees zero-error decoding at all receivers.
	Next we show that its privacy parameter is $\epsilon'$, completing the proof of the lemma.
	
	For files $\xv_{[N]}$ and keys $\kv_{\ak}$, let $\Theta$ and $\hat{\mv}$, respectively, represent the random variables corresponding to values of $\theta$ and $\hat{m}$ at the output of the encoder $\phi'$ (See~\eqref{eq_0repic_enc}).
	The privacy condition at user $i$ is,
	\begin{align}
	I \left( \hat{\mv},  \Theta ;\xv_{[N] \setminus \a_i} \middle | \kv_{\ak_i}, \xv_{\s_i} \right) & = I \left( \Theta ;\xv_{[N] \setminus \a_i} \middle | \kv_{\ak_i}, \xv_{\s_i} \right) + I \left( \hat{\mv} ;\xv_{[N] \setminus \a_i} \middle | \kv_{\ak_i}, \xv_{\s_i}, \Theta \right) \nonumber\\
	& \le H(\Theta) +  I \left( \hat{\mv} ;\xv_{[N] \setminus \a_i} \middle | \kv_{\ak_i}, \xv_{\s_i}, \Theta \right) \label{lem1_priv_eq}.
	\end{align}
	Second term in the RHS can be bounded as follows.
	\begin{align}
	I \left( \hat{\mv} ;\xv_{[N] \setminus \a_i} \middle | \kv_{\ak_i}, \xv_{\s_i}, \Theta \right) & = \pr{\Theta = 0} I \left( \hat{\mv} ;\xv_{[N] \setminus \a_i} \middle | \kv_{\ak_i}, \xv_{\s_i}, \Theta = 0 \right)  \nonumber\\
	& \qquad + \pr{\Theta = 1} I \left( \hat{\mv} ;\xv_{[N] \setminus \a_i} \middle | \kv_{\ak_i}, \xv_{\s_i}, \Theta = 1 \right) \nonumber\\
	&\le \pr{\Theta = 0} I \left( \hat{\mv} ;\xv_{[N] \setminus \a_i} \middle | \kv_{\ak_i}, \xv_{\s_i}, \Theta = 0 \right) + \epsilon H\left(\xv_{[N]}\right) \label{lem1_priv_eq_t1}.
	\end{align}
	By the definition of $\Theta$ in~\eqref{eq_0repic_enc}, $\pr{\Theta = 1} \leq \epsilon$ since the decoding error of $(n, \epsilon)$-$\repic$ scheme is at most $\epsilon$.
	Hence, the inequality~\eqref{lem1_priv_eq_t1} follows from $I\left(\hat{\mv} ;\xv_{[N] \setminus \a_i} \middle | \kv_{\ak_i}, \xv_{\s_i}, \Theta = 1 \right) \leq H\left(\xv_{[N]}\right)$.
	When $\Theta = 0$, $\hat{\mv} = \phi\left(\xv_{[N]}, \kv_{\ak}\right)$. Then,
	\begin{align}
	\pr{\Theta = 0} \cdot I\left( \hat{\mv} ;\xv_{[N] \setminus \a_i} \middle | \kv_{\ak_i}, \xv_{\s_i}, \Theta = 0 \right) & =  \pr{\Theta = 0} \cdot I \left(\phi\left(\xv_{[N]}, \kv_{\ak}\right) ;\xv_{[N] \setminus \a_i} \middle | \kv_{\ak_i}, \xv_{\s_i}, \Theta = 0 \right)  \nonumber\\
	& \le I\left( \phi\left(\xv_{[N]}, \kv_{\ak}\right) ;\xv_{[N] \setminus \a_i} \middle | \kv_{\ak_i}, \xv_{\s_i}, \Theta \right) \nonumber\\
	& \le I\left( \phi\left(\xv_{[N]}, \kv_{\ak}\right), \Theta ;\xv_{[N] \setminus \a_i} \middle | \kv_{\ak_i}, \xv_{\s_i} \right) \nonumber\\
	& \le I\left( \phi\left(\xv_{[N]}, \kv_{\ak}\right) ;\xv_{[N] \setminus \a_i} \middle | \kv_{\ak_i}, \xv_{\s_i} \right) + H(\Theta)  \nonumber \\
	& \leq n \epsilon + H(\Theta). \label{lem1_priv_eq_t2}
	\end{align}
	The last inequality follows from the scheme with encoder $\phi$ and decoders $\{\psi_i\}_{i \in [N]}$ being a $(n, \epsilon)$-$\repic$ scheme.
	Since $H\left(\xv_{[N]}\right) = nN$ in \eqref{lem1_priv_eq_t1} and $H(\Theta) \leq 1$ in \eqref{lem1_priv_eq}~and~\eqref{lem1_priv_eq_t2}, by substituting \eqref{lem1_priv_eq_t2} in \eqref{lem1_priv_eq_t1}, and \eqref{lem1_priv_eq_t1} in \eqref{lem1_priv_eq}, and using $n > \frac{1}{\epsilon}$, we get
	\begin{align}
	I \left(\hat{\mv}, \Theta ;\xv_{[N] \setminus \a_i} \middle | \kv_{\ak_i}, \xv_{\s_i} \right) \le n(1 + N)\epsilon + 2 \le n \epsilon'. \label{Eq_priv_uppr}
	\end{align}
\hfill{\rule{2.1mm}{2.1mm}}

\subsection{Proof of Claim~\ref{clm:in_lem_det}}
		The encoder $\phi$ can be thought of as a channel from $\left( \xv_{[N]}, \kv_{\ak} \right)$ to $\mv \defeq \phi\left(\xv_{[N]}, \kv_{\ak}, W_{\phi}\right)$, with the private randomness of the encoder $W_{\phi}$ being the randomness of the channel.
		Between a sender and receiver, this channel can be simulated (asymptotically) using a channel simulation scheme that uses only common randomness between the sender and the users (and no private randomness at the sender).
		Below we state a result from channel simulation, modified for our purposes.
		\begin{theorem}~\cite[Theorem 10]{cuff}\label{thm_cuff}
			Consider random variables $(U, V)$ jointly distributed over $\cU \times \cV$.
			For any $\epsilon, \delta > 0$, there is a large enough $m$, $\beta \le I(U ; V) + \delta$, common randomness $W$ distributed uniformly over a finite set $\cW$, and deterministic maps $\mathsf{Enc}$ and $\mathsf{Dec}$ such that, when the operator $\lVert . \rVert$ denotes total variation distance,
			\begin{align*}
			\mathsf{Enc} : \cU^m \times \cW \rightarrow 2^{m\beta}, \;
			\mathsf{Dec} : 2^{m\beta} \times \cW \rightarrow \cV^m, \text{ and }
			\left \lVert  \left( U^m, V^m\right) 
			- \left(U^m, \mathsf{Dec}\left(Enc(U^m, W\right), W \right) \right \rVert \le \epsilon.
			\end{align*}
		\end{theorem}
	
		In the sequel, to avoid confusion, for each $\b \in \ak$, the random variable for the keys (of rate $nR_{\b}$) used in the $(n, \epsilon)$-zero-error-$\repic$ scheme is denoted by $\bl{\k}{n}_{\b}$ and the transmitted message is denoted by $\bl{\m}{n}$.
		Furthermore, $m$ i.i.d. copies of $\xv_i : i \in [N]$ is denoted by $\iid{\x}{n, m}$, and $m$ i.i.d. copies of $\bl{\k}{n}_{\b} : \b \in \ak$ and $\bl{\m}{n}$ are denoted by $\iid{\k}{n, m}$ and $\iid{\m}{n, m}$, respectively.
		For $1 \le j \le m$, the $j^{th}$, $n$-sized block in $\iid{\x}{n, m}_i, \iid{\k}{n, m}_{\b}$ and $\iid{\m}{n, m}$ are denoted by $\iid{\x}{n, (j)}_i, \iid{\k}{n, (j)}_{\b}$ and $\iid{\m}{n, (j)}$, respectively.

		Substituting $\left(\xv_{[N]}, \bl{\k}{n}_{\ak}\right)$ and $\bl{\m}{n}$ for $U$ and $V$, respectively, in above theorem we can conclude that there exists a large enough $m$, common randomness $W$, and functions $\mathsf{Enc}$ and $\mathsf{Dec}$ such that the image of $\mathsf{Enc}$ is $\mathbb{F}^{m\beta}$, where
	\begin{align}
	& m \beta \le m \left( I\left(\iid{\x}{n}_{[N]}, \iid{\k}{n}_{\ak}; \iid{\m}{n}\right) + n \cdot \min(\epsilon, \delta)\right) \le mn (R + \min(\epsilon, \delta)), \text{and}\label{eq_rate}\\
	& \bigg \lVert \left(\iid{\x}{n, m}_{[N]}, \iid{\k}{n, m}_{\ak}, \iid{\m}{n, m}\right) - \left(\iid{\x}{n, m}_{[N]}, \iid{\k}{n, m}_{\ak}, \mathsf{Dec}\left(J, W\right) \right) \bigg \rVert \le \epsilon, \text{ where } J \defeq \mathsf{Enc}\left(\left(\iid{\x}{n, m}_{[N]}, \iid{\k}{n, m}_{\ak}\right), W\right). \label{eq_det_norm}
	\end{align}
	Since $\mathsf{Dec}\left(J, W\right)$ is statistically close to $\iid{\m}{n, m}$, we will denote it by $\widehat{\iid{\m}{n, m}}$.

	Below, we define a $mn$ block-length scheme with encoder $\phi'$ and decoders $\{\psi_i\}_{i \in [N]}$ that use common randomness $W$ and \emph{no private randomness at encoder or decoders}.
	\begin{align*}
		& \phi'\left(\iid{\x}{n, m}_{[N]}, \iid{\k}{n, m}_{\ak}, W \right) = \mathsf{Enc} \left(\iid{\x}{n, m}_{[N]}, \iid{\k}{n, m}_{\ak}, W \right) = J,\\
		& \psi'_i\left(J, \iid{\x}{n, m}_{\s_i}, \iid{\k}{n, m}_{\ak_i}, W \right) = \left(\widehat{\iid{\x}{n, (1)}_i}, \ldots, \widehat{\iid{\x}{n, (m)}_i}\right), \text{where } \widehat{\iid{\x}{n, (j)}_i} \defeq \psi_i\left(\widehat{\iid{\m}{n, (j)}}, \iid{\x}{n, (j)}_{\s_i}, \iid{\k}{n, (j)}_{\ak_i}\right).
	\end{align*}
	This is a $mn$ block-length scheme and form~\eqref{eq_rate} it follows that its rate is $(R + \delta, R_{\b} : \b \in \ak)$.
	If $\epsilon < \delta$, the transmission rate can be made $R + \delta$ by zero-padding.
	Next, we establish that the decoding error of this scheme is at most $\epsilon$.
	The decoding error of is given by the expression,
	\begin{align*}
	1 - \pr{\widehat{\iid{\x}{n, (j)}_i} = \iid{\x}{n, (j)}_i, \forall i \in [N], j \in [m]}. 
	\end{align*}
	Inequality~\eqref{eq_det_norm} implies that
	\begin{align*}
		\bigg \lVert \left(\iid{\m}{n, (j)}, \iid{\x}{n, (j)}_{\s_i}, \iid{\k}{n, (j)}_{\ak_i}\right)_{j \in [m], i \in [N]} - \left(\widehat{\iid{\m}{n, (j)}}, \iid{\x}{n, (j)}_{\s_i}, \iid{\k}{n, (j)}_{\ak_i}\right)_{i \in [m], j \in [N]} \bigg \rVert \le \epsilon.
	\end{align*}
	Hence, by the definition of statistical distance,
	\begin{multline*}
	\pr{\psi'_i \left( \widehat{\iid{\m}{n, (j)}}, \iid{\x}{n, (j)}_{\s_i}, \iid{\k}{n, (j)}_{\ak_i} \right) = \iid{\x}{n, (j)}_i, i \in [N], j \in [m]}\\
	\ge \pr{\psi_i \left( \iid{\m}{n, (j)}, \iid{\x}{n, (j)}_{\s_i}, \iid{\k}{n, (j)}_{\ak_i} \right) = \iid{\x}{n, (j)}_i, i \in [N], j \in [m]} - \epsilon.
	\end{multline*}
	Since the scheme with encoder $\phi$ and decoders $\psi_i$ for user $i \in [N]$ is a zero-error $\repic$ scheme, 
	\begin{align*}
	\pr{\psi_i \left( \iid{\m}{n, (j)}, \iid{\x}{n, (j)}_{\s_i}, \iid{\k}{n, (j)}_{\ak_i} \right) = \iid{\x}{n, (j)}_i, i \in [N], j \in [m]} = 1.
	\end{align*}
	Hence the decoding error the new scheme with encoder $\phi'$ and decoders $\psi'_i$ for user $i \in [N]$ is at most $\epsilon$.
	
	To analyze the privacy parameter, we use the following lemma that claims that two distributions that are statistically close are also close in their entropies.
	\begin{lemma}~\cite[Lemma 2.7]{csiszar}\label{lem:csiszar}
		If $P, Q$ are two distributions on a finite set $\cX$ such that $\left \lVert P - Q \right \rVert \le \epsilon$, then
		\begin{align*}
		\left | H\left(P\right) - H\left(Q\right) \right | \le \epsilon(\log{\left|\cX\right|} - \log{\epsilon}).
		\end{align*}
	\end{lemma}
	The privacy condition at user $i \in [N]$ is given by the mutual information,
	\begin{align}
	I \left( J ; \iid{\x}{n, m}_{[N] \setminus \a_i} \middle | \iid{\k}{n, m}_{\ak_i}, \iid{\x}{n, m}_{\s_i}, W \right) \le I \left( J , W; \iid{\x}{n, m}_{[N] \setminus \a_i} \middle | \iid{\k}{n, m}_{\ak_i}, \iid{\x}{n, m}_{\s_i} \right). \label{Eq_priv_mutl_info}
	\end{align}
	Since $\widehat{\iid{\m}{n, m}}$ is a function of $(J, W)$, the above inequality implies that,
	\begin{multline}
	I \left( J ; \iid{\x}{n, m}_{[N] \setminus \a_i} \middle | \iid{\k}{n, m}_{\ak_i}, \iid{\x}{n, m}_{\s_i}, W \right) \le I \left( J, W, \widehat{\iid{\m}{n, m}} ; \iid{\x}{n, m}_{[N] \setminus \a_i} \middle | \iid{\k}{n, m}_{\ak_i}, \iid{\x}{n, m}_{\s_i} \right)\\ = I \left( \widehat{\iid{\m}{n, m}} ; \iid{\x}{n, m}_{[N] \setminus \a_i} \middle | \iid{\k}{n, m}_{\ak_i}, \iid{\x}{n, m}_{\s_i} \right) + I \left( J, W ; \iid{\x}{n, m}_{[N] \setminus \a_i} \middle | \iid{\k}{n, m}_{\ak_i}, \iid{\x}{n, m}_{\s_i}, \widehat{\iid{\m}{n, m}} \right).
	\label{Eq_priv_mutl_info2}
	\end{multline}
	The first term in the RHS of~\eqref{Eq_priv_mutl_info2} is upper bounded using the following observation.
	Using Lemma~\ref{lem:csiszar},
	\begin{multline}\label{Eq_priv_mutl_info3}
	I \left( \widehat{\iid{\m}{n, m}} ; \iid{\x}{n, m}_{[N] \setminus \a_i} \middle | \iid{\k}{n, m}_{\ak_i}, \iid{\x}{n, m}_{\s_i} \right) - I \left( \iid{\m}{n, m} ; \iid{\x}{n, m}_{[N] \setminus \a_i} \middle | \iid{\k}{n, m}_{\ak_i}, \iid{\x}{n, m}_{\s_i} \right)\\
	= I \left( \widehat{\iid{\m}{n, m}}, \iid{\k}{n, m}_{\ak_i}, \iid{\x}{n, m}_{\s_i}  ; \iid{\x}{n, m}_{[N] \setminus \a_i} \right) - I \left( \iid{\m}{n, m}, \iid{\k}{n, m}_{\ak_i}, \iid{\x}{n, m}_{\s_i} ; \iid{\x}{n, m}_{[N] \setminus \a_i} \right).
	\end{multline}
	This follows from $I \left( \iid{\k}{n, m}_{\ak_i}, \iid{\x}{n, m}_{\s_i}  ; \iid{\x}{n, m}_{[N] \setminus \a_i} \right) = 0$ as the files and keys are independently distributed.
	The RHS can be expanded as,
	\begin{multline*}
	H \left( \widehat{\iid{\m}{n, m}}, \iid{\k}{n, m}_{\ak_i}, \iid{\x}{n, m}_{\s_i} \right) - \left(H \left( \widehat{\iid{\m}{n, m}}, \iid{\k}{n, m}_{\ak_i}, \iid{\x}{n, m}_{\s_i}, \iid{\x}{n, m}_{[N] \setminus \a_i} \right) - H \left(\iid{\x}{n, m}_{[N] \setminus \a_i} \right)\right) \\
	- H \left( \iid{\m}{n, m}, \iid{\k}{n, m}_{\ak_i}, \iid{\x}{n, m}_{\s_i} \right) + \left(H\left( \iid{\m}{n, m}, \iid{\k}{n, m}_{\ak_i}, \iid{\x}{n, m}_{\s_i}, \iid{\x}{n, m}_{[N] \setminus \a_i} \right) - H\left(\iid{\x}{n, m}_{[N] \setminus \a_i} \right)\right).
	\end{multline*}
	Since,
	\begin{align*}
		\bigg \lVert \left(\iid{\m}{n, m}, \iid{\x}{n, m}_{[N]}, \iid{\k}{n, m}_{\ak}\right) - \left(\widehat{\iid{\m}{n, m}}, \iid{\x}{n, m}_{[N]}, \iid{\k}{n, m}_{\ak}\right) \bigg \rVert \le \epsilon,
	\end{align*}
	by Lemma~\ref{lem:csiszar}, the above expression can be bounded as,
	\begin{multline}\label{Eq_priv_mutl_info4}
	H \left( \widehat{\iid{\m}{n, m}}, \iid{\k}{n, m}_{\ak_i}, \iid{\x}{n, m}_{\s_i} \right) - H \left( \widehat{\iid{\m}{n, m}}, \iid{\k}{n, m}_{\ak_i}, \iid{\x}{n, m}_{\s_i}, \iid{\x}{n, m}_{[N] \setminus \a_i} \right) \\
	- H \left( \iid{\m}{n, m}, \iid{\k}{n, m}_{\ak_i}, \iid{\x}{n, m}_{\s_i} \right) + H\left( \iid{\m}{n, m}, \iid{\k}{n, m}_{\ak_i}, \iid{\x}{n, m}_{\s_i}, \iid{\x}{n, m}_{[N] \setminus \a_i} \right)\\
	\le 2\epsilon\left(\log{\mathsf{support}\left( \iid{\m}{n, m}, \iid{\x}{n, m}_{[N]}, \iid{\k}{n, m}_{\ak}\right)} - \log{\epsilon}\right)
	\le 2\epsilon\left( mn \left( R + N + \sum_{\b \in \ak} R_b \right) + \log{\frac{1}{\epsilon}}\right).
	\end{multline}
	From~\eqref{Eq_priv_mutl_info3}~and~\eqref{Eq_priv_mutl_info3}, we can conclude that,
	\begin{multline}
	I \left( \widehat{\iid{\m}{n, m}} ; \iid{\x}{n, m}_{[N] \setminus \a_i} \middle | \iid{\k}{n, m}_{\ak_i}, \iid{\x}{n, m}_{\s_i} \right)\\
	\le I \left( \iid{\m}{n, m} ; \iid{\x}{n, m}_{[N] \setminus \a_i} \middle | \iid{\k}{n, m}_{\ak_i}, \iid{\x}{n, m}_{\s_i} \right) + 2\epsilon\left( mn \left( R + N + \sum_{\b \in \ak} R_b \right) + \log{\frac{1}{\epsilon}}\right) \\
	= m \cdot I \left( \bl{\m}{n} ; \iid{\x}{n}_{[N] \setminus \a_i} \middle | \bl{\k}{n}_{\ak_i}, \iid{\x}{n}_{\s_i} \right) + 2\epsilon\left( mn \left( R + N + \sum_{\b \in \ak} R_b \right) + \log{\frac{1}{\epsilon}}\right) \\
	\le mn\epsilon + 2\epsilon\left( mn \left( R + N + \sum_{\b \in \ak} R_b \right) + \log{\frac{1}{\epsilon}}\right).\label{Eq_uppr_bnd_first}
	\end{multline}
	Here, the last inequality follows from the privacy condition at user $i$ for the given block-length $n$ scheme.
	The second term in the RHS of~\eqref{Eq_priv_mutl_info2} can be bounded as 
	\begin{multline}
	I \left( J, W ; \iid{\x}{n, m}_{[N] \setminus \a_i} \middle | \iid{\k}{n, m}_{\ak_i}, \iid{\x}{n, m}_{\s_i}, \widehat{\iid{\m}{n, m}} \right) \\ 
	\le I \left( J, W ; \iid{\x}{n, m}_{[N]}, \iid{\k}{n, m}_{\ak} \middle | \widehat{\iid{\m}{n, m}} \right)
	= I \left( J, W, \widehat{\iid{\m}{n, m}} ; \iid{\x}{n, m}_{[N]}, \iid{\k}{n, m}_{\ak} \right) - I \left( \widehat{\iid{\m}{n, m}} ; \iid{\x}{n, m}_{[N]}, \iid{\k}{n, m}_{\ak} \right).\label{eq_priv_bound}
	\end{multline}
	We upper bound the RHS of~\eqref{eq_priv_bound} by upper bounding  the first term and lower bounding  the second term.
	The first term can be upper bounded as follows.
	\begin{multline}
	I \left( J, W, \widehat{\iid{\m}{n, m}} ; \iid{\x}{n, m}_{[N]}, \iid{\k}{n, m}_{\ak} \right) \stackrel{(a)}{=} I \left( J, W ; \iid{\x}{n, m}_{[N]}, \iid{\k}{n, m}_{\ak} \right) \stackrel{(b)}{=} I \left( J ; \iid{\x}{n, m}_{[N]}, \iid{\k}{n, m}_{\ak} \middle | W \right) \\
	\le H(J) \stackrel{(c)}{\le} m\beta \stackrel{(d)}{\le} m\left(I\left(\iid{\m}{n} ; \iid{\x}{n}_{[N]}, \iid{\k}{n}_{\ak}\right) + n\epsilon\right) = I\left(\iid{\m}{n, m} ; \iid{\x}{n, m}_{[N]}, \iid{\k}{n, m}_{\ak}\right) + mn\epsilon. \label{Eq_mutl_uppr_bnd}
	\end{multline}
	Here, (a) follows from $\widehat{\iid{\m}{n, m}}$ being a deterministic function of $J, W$, (b) follows from the independence between $W$ and $\iid{\x}{n, m}, \iid{\k}{n, m}_{\ak}$, (c) follows from the range of $\mathsf{Enc}$ being $\mathbb{F}^{m\beta}$, and finally, (d) follows from~\eqref{eq_rate}.
	Similar to the bound~\eqref{Eq_uppr_bnd_first} obtained on the first term in the RHS of~\eqref{Eq_priv_mutl_info2}, using Lemma~\ref{lem:csiszar}, we bound the second term in the RHS of~\eqref{eq_priv_bound} as
	\begin{align}
	I\left(\widehat{\iid{\m}{n, m}} ; \iid{\x}{n, m}_{[N]}, \iid{\k}{n, m}_{\ak}\right) \ge I\left(\iid{\m}{n, m} ; \iid{\x}{n, m}_{[N]}, \iid{\k}{n, m}_{\ak}\right) - 2\epsilon\left( mn \left( R + N + \sum_{\b \in \ak} R_b \right) + \log{\frac{1}{\epsilon}}\right). \label{Eq_mutl_low_bnd}
	\end{align}
	By substituting \eqref{Eq_mutl_low_bnd} and \eqref{Eq_mutl_uppr_bnd} in \eqref{eq_priv_bound}, we get 
	\begin{align}
	I \left( J, W ; \iid{\x}{n, m}_{[N] \setminus \a_i} \middle | \iid{\k}{n, m}_{\ak_i}, \iid{\x}{n, m}_{\s_i}, \widehat{\iid{\m}{n, m}} \right)	
	\leq mn\epsilon + 2\epsilon\left( mn \left( R + N + \sum_{\b \in \ak} R_b \right) + \log{\frac{1}{\epsilon}}\right).\label{Eq_uppr_mutl_eps}
	\end{align}
	Using the bounds~\eqref{Eq_uppr_mutl_eps} and~\eqref{Eq_uppr_bnd_first} in~\eqref{Eq_priv_mutl_info2}, for all $i \in [N]$,
	\begin{multline}
	I \left( J ; \iid{\x}{n, m}_{[N] \setminus \a_i} \middle | \iid{\k}{n, m}_{\ak_i}, \iid{\x}{n, m}_{\s_i}, W \right) \le 2mn\epsilon + 4\epsilon\left( mn \left( R + N + \sum_{\b \in \ak} R_b \right) + \log{\frac{1}{\epsilon}}\right)\le mn\epsilon',\\ 
	\text{where } \epsilon' = 4\epsilon\left(\frac{1}{2} + R + N + \sum_{\b \in \ak}R_{\b} + \frac{1}{mn} \cdot \log{\frac{1}{\epsilon}}\right). \label{eq_det_priv}
	\end{multline}
	Hence, the new scheme has privacy error $\epsilon'$.
	The common randomness $W$ used in the scheme can be fixed to an appropriate value (see Lemma~\ref{lem_no_cr}) to obtain a $(mn, 2\epsilon, 2N\epsilon')$-$\pic$ scheme.
	This proves the claim as $\epsilon' \rightarrow 0$ as $\epsilon \rightarrow 0$.
\subsection{Proof of Lemma~\ref{lem:priv_linear}}
Consider perfect private linear index coding scheme for $N$ users of block-length $n$ and rate $R$, that uses private randomness at the encoder.
Let $\bX_i$ denote the row-vector corresponding to $\xv$, $\bK_{\b}: \b \in \ak$, denote the key uniformly distributed in $\mathbb{F}^{nR_{\b}}$ and let $\bP$ denote the private randomness uniformly distributed in $\mathbb{F}^{p}$ for some number $p$.
Since $\ak$ is the same as $\{1, \ldots, 2^N - 2\}$, in the sequel, we will often represent $\b$ as a number between $1$ and $2^N - 2$.
A linear encoder of the scheme may be described as follows.
\begin{align*}
	\bM &= \sum_{i \in [N]} G_i \bX_i^T + \sum_{1 \le \b \le 2^N - 2} H_{\b} \bK_{\b}^T + H' \bP\\
	&= \begin{bmatrix}
		G_1 & \ldots & G_N & H_{1} & \ldots & H_{2^N - 2} & H'
	\end{bmatrix}
	\begin{bmatrix}
		\bX_1 & \ldots & \bX_N & \bK_{1} & \ldots & \bK_{2^N - 2} & \bP
	\end{bmatrix}^T\\
	&\defeq \Psi
	\begin{bmatrix}
		\bX_1 & \ldots & \bX_N & \bK_{1} & \ldots & \bK_{2^N - 2} & \bP
	\end{bmatrix}^T
\end{align*}
where $G_i \in \mathbb{F}^{nR \times n}$, $H_{\b} \in \mathbb{F}^{nR \times nR_{\b}}$ and $H' \in \mathbb{F}^{nR \times np}$ and $\Phi$ is the generator matrix of the encoder. 
For $i \in [N]$, by the zero-error decoding condition, there exists $L^i \in \mathbb{F}^{n \times nR}$, $S^i_j \in \mathbb{F}^{n \times n}$ where $j \in \s_i$ and $T^i_{\b} \in \mathbb{F}^{n \times nR_{b}}$ where $\b \in \ak_i$ such that
\begin{align}\label{eq:priv_linear_dec}
	\bX_i = L^i \bM + \sum_{j \in \s_i} S^i_j \bX_i^T + \sum_{\b \in \ak_i} T^i_{\b} \bK_{\b}^T.
\end{align}
Hence, $L^i \bM$ is determined by $\bX_i, i \in \a_i$ and $\bK_{\b}, \b \in \ak_i$ and it is independent of $\bP$.
Consider the encoding scheme
\begin{align*}
	\begin{bmatrix}
		\bM^1 & \ldots & \bM^N
	\end{bmatrix}^T =
	\begin{bmatrix}
		(L^1)^T & \ldots & (L^N)^T
	\end{bmatrix}^T \Phi.
	\begin{bmatrix}
		\bX_1 & \ldots & \bX_N & \bK_{1} & \ldots & \bK_{2^N - 2} & \bP
	\end{bmatrix}^T.
\end{align*}
It follows from the above observation that this encoder does not use private randomness.
The actual encoding will use a row-reduced form of the matrix $\left[(L^1)^T \; \ldots \; (L^N)^T\right]^T \Phi$; \emph{i.e.,} after removing the independent rows.
Clearly, this encoding scheme is private since $\Psi$ is private.
It follows from~\eqref{eq:priv_linear_dec} that 
\begin{align*}
	\bX_i = \bM^i + \sum_{j \in \s_i} S^i_j \bX_i^T + \sum_{\b \in \ak_i} T^i_{\b} \bK_{\b}^T.
\end{align*}
Hence, we have obtained a linear private index coding scheme that does not use private randomness.
Since the rank of $[(L^1)^T \; \ldots \; (L^N)^T]^T \Phi$ is at most the rank of $\Phi$, the rate of this scheme is at most the rate of the given scheme.
This proves the lemma.
\hfill{\rule{2.1mm}{2.1mm}}

\section{Details omitted in Section~\ref{sec_1b}}

\subsection{Proof of Lemma~\ref{Lem_Subset}}

\label{Sec_proof_subset}
Let $i,j \in [N], i \neq j$ be such that $i \notin \s_j$ and $\s_i \subseteq \a_j$. By Fano's inequality, the decodability criterion for user $i$ implies that for each $\epsilon > 0$, there exists a large enough block length $n$ such that
\begin{align*}
    H \left(\xv_i \mid \mv, \xv_{\s_i}\right) \leq n \epsilon .
\end{align*}
Since $\s_i \subseteq (j \cup \s_j)$, it follows that
$ H \left(\xv_i \mid \mv, \xv_j, \xv_{\s_j}\right)  \leq n \epsilon$.
Since $i\notin \s_j$, consider the weak privacy condition for message $X_i$ at user $j$ 
\begin{align*}
     I \left(\mv, \xv_j,\xv_{\s_j};\xv_i\right) & = H\left(\xv_i\right) - H\left(\xv_i \mid \mv, \xv_j,\xv_{\s_j}\right) \\
     & \geq  n - n \epsilon,
\end{align*}
where the last inequality follows since $H\left(\xv_i \mid \mv,\xv_j,\xv_{\s_j}\right) \leq n \epsilon $, and $ H\left(\xv_i\right)  = n $. Thus, in this case weak privacy cannot be achieved. This proves Lemma~\ref{Lem_Subset}.
\hfill{\rule{2.1mm}{2.1mm}}

\subsection{Proof of Claim~\ref{Clm_subset}}
\label{Sec_proof_example}
From the decodability condition of user 1,  for each $\epsilon > 0$ and large enough block length $n$,  we have
\begin{align}
H\left(\xv_1|\mv,\xv_3,\xv_4\right) & \leq n \epsilon  \notag
\end{align}
which implies that $I\left(\mv, \xv_3,\xv_4 ; \xv_1 \right) \geq n(1-\epsilon)$. Using the chain rule of mutual information, we obtain
\begin{align}
I\left(\mv, \xv_3,\xv_4 ; \xv_1 \right)  & = I\left(\mv,\xv_3;\xv_1\right) + I\left(\xv_4;\xv_1|\mv,\xv_3\right) \label{Eq_mutual_chain}\\
& \geq n(1-\epsilon). \notag
\end{align}
First, let us assume that  $I\left(\mv,\xv_3;\xv_1\right) \geq n(1-\epsilon')$ for some $0< \epsilon' <  \epsilon$. Then, it clearly violates the privacy of $\xv_1$ at user 6 since user 6 has both $\mv$ and $\xv_3$. This implies that $I\left(\mv,\xv_3;\xv_1\right)$ cannot be greater than $n\epsilon'$.
 By assuming $I\left(\mv,\xv_3;\xv_1\right) \leq n\epsilon'$, it follows from~\eqref{Eq_mutual_chain} that  $I\left(\xv_4;\xv_1|\mv,\xv_3\right) \geq n(1-\epsilon'')$, where $\epsilon'' = \epsilon - \epsilon'$. Thus, we have
\begin{align}
I\left(\xv_4;\xv_1, \mv,\xv_3\right) \geq n(1-\epsilon''). \label{Eq_priv_cond}
\end{align}
This implies that user 3 will learn about $\xv_4$ since it has $(\xv_1, \mv,\xv_3)$ which violates the privacy of $\x_4$ at user 3. This shows that if the decodability conditions are satisfied at all users, we cannot achieve weak privacy for this example.

\subsection{Proof of Theorem~\ref{prop_neces}}
\label{Sec_proof_prop}
To prove the theorem, we assume that there exists a user $i \in [N]$ for which the condition is satisfied, and we show that weak privacy is not feasible.  The decodability condition for user $i$ gives that  for each $\epsilon > 0$ and large enough block length $n$,  we have
\begin{align}
I\left(\xv_{\s_i}, \mv;\xv_i\right)  \geq n(1-\epsilon).
\end{align}
First, let us assume that $S =\a_i$. Then  there exists a user $j \in [N]$ and a $ k \in[N]$ such that $\a_i\setminus \{k\} \subseteq \a_j$ and $ k \notin \a_j$. If $k =i$, then it violates  the subset condition given in Lemma~\ref{Lem_Subset}, so weak privacy is not feasible. Hence, let us assume that $k \neq i$. Then for $S = S' \cup \{k\}$
\begin{align}
I\left(\mv, \xv_{S'}, \xv_{k}; \xv_i\right) & = I\left(\mv, \xv_{S'}; \xv_i\right) + I\left(\xv_k; \xv_i| \mv, \xv_{S'}\right) \label{Eq_chain_mutual} \\
& \geq n(1-\epsilon). \notag
\end{align}
If $I\left(\xv_k; \xv_i| \mv, \xv_{S'}\right) \geq n(1-\epsilon')$ for some $0 < \epsilon' < \epsilon$, then  
\begin{align*}
I\left(\mv, \xv_{S'}, \xv_i ; \xv_k\right) &  \geq n(1-\epsilon')
\end{align*}
 since $ I\left(\xv_k; \xv_i| \mv, \xv_{S}\right) \leq I\left(\mv, \xv_{S}, \xv_i ; \xv_k\right)$.
Since user $j$ has $\left(\mv, \xv_{S'}, \xv_i \right)$, it violates the privacy of $\x_k$ at user $j$. So, let us assume that 
$ I\left(\xv_k; \xv_i| \mv, \xv_{S}\right) \leq n\epsilon'$. Then it follows from~\eqref{Eq_chain_mutual} that $ I\left(\mv, \xv_{S'}; \xv_i\right) \geq n(1-\epsilon'')$, where $ \epsilon'' = \epsilon - \epsilon'$. Since the given condition in Theorem~\ref{prop_neces} satisfies for all subsets of $\a_i$, by following the above arguments, we can show that for some $k\in [N]$ and $S'' \subseteq \a_i$ such that $S'= S'' \cup \{k\}$, we have
\begin{align}
  I\left(\mv, \xv_{S''}; \xv_i\right) \geq n(1-\tilde{\epsilon})
\end{align}
for some $\tilde{\epsilon} > 0$.
If we continue this, we get that for some $l \in [N], l \neq i$,
\begin{align*}
I\left(\mv, \xv_l ; \xv_i\right) & \geq n(1-\delta) \text{ for some } \delta> 0.
\end{align*}
Now let us take $S = \{l, i \}$. In this case, if $k = i$, then it clearly violates the privacy of $\x_i$ at the user since the user has both $\mv$ and $\xv_l$. So, let $k = l$. Then,
\begin{align*}
I\left(\mv, \xv_l ; \xv_i\right) & = I \left(\mv ; \xv_i\right) + I\left(\xv_l;\xv_i|\mv\right).
\end{align*}
As we already noted that if $I\left(\mv;\xv_i\right)$ is large, then it violates the privacy at some user. This implies that $I\left(\xv_l;\xv_i|\mv\right) \geq n(1-\delta')$ for some $\delta'>0$. Since $I\left(\xv_l;\xv_i|\mv\right) \leq I\left(\mv, \xv_i; \xv_l\right)$, it violates the privacy of $\x_l$ at the user who does not have $\x_l$. This shows that we cannot achieve weak privacy in this case. This completes the proof of the theorem.

\subsection{ Proof of Theorem~\ref{Thm_Sec_Cliq}}
\label{Proof_prop_cliq_covr}
Let $\cC_{G}$ denote a secure clique cover for the index coding problem represented by graph $G$ and let $C_1,C_2,\ldots, C_{|\cC_G|}$ be the cliques in $\cC_G$. 
Further, let
\begin{align}
M_i = \sum_{j \in C_i} X_j, \; i=1,\ldots, |\cC_G|.  \label{Eq_messg_cliq_def}
\end{align}
The server transmits  message $M=\left\{M_1,\ldots, M_{|\cC_G|}\right\}$. Let $k_i$ denote the index of the clique that $i$ belongs to, i.e., $i \in C_{k_i}$. Since the cliques are disjoint, this index is unique. User $i$ decodes $X_i$ from $M_{k_i}$ since it has $X_j,j\neq i$ for all $j \in C_{k_i}$. Next we prove that this scheme also satisfies the weak privacy. To this end, we show that if $j \notin \s_i$, then
\begin{align}
I(M; X_j|X_{\a_i}) =0. \label{Eq_sec_weak_priv}
\end{align}
We show~\eqref{Eq_sec_weak_priv} by  using the facts that the cliques are disjoint and if user $i$ does not have $X_j$ as side information, then it does not have access to at least one more message $X_l, l \in C_{k_j}$. 
We first define
\begin{align}
\tilde{M}(k_j) =  M \setminus\{M_{k_j}\}. \notag
\end{align}
Then, we have
\begin{align}
I(M; X_j|X_{\a_i}) &=  I\left(\tilde{M}(k_j) ; X_j|X_{\a_i} \right) +  I\left(M_{k_j}; X_j|X_{\a_i}, \tilde{M}(k_j)  \right) \notag \\
& =   I\left(M_{k_j}; X_j|X_{\a_i}, \tilde{M}(k_j)  \right) \label{Eq_sec_cliq1} \\
& = I\left(\sum_{l\in C_{k_j}} X_l; X_j| X_{\a_i},\tilde{M}(k_j)   \right) \label{Eq_sec_cliq2} \\
& = I\left(\sum_{l\in C_{k_j} \setminus {\a_i} } X_l; X_j| X_{\a_i}, \tilde{M}(k_j)   \right)\label{Eq_sec_cliq3} \\
& =  I\left(\sum_{l\in C_{k_j} \setminus {\a_i} } X_l; X_j  \right)\\
& = 0. \label{Eq_sec_cliq4} 
\end{align}
Here, \eqref{Eq_sec_cliq1} follows since $X_j$ is not part of $X_{\a_i}$ or $\tilde{M}(k_j)$, and in~\eqref{Eq_sec_cliq2}, we used the definition of $M_i$ given in~\eqref{Eq_messg_cliq_def}. Further,  \eqref{Eq_sec_cliq3} follows since any $X_l, l\in C_{k_j}\setminus \a_i$ is not part of  $X_{\a_i}$ or $\tilde{M}(k_j)$ because of the definition of secure clique cover, and~\eqref{Eq_sec_cliq4} follows since
there exists an $l \neq j$ such that $l \in C_{k_j} \setminus {\a_i}$ which implies that 
$\sum_{l\in C_{k_j} \setminus {\a_i} } X_l$ is independent of $X_j$.

\subsection{Details of Fig.~\ref{Fig_linear}}
\label{Sec_linear_schm} 
In this subsection, we argue that the example in Fig.~\ref{Fig_linear} does not
have a secure clique cover, but it is still feasible under weak privacy.  For
the given index coding problem, if $i \in \s_j$, then $j \notin \s_i $. Then,
the only clique cover of the side information graph is of all singleton sets
which is not secure. Thus, it has no secure clique cover.  We give a linear
encoding matrix $M$ that satisfies the decodability and privacy condition for
weak privacy.  The matrix $M$ is as given below \[ M= \begin{bmatrix} 1 & 1 & 1
& 0 & 0 & 0 & 1\\ 0 & 1 & 1 & 1 & 0 & 1&0 \\ 0 & 0 & 1 & 1 & 1 & 0&1 \\ 1 & 0 &
0 & 1 & 0 & 1&1 \\ 1 & 1 & 0 & 1 & 1 & 0&0 \\ 1 & 0 & 1 & 0 & 1 & 1&0 \\ 0 & 1
& 0 & 0 & 1 & 1&1 \\ \end{bmatrix} \]

Let $M_i$ denote the $i^{\text{th}}$ column of $M$, and for a subset $S$ of the columns let $ \langle M_{S} \rangle $
denote the span of all the columns in $S$.
 The decodability conditions are satisfied for $M$ since
$   M_1 \notin \langle M_{\{4,5,6\}}\rangle, 
    M_2 \notin \langle M_{\{1,5,7\}}\rangle, 
    M_3 \notin \langle M_{\{1,2,6\}}\rangle, 
    M_4 \notin \langle M_{\{2,3,5\}}\rangle,
    M_5 \notin \langle M_{\{3,6,7\}}\rangle,  
    M_6 \notin \langle M_{\{2,4,7\}}\rangle, \text{ and }
    M_7 \notin \langle M_{\{1,3,4\}}\rangle. 
$
The privacy constraints are satisfied since
\begin{align*}
    M_1 &\in \langle M_{\{5,7\}}\rangle  \cap \langle M_{\{2,6\}} \rangle  \cap \langle M_{\{3,4\}}\rangle   \\
    M_2 &\in \langle M_{\{1,6\}}\rangle \cap \langle M_{\{3,5\}} \rangle  \cap \langle M_{\{4,7\}}\rangle    \\
    M_3 &\in \langle M_{\{2,5\}}\rangle  \cap \langle M_{\{6,7\}} \rangle  \cap \langle M_{\{1,4\}}\rangle    \\
    M_4 &\in \langle M_{\{5,6\}}\rangle  \cap \langle M_{\{2,7\}} \rangle  \cap \langle M_{\{1,3\}} \rangle    \\
    M_5 &\in \langle M_{\{4,6\}}\rangle  \cap \langle M_{\{1,7\}} \rangle  \cap \langle M_{\{2,3\}}\rangle    \\
    M_6 &\in \langle M_{\{4,5\}}\rangle  \cap \langle M_{\{1,2\}} \rangle  \cap \langle M_{\{3,7\}}\rangle    \\
    M_7 &\in \langle M_{\{1,5\}}\rangle  \cap \langle M_{\{3,6\}} \rangle \cap \langle M_{\{2,4\}}\rangle.
\end{align*}

\section{Proof of Theorem~\ref{Thm_multicast}}

\label{Proof_multicast}
 In Claim~\ref{Clm_Multicst} below, we  argue that without loss of generality, we can restrict ourselves to multicast schemes in which 
the following holds for each  $\cS_k$:
 \begin{align}
 \text{if } i,j \in \cS_k, \text{ then } i \in \s_j \text { and } j \in \s_i, \text { for all } i,j \in [N], i\neq j. \label{Eq_restr_multi}
 \end{align}

\begin{claim}
\label{Clm_Multicst}
From any multicast scheme $\cM_K$ with $K$ sessions, we can obtain a scheme $\cM'_K$ with $K$ sessions such that 
for all $\cS'_k$ in $\cM'_K$ it holds that if $i,j \in \cS'_k$, then $i \in \s_j$ and $j \in \s_i$.
\end{claim}

\begin{proof}
For a given $k=1,\ldots, K$, let $ i \in \cS_k$. Then, the perfect privacy criterion for user $i$ implies that
\begin{align}
   I\left(M_k,\xv_i,\xv_{\s_i} ; \xv_{[N]\setminus\a_i}\right) = 0. \label{Eq_perf_priv}
\end{align}
Then, we have
\begin{align}
   I\left(M_k;   \xv_{[N]\setminus\a_i}|\xv_{\a_i} \right) = 0. \label{Eq_perf_priv1}
\end{align} 
For any $j \in [N]\setminus \a_i $, from \eqref{Eq_perf_priv1} we get 
\begin{align}
   I\left(M_k ; \xv_j |\xv_{[N]\setminus j} \right) = 0. \label{Eq_perf_priv2}
\end{align}
Since $\s_j \subseteq [N]\setminus  \{j\} $ and by using the fact that files are independent, it follows from \eqref{Eq_perf_priv2} that
$I\left(M_k ; \xv_j |\xv_{\s_j}\right) = 0$.
This implies that user $j$ learns nothing about $\xv_j$ from $M_k$. Thus,
$M_k$ is not useful to user $j$ in decoding $\xv_j$.
Then it follows that if $j \in \cS_k$, then we can modify $\cS_k$ to $\cS'_k$ by excluding $\{j\}$ in $\cS_k$, i.e., $\cS'_k = \cS_k\setminus \{j\}$. This proves the claim.
\end{proof}

Next we show that at block length $n$, the minimum number of multicasts required is the $n$-fold chromatic number of $G^c$.
\begin{claim}
\label{clm_multi_sessn}
 $\kappa_n(G)= \chi_{n}(G^{c})$. 
\end{claim}
\begin{proof}
	To prove this claim, we restrict ourselves to the multicast schemes of the form  \eqref{Eq_restr_multi}.  We prove the claim by showing that at block length $n$, from any multicast scheme of the form \eqref{Eq_restr_multi} we can obtain an $n$-fold coloring, and
	we can also  obtain a multicast scheme from any $n$-fold coloring.
	
 Recall that an $n$-fold coloring is an assignment of sets of size $n$ to each vertex such that adjacent vertices get  disjoint sets.
For an $n$-fold coloring, let $\cC_i$ denote the set of size $n$ assigned to vertex $i$, $i \in [N]$.	
From a multicast scheme with $K$ sessions, we obtain an $n$-fold coloring $\cC_i, i \in [N]$   as follows:  Since each multicast session transmits one symbol from $\mathbb{F}$, to recover $\xv_i$ at user $i$, the cardinality of the set $\{k: i \in \cS_k\}$ should be at least $n$. We obtain the set $\cC_i$ by including the first $n$ indices from the set $\{k: i \in \cS_k\}$ to $\cC_i$. Suppose vertices $i$ and $j$ are adjacent in $G^c$.
Then,  the sets $\cC_i$ and $\cC_j$ obtained by this method are disjoint since for any $k=1,\ldots,K$, both $i$ and $j$ do not belong to $\cS_k$. This is because $i,j \in \cS_k$ for some $k$ implies that $i$ and $j$ are not adjacent in $G^c$ due to~\eqref{Eq_restr_multi}. Thus, we have an $n$-fold coloring.
 
 Now we show that we can obtain a multicast scheme with $K$ sessions from an $n$-fold coloring with $K$ colors.
 Let us assume that the $K$ colors of the $n$-fold coloring are numbered from 1 to $K$.
 Without loss of generality, we assume that the set $\cC_i, i\in[N]$ is an ordered set.
 For $k\in \cC_i , k=1,\ldots, K$, let $\cC_i(k)$ denote its position in the set $\cC_i$. By denoting the $j^{\text{th}}$ instance of $i^{\text{th}}$ message by $X_i^{(j)}$, the message $M_k$ transmitted in the $k^{\text{th}}$ multicast session is given by
 \begin{align*}
  M_k = \sum_{i: k \in \cC_i} X_i^{\left(\cC_i(k)\right)}, \; k=1,\ldots, K.
 \end{align*}
 If $j \in \cS_k$, then user $j$ can decode $X_j^{\left(\cC_j(k)\right)}$ since it has access to all the other random variables $\left\{ X_i^{\left(\cC_i(k)\right)}, i \neq j \right\}$ in $M_k$. It is easy to observe that this scheme achieves perfect privacy. This shows that $\kappa_n(G)$ is given by the $n$-fold chromatic number $\chi_{n}(G^c)$.
\end{proof}
From Claim~\ref{clm_multi_sessn}, it follows that
\begin{align*}
 \kappa(G) = \inf_{n} \frac{\kappa_n(G)}{n} = \inf_{n} \frac{\chi_{n}(G^{c})}{n} = \chi_f(G^c).
\end{align*}
This proves Theorem~\ref{Thm_multicast}.

	\end{appendices}
	\fi
	
	\section*{Acknowledgments}
	V. Prabhakaran and N. Karamchandani acknowledge initial discussions with Parathasarathi Panda and Vaishakh Ravi.
	V. Narayanan was supported by a travel fellowship from the Sarojini Damodaran Foundation.
	V. Narayanan, J. Ravi, and V. Prabhakaran acknowledge support of the Department of Atomic Energy, Government of India, under project no. 12-R\&D-TFR-5.01-0500.
	This work was done while J. Ravi was at Tata Institute of Fundamental Research and IIT Bombay. 
	J.~Ravi has received funding from European Research Council (ERC) under the European Union's Horizon 2020 research and innovation programme (Grant No. 714161).

\end{document}